\documentclass[11pt,twoside]{article}
\pdfoutput=1

\usepackage{a4}

\usepackage[utf8]{inputenc}
\usepackage{times}
\usepackage{microtype}
\usepackage{ziffer}

\usepackage{cite}

\usepackage{amsfonts}
\usepackage{amsmath}
\usepackage{amssymb}
\usepackage{amsthm}
\usepackage{amscd}

\usepackage{eucal}
\usepackage{bm}

\usepackage{booktabs}

\usepackage{graphicx}
\usepackage{float}

\input{macro.lib}

\usepackage[nottoc]{tocbibind}
\usepackage{tocloft}

\advance\cftsecnumwidth 0.5em\relax
\advance\cftsubsecindent 0.5em\relax
\advance\cftsubsecnumwidth 0.5em\relax

\allowdisplaybreaks

\newcommand{\mysection}{\setcounter{equation}{0}\section}

\renewcommand\+{+}

\begin{document}

\thispagestyle{empty}

\begin{center}
{\LARGE \bf Introduction to Solid State Physics}\\[7ex]
{\large \sc Frank Göhmann}\\[2ex]
Fakult\"at f\"ur Mathematik und Naturwissenschaften\\
Bergische Universit\"at Wuppertal
\end{center}

\clearpage
\section*{Preface}
This script is based on lecture notes prepared for the
regular Introduction to Theoretical Solid State Physics
at the University of Wuppertal held by the author in
the winter semesters of 2003/04, 2004/05, 2010/11,
2011/12, 2013/14, 2014/15 and in the summer semesters
of 2006 and 2020. Due to the prevailing Covid 19 pandemic
all teaching at the University of Wuppertal in the summer of
2020 went online. In order to support my students with
their home training programme I decided to typeset at least
the present part of my lecture notes. In regular semesters
I would have delivered 28-30 lectures, 90 minutes each.
Since the beginning of our semester was delayed due to
the pandemic, the number of lectures was restricted to 23.
For this reason I cut out two lectures on the Hartree-Fock
approximation and two lectures usually devoted to recall the
formalism of the second quantization. I condensed the
remaining material to fit into the available 23 time slots.
The missed out material as well as some of the material of
the regular lectures on Advanced Theoretical Solid State
Physics may follow on a later occasion, e.g., after the
next pandemic. Most of the lectures are complemented with
a few intermediate level homework exercises which are
considered an integral part of the course. These exercises
were discussed by the participants in a separate online
exercise session on a weekly basis. 

At the university of Wuppertal the introduction to theoretical
solid state physics is part of the Master course. Students
who attend this course are expected to have successfully
passed the basic courses in theoretical physics (Mechanics,
Electrodynamics, Quantum Mechanics, Statistical Mechanics)
and a course on Advanced Quantum Mechanics.

I would like to thank all those participants of my lectures
whose attention and whose questions helped to improve this
manuscript. Particular thanks are due to Saskia Faulmann and
Siegfried Spruck who read the entire text and pointed out a
number of typos and inaccuracies to me. When I started
teaching Theoretical Solid State Physics in 2003, I devised 
my lecture after lecture notes of my dear colleague Holger
Fehske whom I would like to thank at this point. My own
first lecture notes then gradually evolved over the years and
most probably will continue to evolve in future. I prefer
to think of the following typeset version as of a snapshot
taken in the year of the pandemic 2020.

\subsubsection*{What is solid state physics?}
\begin{itemize}
\item
Application of what we have learned heretofore (QM + StatMech)
to the description of `condensed matter', no new fundamental
theory.
\item
Arguably the most important branch of physics as far as the
daily life of non-physicists is concerned, since many 
important technologies owe their existence our knowledge of
solid state physics. Examples are the semiconductor technology,
lasers, magnetic and charge based information storage devices.
\item
Solid state physics is intellectually most challenging with
new questions arising from experiments year by year. Solids
show us the universe in a nutshell. Their behaviour reaches
from single-particle to highly collective. All tools of
modern theoretical physics are needed for at least an
approximate understanding, in particular, QFT, Feynman graphs,
high-performance computing, and non-perturbative many-body
techniques.
\end{itemize}
\subsubsection*{What are the goals of this lecture?}
\begin{itemize}
\item
Introduction to the basic concepts, meaning that the
emphasis is, in the first instance, on the single-particle
aspects.
\item
Service for Experimental Solid State Physics.
\item
Emphasis on the explanation of concepts and basic ideas,
not always quantitative, justification of the use of
simplified `model Hamiltonians'.
\item
Raise some understanding why many-body physics is mostly
phenomenology.
\item
Convey the following main idea: (Collective) elementary
excitations are `quasi-particles' characterized by their
dispersion relation $\pv \mapsto \e (\pv)$ and by certain
quantum numbers like spin and charge. The most important
two are `the phonon' (= quantized lattice vibration) and
`the electron' (= quantized charge excitation of the solid,
which has as much to do with the electron of elementary
particle physics as water waves have to do with water).
\end{itemize}

\newpage

\addtocontents{toc}{\protect\thispagestyle{myheadings}}

\tableofcontents

\newpage

\mysection{The Hamiltonian of the solid and its eigenvalue problem}
\subsection{Hamiltonian of the solid}
Solid state physics is quantum mechanics and statistical mechanics
of many ($\sim 10^{23}$) particles at `low energies' (typically
$\sim 1\, {\rm meV}$). Relativistic effects (retardation, spin-orbit
coupling, \dots) can often be neglected in the explanation of
phenomena in solids (apart from the fact the spin is a relativistic
effect). The large number of particles $N_p$ is best taken
care of by considering the thermodynamic limit $N_p \rightarrow \infty$
which typically brings about simplifications of the theory.

In solid state physics not only the electrons composing the solid
but also the ions are regarded as elementary. The only relevant
interaction is the Coulomb interaction. Thus, the Hamiltonian
of the solid is the Hamiltonian of non-relativistic ions and
electrons interacting via Coulomb interaction,
\begin{align} \label{hamsol}
     H & = \sum_{j=1}^L \frac{\|\Pv_j\|^2}{2 M_j}
           + \sum_{1 \le j < k \le L} \frac{Z_j Z_k e^2}{\|\Rv_j - \Rv_k\|}
	   && \text{$\leftarrow$ ions, charge $Z_j e$, mass $M_j$}
	   \notag \\[.5ex] & \qd
           + \sum_{j=1}^N \frac{\|\pv_j\|^2}{2 m}
           + \sum_{1 \le j < k \le N} \frac{e^2}{\|\rv_j - \rv_k\|}
	   && \text{$\leftarrow$ electrons, charge $-e$, mass $m$}
	   \notag \\[.5ex] & \mspace{123.mu}
           - \sum_{j=1}^L \sum_{k=1}^N \frac{Z_j e^2}{\|\Rv_j - \rv_k\|} \epp
	   && \text{$\leftarrow$ ion-electron interaction} \\[2ex] \notag
	   & \qqd \text{$\uparrow$ kinetic energy}
	     \qqd \text{$\uparrow$ Coulomb interaction}
\end{align}
Here $L$ is the number of ions and $N$ is the number of electrons.
The Hamiltonian of the solid is the same as in atomic and
molecular physics. The only difference is in the number of
constituents. A solid is a very large molecule. The numbers
that separate the sub-disciplines are
\begin{itemize}
\item
$L = 1$ atomic physics
\item
$L \sim 100$ molecular physics
\item
$L \sim 10^{23}$ solid state physics
\end{itemize}
There is no sharp distinction between a molecule and a solid. In between
molecules and solids are macro molecules. DNA, for instance, is composed
of about $2 \cdot 10^{11}$ atoms.
\subsection{Natural units}
A first understanding of the length and energy scales in solids
comes from a dimensional analysis that allows us to introduce natural
units. Asking, when the kinetic energy of an electron associated with
a given wavelength is of the same order of magnitude as the Coulomb
energy of two electrons one wavelength apart from each other,
\begin{equation} \label{coulombscale}
     \frac{\hbar^2}{2 m \ell^2} \sim E \sim \frac{e^2}{\ell} \epc
\end{equation}
we recover the typical length scale of atomic physics
\begin{equation}
    \ell \sim \frac{\hbar^2}{m e^2} = a_0 = 0,529 \cdot 10^{- 10}\, {\rm m}
         \sim \frac{1}{20} {\rm nm} \epc
\end{equation}
the Bohr radius. The Coulomb energy of two electrons at this
distance,
\begin{equation}
     E = \frac{e^2}{a_0} = \frac{m e^4}{\hbar^2}
                         = 2 {\rm Ry} = 27,2\, {\rm eV} \epc
\end{equation}
is twice the ionization energy of a hydrogen atom (which is one
Rydberg ($1\, {\rm Ry}$)) and thus a typical atomic binding energy.

Measuring all lengths in units of the Bohr radius (i.e., replacing
$\rv_j/a_0 \rightarrow \rv_j$, $\Rv_j/a_0 \rightarrow \Rv_j$)
and the energy in units of $2\, {\rm Ry}$ (by replacing
$H\bigl/(me^4/\hbar^2) \rightarrow H$) we obtain the Hamiltonian of
the solid in natural units,
\begin{equation} \label{hsolnatunits}
     H = - \2 \sum_{j=1}^L \frac{m}{M_j} \6_{R_j^\a}^2 + H_{\rm el} (\Rv) \epc
\end{equation}
where
\begin{multline}
     H_{\rm el} (\Rv) = \sum_{1 \le j < k \le L} \frac{Z_j Z_k}{\|\Rv_j - \Rv_k\|} \\[1ex]
                    - \2 \sum_{j=1}^N \6_{\rv_j^\a}^2
                    + \sum_{1 \le j < k \le N} \frac{1}{\|\rv_j - \rv_k\|}
		    - \sum_{j=1}^L \sum_{k=1}^N \frac{Z_j}{\|\Rv_j - \rv_k\|} \epp
\end{multline}
Here we have used the summation convention with respect to Greek indices
and the further convention that $\Rv = (\Rv_1, \dots, \Rv_L)$. Similarly,
we shall write $\rv = (\rv_1, \dots, \rv_N)$.

Our scale analysis shows us that the ionisation numbers $Z_j$ and
the mass ratios $m/M_j$ are the only parameters of the system. The variation
of these pure numbers is responsible for the rich phenomenology of molecules
and solids and, in fact, of the world around us as we perceive it with
our senses.

Note that the mass ratio
\begin{equation}
     \frac{m}{M_j} \sim 10^{-4}
\end{equation}
in (\ref{hsolnatunits}) is a small parameter. This fact turns
out to be of fundamental importance for the theory of solids
and determines much of the structure of the world around us.
\subsection{An eigenvalue problem for the ions}
If we naively send all masses $M_j$ in (\ref{hsolnatunits}) to
infinity, the kinetic energy of the ions goes to zero, the ions stop
moving. The corresponding ionic parts of the eigenfunctions
separate multiplicatively and become products of delta functions
of the form $\de (\Rv_j - \Rv^{(0)}_j)$. If the masses are large
but finite, the ions will still move, but slower than the
electrons. Their wavefunctions will not be delta functions, but
typically more localized than those of the electrons.

Our favorite classical example system of interacting point particles
of very different masses is the planetary system. The ratio of the
earth mass $m$ to the sun mass $M$, for instance, is about
$m/M = 1/3 \times 10^{-5}$. Earth and sun exert equal but oppositely
directed forces $\pm \Fv$ onto each other, $M \ddot \Xv = \Fv = -
m \ddot \xv$, if $\Xv$ and $\xv$ are the position vectors of sun
and earth, respectively. This means that sun experiences a much
smaller average acceleration than the earth. Consequentially, as
compared to the sun, the earth moves much faster and has a much
larger orbit around the center of mass of the sun-earth system.
In this case, as the mass ratio is so small, the center of mass
lies inside the sun. Hence, to a very good approximation, the
earth moves around the sun and follows it along its way through
the universe.

In a similar way we expect the electrons in a solid to follow the
slower motion of the more massive ions. Translated into the language
of quantum mechanics we expect that the joint motion of electrons
and ions can be approximately described by a product of an ionic
wave function times an electronic wave function calculated for
fixed positions of the ions. The latter would be interpreted as
a conditional probability amplitude for the electrons given the
positions of the ions. The product structure would mimic the
fact in probability theory that the joint probability $p(A \cap B)$
of two events $A$ and $B$ (corresponding to the wave function of
electrons and ions) is equal to $p(A|B) p(B)$, where $p(B)$ is
the probability of $B$ (corresponding to the ions) and $p(A|B)$
is the conditional probability of $A$ given $B$ (corresponding
to the electronic wave function for fixed positions of the ions).

We shall try to work out this idea more formally. Let us start with
the `electronic eigenvalue problem'
\begin{equation} \label{elevap}
     H_{\rm el} (\Rv) \ph(\rv|\Rv) = \e (\Rv) \ph(\rv|\Rv)
\end{equation}
which depends parametrically on the positions $\Rv$ of all ions.
For every $\Rv$ the eigenstates $\ph_n (\rv|\Rv)\bigr|_{n \in {\mathbb N}}$
of $H_{\rm el} (\Rv)$ corresponding to the eigenvalues $\e_n (\Rv)$
form a basis of the electronic Hilbert space. Hence, every solution
$\Ps (\rv, \Rv)$ of the full eigenvalue problem $H \Ps = E \Ps$ can
be expanded in terms of the $\ph_n$,
\begin{equation} \label{boexpansion}
     \Ps(\rv, \Rv) = \sum_{n \in {\mathbb N}} \ph_n (\rv|\Rv) \Ph_n (\Rv) \epp
\end{equation}

Assuming the $\ph_n$ to be known we want to derive an eigenvalue
problem for the $\Ph_n$ which will later be interpreted as the
ionic wave functions. For this purpose we insert (\ref{boexpansion})
into the full eigenvalue problem and write the result as
\begin{multline}
     (H - E) \Ps (\rv, \Rv) =
        \sum_{m \in {\mathbb N}} \biggl\{\ph_m (\rv|\Rv)
	   \biggl[- \2 \sum_{j=1}^L \frac{m}{M_j} \6_{R_j^\a}^2
	          + \e_m (\Rv) - E\biggr] \Ph_m (\Rv) \\[.5ex]
	   - \2 \sum_{j=1}^L \frac{m}{M_j} \biggl[
	   2 \bigl(\6_{R_j^\a} \ph_m\bigr) (\rv|\Rv) \6_{R_j^\a} +
	     \bigl(\6_{R_j^\a}^2 \ph_m\bigr) (\rv|\Rv)\biggr] \Ph_m (\Rv)\biggr\} = 0 \epp
\end{multline}
From here we obtain an equation for the $\Ph_n$ upon multiplication
by $\ph_n^* (\rv|\Rv)$ and integration over $\rv$,
\begin{equation} \label{evapphn1}
     \biggl[- \2 \sum_{j=1}^L \frac{m}{M_j} \6_{R_j^\a}^2
            + \e_n (\Rv) - E\biggr] \Ph_n (\Rv)
	    = \sum_{m \in {\mathbb N}} C_{nm} (\Rv) \Ph_m (\Rv) \epc
\end{equation}
where
\begin{subequations}
\begin{align}
     C_{nm} (\Rv) & = A_{nm} (\Rv) + B_{nm} (\Rv) \epc \\[1ex]
     A_{nm} (\Rv) & = \sum_{j=1}^L \frac{m}{M_j}
                         \int \rd^{3N} r \: \ph_n^* (\rv|\Rv) 
                         \bigl(\6_{R_j^\a} \ph_m\bigr) (\rv|\Rv) \6_{R_j^\a} \epc \\
     B_{nm} (\Rv) & = \2 \sum_{j=1}^L \frac{m}{M_j}
                         \int \rd^{3N} r \: \ph_n^* (\rv|\Rv) 
                         \bigl(\6_{R_j^\a}^2 \ph_m\bigr) (\rv|\Rv) \epp
\end{align}
\end{subequations}
If there is no external magnetic field we may assume that the
$\ph_n$ are real. Then, using that
\begin{equation}
     \int \rd^{3N} r \: \ph_n (\rv|\Rv) 
        \bigl(\6_{R_j^\a} \ph_n \bigr) (\rv|\Rv) =
     \2 \6_{R_j^\a} \int \rd^{3N} r \: \ph_n^2 (\rv|\Rv) = 0 \epc
\end{equation}
we see that $A_{nn} (\Rv) = 0$. Introducing the notation
\begin{equation}
     T = - \2 \sum_{j=1}^L \frac{m}{M_j} \6_{R_j^\a}^2
\end{equation}
we can therefore rewrite (\ref{evapphn1}) in the form
\begin{equation} \label{evapphn2}
     \biggl[T + \e_n (\Rv) - B_{nn} (\Rv) - E\biggr] \Ph_n (\Rv)
	    = \sum_{m \in {\mathbb N}, m \ne n} C_{nm} (\Rv) \Ph_m (\Rv) \epp
\end{equation}
This equation still contains the full information about the ion
system. The ions are now coupled through their Coulomb interaction
(contained in $\e_n (\Rv)$) and through the electrons, whose degrees
of freedom have been formally integrated out. Equation (\ref{evapphn2})
will turn out to be an appropriate starting point for a perturbative
analysis of the ion system with $m/M_j$ taken as a small parameter.

\mysection{Born-Oppenheimer approximation}
\label{lect:bo}
\subsection{More scaling arguments}
Classically we have a general idea, how the time scale, e.g.\ for the
motion of a particle of mass $m$ with one degree of freedom
in a potential $V$, depends on the mass. The Lagrangian of such
a system is
\begin{equation}
     L = \frac{M (\6_t x)^2}{2} - V(x)
       = \frac{m}{2} \Bigl(\6_\frac{t}{\sqrt{M/m}} x\Bigr)^2 - V(x) \epp
\end{equation}
Denoting by $x_\m (t)$ a trajectory of a particle with mass $\m$,
we see that
\begin{equation}
     x_M (t) = x_m \biggl(\frac{t}{\sqrt{M/m}}\biggr)
\end{equation}
is a trajectory of a particle of mass $M$. The motion slows
down if $M > m$. Heavier particles (of the same energy) move
slower. In particular, oscillators oscillate with lower frequency
if their mass is enlarged.

In quantum mechanics we have to consider the stationary Schrödinger
equation, a time-independent problem. Hence, we are rather interested
in how the spatial behaviour of the eigenfunctions varies with mass.
For the bounded motion around an equilibrium position described
by a quadratic minimum of the potential we may recourse to the
harmonic oscillator
\begin{equation} \label{hoham}
     H = \frac{p^2}{2M} + \frac{K x^2}{2} \epp
\end{equation}
Comparing kinetic and potential energy in a similar scaling argument
as in (\ref{coulombscale}),
\begin{equation}
     \frac{\hbar^2}{2 m \ell^2} \sim E \sim \frac{K \ell^2}{2} \epc
\end{equation}
we find the intrinsic length scale
\begin{equation} \label{defell}
     \ell = \biggl(\frac{\hbar^2}{m K}\biggr)^\frac{1}{4} \epp
\end{equation}
Comparing the extension $L$ of an eigenfunction of a heavy particle
of mass $M$ with the extension $\ell$ of a lighter particle of
mass $m$ we obtain a ratio of
\begin{equation}
     \frac{L}{\ell} = \biggl(\frac{m}{M}\biggr)^\frac{1}{4} \epp
\end{equation}
This means that the larger the mass the more localized becomes the
wave function. On the other hand, in a more localized wave function
the particle is closer to the origin and the harmonic approximation
is better justified.
\subsection{Exercise 1. The heavy harmonic oscillator}
The same conclusions as above can be drawn from the solution of the
eigenvalue problem of the harmonic oscillator (\ref{hoham}).
\begin{enumerate}
\item
Determine the full width at half height of the ground state
wave function of the harmonic oscillator. How does it depend on
the mass of the oscillator?
\item
Consider the oscillator with a quartic correction term
\begin{equation}
     H = \frac{p^2}{2m} + \frac{K x^2}{2} + \frac{K' x^4}{4!} \epc
\end{equation}
where $K' > 0$. Show that the correction can be taken into
account perturbatively if the mass is large. For this purpose
use $\ell$ as defined in (\ref{defell}) as a small parameter.
Calculate the correction to the ground state energy in first
oder perturbation theory. How does it depend on the mass of
the oscillator?
\end{enumerate}
\subsection{Application to the ionic motion and adiabatic decoupling}
\label{sec:adiadecoupling}
Let $M$ be a typical ion mass, e.g.\ the arithmetic average of
all ion masses $M = \bigl\< \{M_j\}\bigr\>$. Then
\begin{equation}
     \k = \biggl(\frac{m}{M}\biggr)^\frac{1}{4}
\end{equation}
is an intrinsic length parameter for the bounded motion of 
the ions.

We expect that the eigenvalue problem (\ref{evapphn2}) has solutions
which, on the scale of the electronic wave functions, are strongly
localized around certain equilibrium positions $\Rv^{(0)}$.
To take account of this expectation we introduce new coordinates
$\uv$ on the scale of the electronic wave functions and relative
to this equilibrium position, setting
\begin{equation} \label{coordinateselscale}
     \Rv = \Rv^{(0)} + \k \uv \epp
\end{equation}
For the wave functions we shall write
\begin{equation}
     \widetilde \Ph_n (\uv) = \Ph_n (\Rv) \epp
\end{equation}
Using the new coordinates (\ref{coordinateselscale}) we can make the
$\k$ dependence of the operators in (\ref{evapphn2}) explicit. For
this purpose we define the rescaled operators
\begin{subequations}
\begin{align} \label{deftu}
     T_\uv & = - \2 \sum_{j=1}^L \frac{M}{M_j} \6_{u_j^\a}^2 \epc \\
     \widetilde A_{nm} (\Rv) &
        = \sum_{j=1}^L \frac{M}{M_j}
	     \int \rd^{3N} r \: \ph_n^* (\rv|\Rv) 
	        \bigl(\6_{R_j^\a} \ph_m\bigr) (\rv|\Rv) \6_{u_j^\a} \epc \\
     \widetilde B_{nm} (\Rv) &
        = \2 \sum_{j=1}^L \frac{M}{M_j}
	     \int \rd^{3N} r \: \ph_n^* (\rv|\Rv) 
	        \bigl(\6_{R_j^\a}^2 \ph_m\bigr) (\rv|\Rv) \epc
\end{align}
\end{subequations}
which remain finite for $\k \rightarrow 0$. With these definitions the
eigenvalue problem (\ref{evapphn2}) assumes the form
\begin{multline} \label{evapphn3}
     \biggl[\k^2 T_\uv + \e_n (\Rv^{(0)} + \k \uv)
              - \k^4 \widetilde B_{nn} (\Rv^{(0)} + \k \uv) - E\biggr]
	        \widetilde \Ph_n (\uv) \\[1ex]
	    = \k^3 \sum_{m \in {\mathbb N}, m \ne n}
	                 \bigl(\widetilde A_{nm} (\Rv^{(0)} + \k \uv)
	                 + \k \widetilde B_{nm} (\Rv^{(0)} + \k \uv) \bigl)
			   \widetilde \Ph_m (\uv) \epp
\end{multline}

Recall that we assume that the wave functions $\Ph_n$ are strongly
localized around $\Rv^{(0)}$, implying that the redefined functions
$\widetilde \Ph_n$ are strongly localized around $\uv = 0$. For this
reason it makes sense to expand the operators in (\ref{evapphn3})
which act on the functions $\widetilde \Ph_n$ in a Taylor series
in $\k \uv$ and to solve the resulting eigenvalue problem perturbatively
for small~$\k$. The latter means to look for solutions in the form
of formal series in $\k$,
\begin{subequations}
\begin{align}
     E & = E^{(0)} + \k E^{(1)} + \dots \epc \\[1ex]
     \widetilde \Ph_n (\uv) & = \widetilde \Ph_n^{(0)} (\uv)
                              + \k \widetilde \Ph_n^{(1)} (\uv) + \dots
\end{align}
\end{subequations}
Inserting these perturbation series into (\ref{evapphn3}) and performing
the Taylor expansion we obtain
\begin{multline}
     \big(\k^2 T_\uv + \e_n (\Rv^{(0)}) + \k \e_n^{(1)} (\uv)
          + \k^2 \e_n^{(2)} (\uv) + \k^3 \e_n^{(3)} (\uv)
	  + \k^4 \e_n^{(4)} (\uv) + \\[.5ex] \mspace{54.mu}
	  \dots - \k^4 \widetilde B_{nn} (\Rv^{(0)})
	  - \dots - E^{(0)} - \k E^{(1)} - \k^2 E^{(2)} - \k^3 E^{(3)}
	  - \k^4 E^{(4)} - \dots \bigr) \\[1ex] \mspace{216.mu} \times
     \bigl(\widetilde \Ph_n^{(0)} (\uv) + \k \widetilde \Ph_n^{(1)} (\uv)
           + \k^2 \widetilde \Ph_n^{(2)} (\uv)
           + \k^3 \widetilde \Ph_n^{(3)} (\uv)
	   + \dots \bigr) \\[1ex]
     = \k^3 \sum_{m \in {\mathbb N}, m \ne n}
                  \bigl(\widetilde A_{nm} (\Rv^{(0)}) + \dots \bigl)
       \bigl(\widetilde \Ph_m^{(0)} (\uv) + \k \widetilde \Ph_m^{(1)} (\uv)
             + \k^2 \widetilde \Ph_m^{(2)} (\uv) + \dots \bigr) \epp
\end{multline}

Here we compare the coefficients in front of the powers of $\k$ order
by order. To the order $\k^0$ we obtain
\begin{equation}
     (\e_n (\Rv^{(0)}) - E^{(0)}) \widetilde \Ph_n^{(0)} (\uv) = 0 \epp
\end{equation}
Let $\Ph_\ell^{(0)} (\uv) \ne 0$ for some $\ell \in {\mathbb N}$,
$\e_\ell (\Rv^{(0)})$ non-degenerate. Then
\begin{equation}
     E^{(0)} = \e_\ell (\Rv^{(0)}) \epc \qd
     \widetilde \Ph_n^{(0)} (\uv) = 0 \qd \forall n \ne \ell \epp
\end{equation}
For $n \ne \ell$ we conclude for the higher orders in $\k$ that
\begin{subequations}
\begin{align}
     {\cal O} (\k): && \bigl(\e_n (\Rv^{(0)}) - \e_\ell (\Rv^{(0)})\bigr)
                       \widetilde \Ph_n^{(1)} (\uv) & = 0 
		       \qd \then\ \widetilde \Ph_n^{(1)} (\uv) = 0 \epc \\[1ex]
     {\cal O} (\k^2): && \bigl(\e_n (\Rv^{(0)}) - \e_\ell (\Rv^{(0)})\bigr)
                         \widetilde \Ph_n^{(2)} (\uv) & = 0
		         \qd \then\ \widetilde \Ph_n^{(2)} (\uv) = 0 \epc \\[1ex]
     {\cal O} (\k^3): && \bigl(\e_n (\Rv^{(0)}) - \e_\ell (\Rv^{(0)})\bigr)
                         \widetilde \Ph_n^{(3)} (\uv)
                            & = \widetilde A_{n \ell} (\Rv^{(0)})
			      \widetilde \Ph_\ell^{(0)} (\uv) \notag \\[1ex]
     && \then\ \widetilde \Ph_n^{(3)} (\uv) & =
                 \frac{\widetilde A_{n \ell} (\Rv^{(0)}) \widetilde \Ph_\ell^{(0)} (\uv)}
		      {\e_n (\Rv^{(0)}) - \e_\ell (\Rv^{(0)})} \epp
\end{align}
\end{subequations}
The latter equation means that at order $\k^3 \sim 1/1000$ the
wave function $\Ps (\rv|\Rv)$ (\ref{boexpansion}) of the coupled
electron ion system ceases to be a simple product of two factors.

For $n = \ell$ we find at order $\k$ that
\begin{equation}
     \bigl(\e^{(1)} (\uv) - E^{(1)}\bigr) \widetilde \Ph_\ell^{(0)} (\uv) = 0 \qd
     \then\ E^{(1)} = \e_\ell^{(1)} (\uv)
        = \sum_{j=1}^L \frac{\6 \e_\ell (\Rv^{(0)})}{\6 R_j^\a} \uv_j^\a \epp
\end{equation}
But $E^{(1)}$ must be independent of $\uv$ which can only hold if
\begin{equation}
     \frac{\6 \e_\ell (\Rv^{(0)})}{\6 R_j^\a} = 0 \epp
\end{equation}
Then, necessarily,
\begin{equation}
     E^{(1)} = 0 \epp
\end{equation}
For higher orders of $\k$ we remain with the equation
\begin{align}
     \big(T_\uv + & \e_\ell^{(2)} (\uv) + \k \e_\ell^{(3)} (\uv)
	  + \k^2 \e_\ell^{(4)} (\uv) - \k^2 \widetilde B_{nn} (\Rv^{(0)})
	  - E^{(2)} - \k E^{(3)} - \k^2 E^{(4)}\bigr) \notag
	     \\[1ex] & \mspace{288.mu} \times
     \bigl(\widetilde \Ph_\ell^{(0)} (\uv) + \k \widetilde \Ph_\ell^{(1)} (\uv)
           + \k^2 \widetilde \Ph_\ell^{(2)} (\uv) \bigr) \notag \\[1ex]
     & = \k^4 \sum_{m \in {\mathbb N}, m \ne \ell}
                  \widetilde A_{\ell m} (\Rv^{(0)}) \widetilde \Ph_m^{(3)} (\uv) 
		  + {\cal O} (\k^5) \notag \\[1ex]
     & = \k^4 \sum_{m \in {\mathbb N}, m \ne \ell}
                 \frac{\widetilde A_{\ell m} (\Rv^{(0)})
		       \widetilde A_{m \ell} (\Rv^{(0)}) \widetilde \Ph_\ell^{(0)} (\uv)}
		      {\e_m (\Rv^{(0)}) - \e_\ell (\Rv^{(0)})} + {\cal O} (\k^5) \epp
\end{align}
Conceiving this as an equation up to the order $\k^2$, we see that the
right hand side can be consistently neglected.
\subsection{Summary and interpretation}
\begin{enumerate}
\item
Within the scheme of the above perturbation theory the solutions
$\Ph_{n, \ell}$, $E_{n, \ell}$ of the eigenvalue problem
\begin{equation} \label{boevap}
     \bigl(T + \e_n (\Rv) - B_{nn} (\Rv) - E\bigr) \Ph (\Rv) = 0
\end{equation}
consistently determine the eigenfunctions of the solid in product form
\begin{equation} \label{bowavefun}
     \Ps_{n, \ell} (\rv, \Rv) = \ph_n (\rv|\Rv) \Ph_{n, \ell} (\Rv)
\end{equation}
up to the fourth order expansion of $\e_n (\Rv)$ in $\k$ and
up to the zeroth order expansion of $B_{nn} (\Rv)$ in $\k$.
$\ph_n (\rv|\Rv)$ is interpreted as a conditional probability
amplitude and $\Ph_{n, \ell} (\Rv)$ as an ionic wave function. The
product form then means that the electrons follow the motion
of the ions.
\item
The approximation (\ref{boevap}), (\ref{bowavefun}) is called
the Born-Oppenheimer approximation and originated in molecular
physics \cite{BoOp27,BoHu54}. Another common name is `the adiabatic
approximation'.
\item
Within the Born-Oppenheimer approximation the equilibrium positions
of the ions are determined by the condition
\begin{equation}
     \frac{\6 \e_n (\Rv^{(0)})}{\6 R_j^\a} = 0 \epp
\end{equation}
In solids the ions arrange themselves in regular lattices.
\item
Consistently up to the second order in $\k$ the eigenvalue
problem (\ref{boevap}) takes the form
\begin{equation} \label{lathamharmonic}
     \bigl(T_\uv + \e_n^{(2)} (\uv) - E\bigr) \widetilde \Ph^{(0)} (\uv) = 0
\end{equation}
where
\begin{equation} \label{harmpot}
     \e_n^{(2)} (\uv) = \2 \sum_{j,k=1}^L
                        \frac{\6^2 \e_n (\Rv^{(0)})}{\6 R_j^\a \6 R_k^\be}
			u_j^\a u_k^\be \epp
\end{equation}
This is called the harmonic approximation. Since only
$\widetilde \Ph^{(0)}$ is taken into account in the harmonic
approximation, it is consistent to consider the electronic
wave function only to lowest order $\ph (\rv|\Rv^{(0)})$ as
well. In this approximation the electrons move in a static
lattice determined by the equilibrium positions of the ions
(like the earth was approximately moving around a static
sun in our entrance example).
\item
Notice that, in spite of the ratio of electron to ion mass
being very small, our actual expansion parameter $\k \sim 1/10$
is only moderately small.
\item
Our perturbative analysis becomes questionable, if the electronic
levels are degenerate or close to degenerate (it is well justified
for the electronic ground state of an insulator, but problematic
for a metal).
\item
Many properties of solids can be understood qualitatively and
often also quantitatively within the harmonic approximation.
Most of the treatment of solids in this lecture will be based
on it. It explains e.g.\ the scattering of light or neutrons,
the propagation of sound and the specific heat.
\item
An example of an effect which cannot be explained within the
harmonic approximation is the thermal extension of a solid. It
is a higher order effect and therefore small. Still it can
be understood within the adiabatic approximation if we proceed
to the fourth order expansion of $\e_n (\Rv)$.
\end{enumerate}
\subsection{Exercise 2. A simple application of Born-Oppenheimer}
A heavy particle ($M$) and a light particle ($m \ll M$) move inside an infinitely
high potential well of width $L$. The particles experience an attractive interaction
described by the potential $W(r-R) = -\lambda \cdot \delta (r-R)$, where $\lambda > 0$
and $R$ and $r$ are the positions of the heavy and light particle, respectively.

Calculate the spectrum of the corresponding Hamiltonian
\begin{equation}
  H = - \frac{\hbar^2}{2m} \frac{\partial^2}{\partial r^2}
      - \frac{\hbar^2}{2M} \frac{\partial^2}{\partial R^2}
      - \lambda \cdot \delta (r-R)
\end{equation}
in  Born-Oppenheimer approximation:
\begin{enumerate}
\item
In step one transform the Hamiltonian into a dimensionless form and
neglect the kinetic energy of the heavy particle. Determine the (unnormalized)
eigenfunctions and an equation for the energy eigenvalues $\varepsilon (R)$ of
the light particle. It may turn out to be useful to distinguish the cases
$\varepsilon<0$, $\varepsilon=0$ and $\varepsilon>0$. Recall the relation
\begin{equation}
  \frac{\partial \varphi}{\partial r}(R+0) -\frac{\partial \varphi}{\partial r}(R-0) =
  - \frac{2m\lambda}{\hbar^2} \varphi(R) \epc
\end{equation}
for the jump of the derivative of the wave function caused by the $\delta$-function.
\item
In step two sketch the energy $\varepsilon(R)$ of the light particle for the
ground state and for the first excited state as a function of the position of
the heavy particle. What is the meaning of $\varepsilon(R)$ and how should we
include (qualitatively) the energy levels of the heavy particle into our picture?
\end{enumerate}
\mysection{Crystal lattices}
Within the adiabatic approximation the electrons and ions form bound states,
molecules or solids, in which the ions oscillate around the minima
$\Rv^{(0)}$ of the effective potentials $\e_n (\Rv)$ determined by the
mutual Coulomb interaction of the ions and the energy eigenvalues of
the electrons for fixed ion positions. We derived this statement in the
previous lecture presupposing that such minima exist. The argument would have
been more convincing if we would have been able to prove the existence
of minima starting with the Hamiltonian (\ref{hamsol}). Such undertaking
seems to be out of reach with our current methods. Still, in nature, atoms
always form bound states at low enough temperatures (the only notable exception
being helium which condenses into a `super fluid' before solidification).
If many atoms are put together in stoichiometric ratios they form periodic
structures called crystals. Crystals are formed, since (i) two atoms prefer
a certain binding length and (ii) no direction is preferred in the large
(this is very much like packing balls into a box and carefully shaking it
such that the balls find their equilibrium positions).

In this and in the following lecture we introduce the terminology and
certain mathematical structures needed for the description of crystal
lattices.
\subsection{The Bravais lattice}
The most important structural feature of a crystal is that it can be thought
of as being generated by periodic repetitions of a finite elementary structure,
the unit cell, in three space directions.

Let $\av_1, \av_2, \av_3 \in {\mathbb R}^3$, $\<\av_1, \av_2 \times \av_3\> =
\det (\av_1, \av_2, \av_3) \ne 0$. Then
\begin{equation} \label{bravaislattice}
     B = \bigl\{\xv \in {\mathbb R}^3 \big| \xv = \ell_1 \av_1 + \ell_2 \av_2
                + \ell_3 \av_3, \ell_j \in {\mathbb Z} \bigr\} \epc
\end{equation}
conceived as an Abelian group, is called a (3d) Bravais lattice. $B$ is
called the Bravais lattice of a crystal, if $B$ is the group of all
translations which map the (infinitely extended) crystal onto itself.
Similarly we can define Bravais lattices in any number of dimensions.
In the examples and exercises we will frequently work with $d = 1$ and
$d = 2$.

\begin{remark}
We shall develop part of the theory of infinite crystals which is
expected to give a realistic description, if the ratio of the number
of ions at the surface to the number of ions in the bulk is small.
For a typical macroscopic total number of $L \sim 10^{24}$ ions,
of the order of $L^{2/3} \sim 10^{16}$ of them are at the surface,
and the ratio is $\sim 10^{-8}$.
\end{remark}

Bravais lattice vectors $\av_j$, $j = 1, 2, 3$, that generate 
the Bravais lattice as an Abelian group with respect to vector addition
are called primitive (lattice) vectors. They are not unique.
\begin{lemma} \label{lem:primvec}
Let $\av_j$, $j = 1, 2, 3$, a set of primitive vectors of a Bravais
lattice. Then
\[
     \bv_i = \sum_{j=1}^3 m_{ij} \av_j,\ i = 1, 2, 3,\ \text{primitive}\
        \Leftrightarrow\ m_{ij} \in {\mathbb Z}\ \text{and}\ |\det m| = 1 \epp
\]
\end{lemma}
\begin{proof}
$\then$: primitive $\then\ m_{ij}^{-1} \in {\mathbb Z}$ $\then \det m^{-1}
= 1/\det m \in {\mathbb Z}\ \then$ (since $\det m \in {\mathbb Z}$):
$\det m = \pm 1$.

\noindent $\Leftarrow$: Cramer's rule and $\det m = \pm 1\ \then\ m_{ij}^{-1}
\in {\mathbb Z}\ \then$ primitive.
\end{proof}

Since $|\det m| = 1$ it follows that
\begin{equation}
     V_u = |\det (\bv_1, \bv_2, \bv_3)| = |\det (\av_1, \av_2, \av_3)| \epc
\end{equation}
the volume of the parallelepiped spanned by a set of primitive vectors,
is independent of their choice.

By definition, the unit cell of a crystal with Bravais lattice $B$ is a
simply connected finite volume (of size~$V_u$) which covers ${\mathbb R}^3$
through translation by $B$. We often use a special unit cell, the Wigner-Seitz
cell, which reflects the symmetries of the Bravais lattice, and is defined as
\begin{equation}
     W = \bigl\{\xv \in {\mathbb R}^3 \big| \|\xv\| \le \|\xv - \gv\|\
                \forall \gv \in B \setminus \{0\}\bigr\} \epp
\end{equation}
Geometrically this is the set of points in ${\mathbb R}^3$ for which the
closest Bravais lattice point is the origin. It can be constructed by
drawing lines from the origin to the neighbouring sites in the Bravais
lattice and erecting the perpendicular bisectors on these lines.

Exercise: Draw the Wigner-Seitz cells for 2d Bravais lattices composed
of equilateral squares and triangles.
\subsection{The reciprocal lattice}
The reciprocal lattice which we shall define now is a lattice which, in
a sense, is dual to the Bravais lattice. It is one of the most important
notions in solid state physics and will accompany us throughout this
lecture.

Given a Bravais lattice $B$ and a set of primitive vectors $\av_j$,
$j = 1, 2, 3$, generating it, we would like to know how to decompose
any $\xv \in {\mathbb R}^3$ with respect to the $\av_j$, i.e.,
we would like to know its coordinates with respect to the basis
$\{\av_1, \av_2, \av_3\}$. Suppose that $\bv_j$, $j = 1, 2, 3$, exist
such that
\begin{equation} \label{dualbases}
     \<\av_j, \bv_k\> = 2\p \de_{jk} \epp
\end{equation}
Then
\begin{equation}
     2 \p x_j = \<\xv, \bv_j\> \epc \qd \xv = \frac{\<\xv,\bv_j\> \av_j}{2 \p} \epc
\end{equation}
where we have employed the summation convention in the second equation. If
(\ref{dualbases}) is satisfied, then $\{\av_j\}_{j=1}^3$ and
$\{\bv_j\}_{j=1}^3$ are called reciprocal (dual) to each other. The
lattice
\begin{equation} \label{reciprocallattice}
     \overline{B}
        = \bigl\{\xv \in {\mathbb R}^3 \big| \xv = \ell_1 \bv_1 + \ell_2 \bv_2
                 + \ell_3 \bv_3, \ell_j \in {\mathbb Z} \bigr\}
\end{equation}
is called the reciprocal lattice (associated with the Bravais lattice
$B$).

It is easy to solve (\ref{dualbases}) for the $\bv_j$. For this
purpose we rewrite it in matrix form
\begin{equation} \label{dualbasesmf}
     (\av_1, \av_2, \av_3)^t (\bv_1, \bv_2, \bv_3) = 2 \p I_3 \epc
\end{equation}
where $I_3$ is the $3 \times 3$ unit matrix, and use Cramer's rule,
\begin{equation}
     (\bv_1, \bv_2, \bv_3) = 2 \p \bigl((\av_1, \av_2, \av_3)^{-1}\bigr)^t
        = \frac{2\p}{V_u} (\av_2 \times \av_3, \av_3 \times \av_1, \av_1 \times \av_2) \epp
\end{equation}
Without restriction of generality we have assumed here that
$\det(\av_1, \av_2, \av_3) > 0$.

\subsection{Properties of the reciprocal lattice} \label{subsec:propreclat}
\begin{enumerate}
\item
Involutivity. Equation (\ref{dualbases}) implies that $\overline{\overline{B}} = B$.
The reciprocal of the reciprocal lattice is the original Bravais lattice.
\item
Brillouin zone. The Wigner-Seitz cell of the reciprocal lattice is
called the (first) Brillouin zone. It plays an important role in
solid state physics.
\item
Volume of the Brillouin zone. Fix a set of primitive reciprocal
lattice vectors $\bv_j$ and denote the volume of the parallelepiped
spanned by these vectors by
\begin{equation}
     V_R = \det (\bv_1, \bv_2, \bv_3) \epp
\end{equation}
Taking the determinant on the left and right hand side of the first
equation (\ref{dualbasesmf}) we see that
\begin{equation} \label{volbz}
     V_R = \frac{(2\p)^3}{V_u}
\end{equation}
is the volume of the unit cells of the reciprocal lattice, which is
the same as the volume of the Brillouin zone.
\item
Lattice planes. A lattice plane $S \subset {\mathbb R}^3$ associated
with a Bravais lattice $B$ is a plane for which $\exists\ \av_1, \av_2,
\av_3 \in B$ primitive, $\exists\ \ell \in {\mathbb Z}$ such that
\begin{equation}
     \xv_{\ell; m, n} = \ell \av_1 + m \av_2 + n \av_3 \in S \epc \qd
                        \forall\ m, n \in {\mathbb Z} \epp
\end{equation}
Let $\bv_1 = 2 \p (\av_2 \times \av_3)/V_u$ the reciprocal to $\av_1$. Then 
\begin{equation} \label{perpformlp}
     S = \bigl\{\xv \in {\mathbb R}^3 \big| \<\xv, \bv_1\> = 2 \p \ell\bigr\} \epp
\end{equation}
Thus, for every lattice plane there is a primitive reciprocal
lattice vector $\bv_1 \in \overline{B}$ and an $\ell \in {\mathbb Z}$
such that (\ref{perpformlp}) holds. If we vary $\ell$ in
(\ref{perpformlp}) we obtain a family of equidistant lattice planes.
Conversely, given any primitive vector $\bv_1 \in \overline{B}$
and any $\ell \in {\mathbb Z}$, the plane $S$ defined by
(\ref{perpformlp}) is a lattice plane. Thus, families of lattice
planes are in one-to-one correspondence with primitive vectors of
the reciprocal lattice.
\item
Miller indices. Lattice planes play an important role in the spectroscopy
of solids, We shall see in the course of this lecture that waves
impinging on a crystal are reflected as if they were reflected by
families of lattice planes. In spectroscopy the families of lattice
planes are usually labeled by the so-called Miller indices defined
relative to a fixed triple $\{\bv_1, \bv_2, \bv_3\}$ of primitive
vectors of the reciprocal lattice. According to lemma~\ref{lem:primvec}
every primitive vector $\bv \in \overline{B}$ can be uniquely presented
as an integer linear combination
\begin{equation}
     \bv = m_1 \bv_1 + m_2 \bv_2 + m_3 \bv_3 \epc \qd
           m_j \in {\mathbb Z} \epp
\end{equation}
The triple $(m_1, m_2, m_3)$ is then the Miller index of the family
of lattice planes associated with $\bv$. Accordingly, one speaks of
the $(m_1, m_2, m_3)$ plane (e.g., of the $(1,0,0)$ plane). By convention,
a bar is used instead of a minus sign, such that e.g.\ $(1,-1,0) =
(1, \bar 1, 0)$. As a reference triple $\{\bv_1, \bv_2, \bv_3\}$
one usually uses primitive vectors of minimal possible length.
\end{enumerate}

\subsection{Exercise 3. Real space interpretation of the Miller indices}
Let $m_j \ne 0$, $j = 1, 2, 3$, the Miller indices of a family of
lattice planes. Convince yourself that they indicate in which three
points
\begin{equation}
     \xv_j = \frac{\av_j}{m_j} \epc \qd j = 1, 2, 3 \epc
\end{equation}
the lines along the $\av_j$ directions cut the lattice plane with
$\ell = 1$ in (\ref{perpformlp}). Note that in the excluded cases 
with one or two of the $m_j$ being equal to zero there is no intersection
in the corresponding direction. The lines are parallel to the lattice
plane.

\subsection{Lattice periodic functions}
\label{subsec:fourier}
Many quantities that characterize the state of a crystal have the same
periodicity as the crystal itself. For this reason we will often have
to deal with periodic functions whose periods are the primitive vectors
of a Bravais lattice. These are often conveniently described by their Fourier
series.

Let us recall Fourier series in one spatial dimension. In this case the
unit cell is necessarily an interval $[0,a]$, where $a > 0$ is called the
lattice spacing or lattice constant, and the primitive vector is equal 
to $a$. A lattice periodic function $f: {\mathbb R} \rightarrow {\mathbb C}$
is a function satisfying
\begin{equation}
     f(x) = f(x + a) \qd \forall x \in {\mathbb R} \epp
\end{equation}

A natural basis for the expansion of such functions can be constructed as
follows. Let $\ph \in [0, 2\p)$ and
\begin{equation}
     z = \re^{\i \ph} = \re^{\i (\ph + 2\p)}
       = \re^{\i \frac{2 \p}{a} (\frac{\ph a}{2 \p} + a)} \epp
\end{equation}
Setting
\begin{equation}
     x = \frac{\ph a}{2 \p} \epc \qd b = \frac{2 \p}{a} \epc \qd k_m = m b
\end{equation}
we see that
\begin{equation}
     a b = 2 \p
\end{equation}
implying that $b$ generates the reciprocal lattice, and that the functions
\begin{equation}
     z^m = \re^{\i k_m x} \epc \qd m \in {\mathbb Z} \epc
\end{equation}
are linear independent and periodic with period $a$.

If the series
\begin{equation} \label{fseries}
     f(x) = \sum_{m \in {\mathbb Z}} A_{k_m} \re^{\i k_m x}
\end{equation}
converges uniformly, $f$ is periodic with period $a$ and
continuous, and
\begin{equation}
     \frac{1}{a} \int_0^a \rd x \: \re^{- \i k_n x} f(x)
        = \sum_{m \in {\mathbb Z}} A_{k_m} \frac{1}{a}
	              \int_0^a \rd x \: \re^{\i (k_m - k_n)x} = A_{k_n} \epp
\end{equation}
By means of the latter formula we can associate a sequence of Fourier
coefficients $(A_{k_n})_{n \in {\mathbb Z}}$ and a Fourier series
(\ref{fseries}) with every complex valued function $f$ that is
integrable on $[0,a]$. In the following we will understand Fourier
series in such a formal sense. But when we will be dealing with concrete
examples, we will have in mind that the convergence properties of
Fourier series are a delicate matter and that, in general, neither
uniform nor pointwise convergence is guaranteed.

In order to construct Fourier series that have the periodicity
defined by a Bravais lattice~$B$, we fix a set of primitive
vectors $\{\av_1, \av_2, \av_3\}$ and a set of reciprocal
vectors $\{\bv_1, \bv_2, \bv_3\}$ satisfying (\ref{dualbases}).
Then
\begin{equation}
     z_j = \re^{\i \<\bv_j, \xv\>} = \re^{\i \<\bv_j, \xv + \av_k\>} \epc
           \qd j, k = 1, 2, 3 \epc
\end{equation}
and the monomials
\begin{equation}
     z_1^\ell z_2^m z_3^n = \re^{\i \<\ell \bv_1 + m \bv_2 + n \bv_3, \xv\>} \epc
        \qd \ell, m, n \in {\mathbb Z} \epc
\end{equation}
are linear independent and periodic with all $\av \in B$ being
periods. In analogy with the 1d case we may define $\kv_{\ell m n}
= \ell \bv_1 + m \bv_2 + n \bv_3$ and the Fourier series
\begin{equation}
     f(\xv) = \sum_{\ell, m, n \in {\mathbb Z}}
              A_{k_{\ell m n}} \re^{\i \<\kv_{\ell m n}, \xv\>}
            = \sum_{\kv \in \overline{B}} A_\kv \re^{\i \<\kv, \xv\>} \epc
\end{equation}
where
\begin{equation}
     A_\kv = \frac{1}{V_u} \int_{U} \rd^3 x \:
                         \re^{- \i \<\kv, \xv\>} f(\xv)
\end{equation}
and the integral is over a unit cell $U$ of volume $V_u$.

With these remarks on Fourier series we have obtained an interpretation
of the reciprocal lattice. The reciprocal lattice is a lattice in
Fourier space dual to the real space Bravais lattice.
\mysection{Crystal symmetries}
\subsection{The crystal lattice}
We define a (physical) crystal lattice as the set $M$ of the positions
(measured as expectation values) of the ions of a solid in its ground
state. The translational symmetry of the crystal (lattice) is described
by the corresponding Bravais lattice. In general there are several ions
in a unit cell of the Bravais lattice. The positions of the ions inside
a unit cell determine the so-called lattice basis. If every unit cell
contains only one ion, the crystal is called simple. In this case it
can be identified with its Bravais lattice. A crystal that is not simple
is called a crystal with basis.
\subsection{The Euclidean group}
The Euclidean group is the group of all maps ${\mathbb R}^3 \mapsto
{\mathbb R}^3$ which leave the distance between any two arbitrary
points invariant. It consists of all pairs $(A, \av)$ of orthogonal
transformations $A \in O(3)$ and translations $\av \in {\mathbb R}^3$,
\begin{equation}
     (A, \av) \xv = A \xv + \av \epp
\end{equation}
Then
\begin{equation}
     (A, \av) (B, \bv) \xv = (A, \av)(B \xv + \bv)
        = A B \xv + A \bv + \av = (AB, A \bv + \av) \xv \epc
\end{equation}
implying that two group elements are multiplied according to the rule
\begin{equation}
     (A, \av) (B, \bv) = (AB, A \bv + \av) \epp
\end{equation}
\subsection{The symmetry group of a crystal}
The symmetry group $R$ of a crystal $M$ is defined as
\begin{equation}
     R = \bigl\{(A, \av) \in E(3) \big| (A, \av) M = M \bigr\} \epc
\end{equation}
i.e., as the subgroup of $E(3)$ that leaves $M$ invariant. If
we identify $\bv$ with $(\id, \bv)$ for every $\bv \in B$, the
Bravais lattice of $M$, we see that $B$ is a subgroup of $R$
(since it leaves $M$ invariant and is a group).
\begin{remark}
$(\id, \av) \in R$ implies that $\av \in B$, but $(A, \av) \in R$,
$A \ne \id$ does not imply that $\av \in B$. In crystals with basis
smaller translations may occur which are parts of so-called
glide reflections or screw rotations. For an example see
Figure~\ref{fig:glidereflection}.

\begin{figure}
\begin{center}
\includegraphics[width=.45\textwidth]{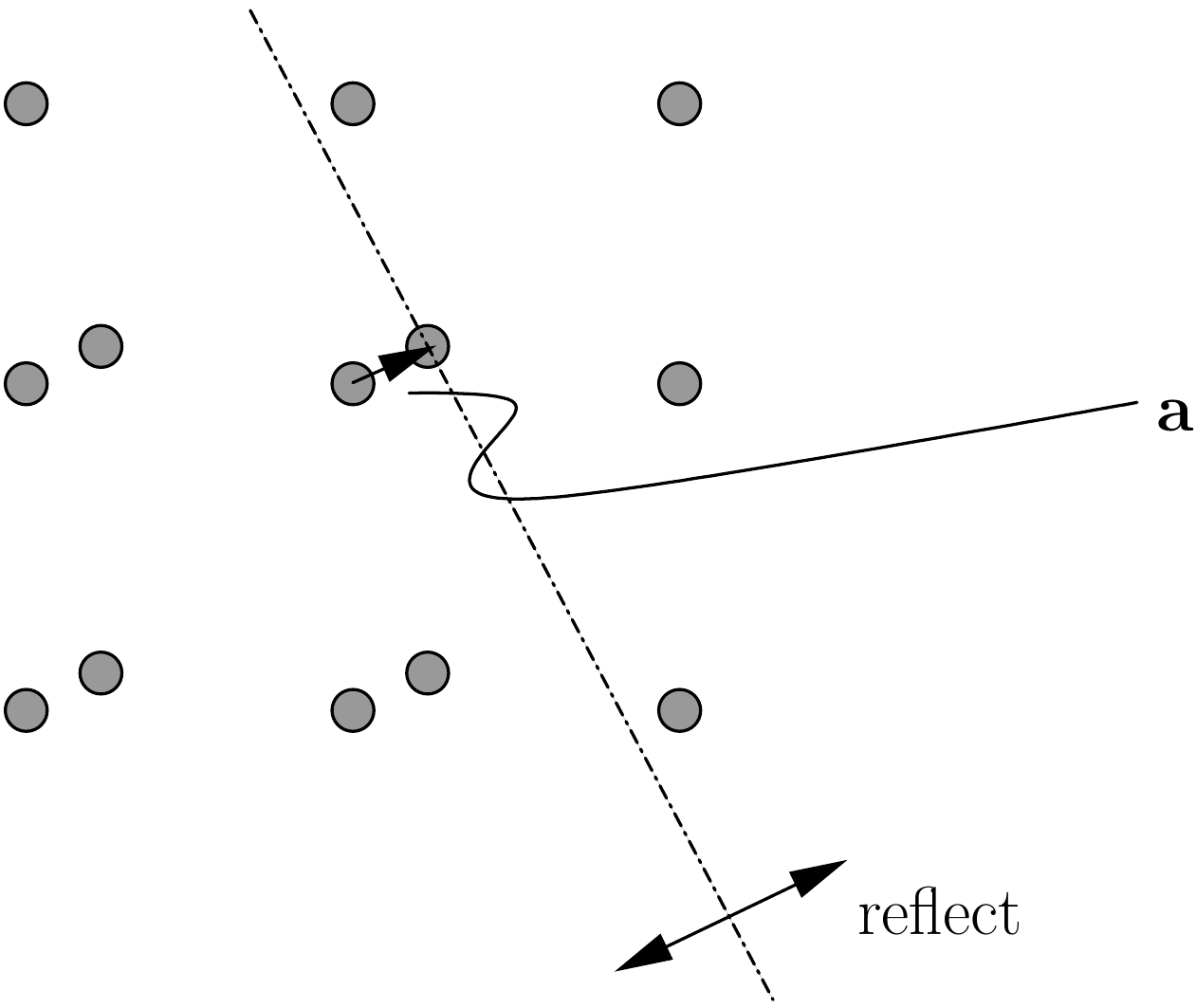}
\caption{\label{fig:glidereflection} 2d example of a symmetry
involving a translation which is not in the Bravais lattice.
A glide reflection: glide by $\av$ and reflect at the dashed
line. Clearly, reflecting without gliding is not a symmetry
of the sketched point configuration, but leaves its Bravais
lattice invariant.}
\end{center}
\end{figure}

\end{remark}
\begin{lemma} \label{lem:bravaisnormal}
Let $R$ be the symmetry group of a crystal with Bravais lattice $B$.
Then $B \subset R$ is a normal subgroup of $R$.
\end{lemma}
\begin{proof}
Let $\bv \in B$, $(A, \av) \in R$. $\then\ (A, \av) (\id, \bv) =
(A, A \bv + \av) = (\id, A \bv) (A, \av)\epc\ \then\ (\id, A \bv) M = M\
\then\ A \bv \in B$. Thus, $(A, \av) (\id, \bv) (A, \av)^{-1} \in
(\id, B)$ for every $b \in B$, meaning that the Bravais lattice
is an invariant (= normal) subgroup of $R$.
\end{proof}
\subsection{The point group of a crystal}
For a crystal $M$ with symmetry group $R$ define
\begin{equation}
     R_0 =\bigl\{A \in O(3) \bigr| \exists\, g \in R\ \text{s.\ th.}\
                 g = (A, \av)\bigr\}
\end{equation}
the set of all `$O(3)$ parts' of $R$.
\begin{lemma}
$R_0$ is a subgroup of $O(3)$, the so-called point group of the
crystal.
\end{lemma}
\begin{proof}
(i)\ $A, B \in R_0\ \then\ \exists\, \av, \bv \in {\mathbb R}^3\
\text{s.\ th.}\ (A, \av), (B, \bv) \in R\ \then\
(AB, A \bv + \av) \in R\ \then\ AB \in R_0$. (ii) $(A, \av)^{-1}
= (A^{-1}, - A^{-1} \av) \in R\ \then\ A^{-1} \in R_0$. (iii)
$\id \in R_0$. (i)-(iii) $\then\ R_0$ is a group.
\end{proof}
\begin{remark}
In general, $R_0 M \ne M$, i.e.\, the point group of $M$ does not
necessarily leave $M$ invariant.
\end{remark}
We have seen in the proof of lemma~\ref{lem:bravaisnormal} that
$\bv \in B, (A, \av) \in R\ \then\ A \bv \in B$. Since, $A \in R_0$,
it follows that $R_0 B \subset B$, $B$ is invariant under $R_0$.
This has two important implications:
\begin{enumerate}
\item
The set of all point groups must be restricted, since not all
subgroups of $O(3)$ can leave a Bravais lattice invariant.
\item
Bravais lattices can be classified according to the point
groups which leave them invariant.
\end{enumerate}
\subsection{Remarks on point groups}
\begin{enumerate}
\item
All point groups are subgroups of $O(3)$ by construction. Thus,
$A \in R_0\ \then\ \det A = \pm 1$. If $\det A = 1\ \forall\, A
\in R_0$, then $R_0$ is called a point group of the first kind,
otherwise a point group of the second kind. Point groups of
the first kind consist of only rotations.
\item
The inversion $i \in O(3)$ is defined by $i\, \xv = - \xv\
\forall\, \xv \in {\mathbb R}^3$. $\then\ \det i = - 1$. For
point groups of the second kind we distinguish point groups
containing $i$ from point groups not containing $i$.
\item
Since $\bv \in B\ \then\ - \bv \in B$ for every Bravais lattice
vector, the symmetry groups of the Bravais lattices are point
groups of the second kind which do contain $i$.
\item
Which rotations are possible? The fact that any point group must
leave a Bravais lattice invariant restricts the possible rotation
angles. It means that every rotation $D \in R_0$ must map
primitive vectors on vectors in $B$,
\begin{equation}
     D \av_j = m_{jk} \av_k \epc \qd m_{jk} \in {\mathbb Z} \epp
\end{equation}
Define a basis transformation $M$ in ${\mathbb R}^3$ by
\begin{align} \label{tracepgrotation}
     & M \av_j = \ev_j \epc \qd \then\ M D M^{-1} \ev_j = m_{jk} \ev_k \epc \notag \\
     \then & \tr D = \tr M D M^{-1} = m_{jk} \<\ev_j, \ev_k\> = m_{jj} \in {\mathbb Z} \epp
\end{align}
Recall how the rotation angle is related to the trace of a rotation
matrix. The trace is invariant under coordinate transformation.
Hence, we may calculate it in a coordinate system in which the
axis of rotation is the $z$-axis. Denoting the rotation angle
by $\ph$ we obtain
\begin{equation} \label{tracerotationangle}
     \tr D = \tr \begin{pmatrix}
                    \cos(\ph) & \sin(\ph) & \\
                    - \sin(\ph) & \cos(\ph) & \\
		    & & 1
                 \end{pmatrix} = 1 + 2 \cos(\ph) \epp
\end{equation}
Combining (\ref{tracepgrotation}) and (\ref{tracerotationangle}) we
conclude that allowed angles $\ph$ must satisfy the condition
$2 \cos(\ph) \in {\mathbb Z}$, or
\begin{equation}
     \cos(\ph) = 0, \pm \2, \pm 1 \epp
\end{equation}
Thus, the only admissible values of $\ph \in [0,2\p)$ are
\begin{equation} \label{allowedangles}
     \ph = 0, \frac{\p}{3}, \frac{\p}{2}, \frac{2 \p}{3}, \p \epp
\end{equation}
The corresponding rotation axes are called $6$-fold, $4$-fold,
$3$-fold, $2$-fold.
\item
Point groups do not only contain only rotation axes of finite order,
they are also finite groups. Comparison with the known finite
subgroups of $SO(3)$ leaves 11 point groups of the first kind
compatible with (\ref{allowedangles}). From these we can construct
altogether $32$ point groups. Their number is, in particular, finite.
\end{enumerate}
\subsection{Classification of all point and space groups}
The full symmetry groups $R$ of crystals $M$ are discrete subgroups
of $E(3)$ which contain a Bravais lattice $B$ as a normal subgroup.
Their number is finite as well. In mathematics (and crystallography)
such groups are called space groups. They have been completely
classified and can be described by symmetry elements like rotations,
reflections, glide reflections and screw rotations. The table gives
an overview over the group theoretic classification of the Bravais
lattices and crystals.
\begin{center}
\addtolength{\tabcolsep}{14pt}
\begin{tabular}{@{}lll@{}}
\toprule
& Bravais lattices & crystals \\
\midrule
point groups & 7 crystal systems (4 in 2d) & 32 crystal classes (13 in 2d) \\[1ex]
space groups & 14 Bravais classes (5 in 2d) & 230 (17 in 2d) \\
\bottomrule
\end{tabular}
\end{center}
It is very instructive to have a look at least at the pictorial
representations of the crystal systems and Bravais classes.
They are shown in Figure~\ref{fig:crystalclasses}.

\begin{figure}
\begin{center}
\renewcommand{\arraystretch}{1.3}
\begin{tabular}{@{}llcccc@{}}
\toprule
Crystal system & Group & primitive & base-centered & body-centered & face-centered \\
\midrule
Triclinic & $C_i$ &
\raisebox{-34pt}{\includegraphics[width=.12\textwidth]{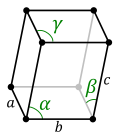}} &&& \\[7ex]
Monoclinic & $C_{2h}$ &
\raisebox{-34pt}{\includegraphics[width=.12\textwidth]{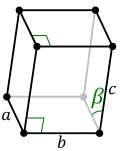}} &
\raisebox{-34pt}{\includegraphics[width=.12\textwidth]{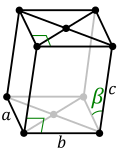}} && \\[8ex]
Orthorhombic & $D_{2h}$ &
\raisebox{-34pt}{\includegraphics[width=.11\textwidth]{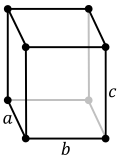}} &
\raisebox{-34pt}{\includegraphics[width=.11\textwidth]{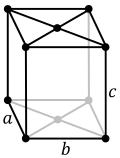}} &
\raisebox{-34pt}{\includegraphics[width=.11\textwidth]{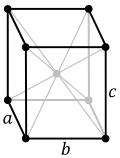}} &
\raisebox{-34pt}{\includegraphics[width=.11\textwidth]{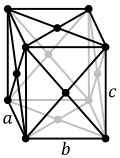}} \\[8ex]
Tetragonal & $D_{4h}$ &
\raisebox{-36pt}{\includegraphics[width=.11\textwidth]{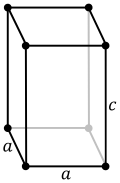}} &&
\raisebox{-36pt}{\includegraphics[width=.11\textwidth]{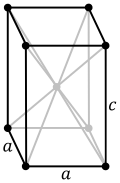}} & \\[8ex]
Rhombohedral & $D_{3d}$ &
\raisebox{-24pt}{\includegraphics[width=.12\textwidth]{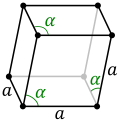}} &&& \\[6ex]
Hexagonal & $D_{6h}$ &
\raisebox{-24pt}{\includegraphics[width=.14\textwidth]{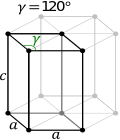}} &&& \\[6ex]
Cubic & $O_{h}$ &
\raisebox{-34pt}{\includegraphics[width=.11\textwidth]{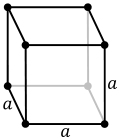}} & &
\raisebox{-34pt}{\includegraphics[width=.11\textwidth]{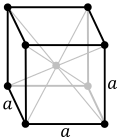}} &
\raisebox{-34pt}{\includegraphics[width=.11\textwidth]{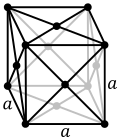}} \\
\bottomrule
\end{tabular}
\end{center}
\caption{\label{fig:crystalclasses} The crystal systems and Bravais classes
(from Wikipedia, the free encyclopedia). A parallelepiped representing the
point-group symmetry of a Bravais lattice has six parameters, three lengths
of its edges and three angles. If all edge lengths and all angles are mutually
distinct, the symmetry is minimal (only the inversion). This is the triclinic
crystal class in the table. Considering all possible degeneracies (two right
angles, to equal edge lengths etc.) one runs through all the listed symmetry
classes, the crystal systems. Some of them can be realized by several Bravais
lattices, giving the different Bravais classes. The second column in the
table contains the name of the point group in so-called Schönflies notation.}
\end{figure}

\subsection{Exercise 4. Lattice planes in the cubic face-centered lattice}
\label{ex:latplan}
As we have seen in section~\ref{subsec:propreclat} all families of lattice
planes of a Bravais lattice can be characterized by normal vectors, which
can be expanded in a basis $\bv_1, \bv_2, \bv_3$ of primitive vectors of
the reciprocal lattice. For a (perpendicular) distance $d$ of the planes
the reciprocal lattice vector $\kv = \sum_{i=1}^3 m_i \bv_i$ is of length
$2 \pi / d$. Since the $m_j$ have no common divisor, $\kv$ is the shortest
reciprocal vector perpendicular to the planes.

\begin{enumerate}
\item \label{3i}
Show that the density of lattice points per unit area in the lattice planes is $d/V_u$,
where $V_u$ if the volume of the unit cell spanned by $\av_1, \av_2, \av_3$.
\item
Show that the reciprocal lattice of the face-centered cubic lattice with lattice
constant $a$ is a body-centered cubic lattice with lattice constant $4\pi / a$.
By definition the lattice constant $a$ is the edge length of the cube which
envelops the unit cell of the face-centered cubic lattice with primitive vectors
\[
   \av_1 = \frac{a}{2} ( \ev_y + \ev_z) \epc \qd
   \av_2 = \frac{a}{2} ( \ev_x + \ev_z) \epc \qd
   \av_3 = \frac{a}{2} ( \ev_x + \ev_y) \epp
\]
$\ev_x, \ev_y, \ev_z$ is the canonical orthonormal basis of $\mathbb{R}^3$.
The corresponding primitive vectors for the body-centered cubic lattice
with lattice constant $a'$ are
\[
   \av_1 = \frac{a'}{2} ( - \ev_x + \ev_y + \ev_z) \epc \qd
   \av_2 = \frac{a'}{2} ( \ev_x - \ev_y + \ev_z) \epc \qd
   \av_3 = \frac{a'}{2} ( \ev_x + \ev_y - \ev_z) \epp
\]
\item
Find the Miller indices $(m_1, m_2, m_3)$ of that plane of the face-centered
cubic lattice which has the highest density of lattice points. Here it may
be helpful to use the connection between the density and the reciprocal lattice
vector $\kv$.
\end{enumerate}
\subsection{Exercise 5. Face centered tetragonal structure}
Why does the face centered tetragonal structure not appear in the list
of the 14 Bravais classes? How does this lattice fit into one of the 14 Bravais
classes?
\mysection{The action of the Bravais lattice on states}
\subsection{Shift operators, lattice momentum and Bloch's theorem}
\label{sec:bloch}
For a set of primitive vectors $\{\av_1, \av_2, \av_3\} \in B$ define the
corresponding shift operators $U_{\av_j}$, acting on a single-particle
space of states, by
\begin{equation}
     U_{\av_j} \Ps (\xv) = \Ps (\xv + \av_j) \epp
\end{equation}
Then
\begin{equation}
     U_{\av_j}^{-1} = U_{\av_j}^+ \epc \qd
     [U_{\av_j}, U_{\av_k}] = 0 \epc \qd j, k = 1, 2, 3 \epp \label{comu}
\end{equation}
For any Bravais lattice vector $\Rv = \ell \av_1 + m \av_2 + n \av_3$
the operator
\begin{equation}
     U_\Rv = U_{\av_1}^\ell U_{\av_2}^m U_{\av_3}^n
\end{equation}
is therefore uniquely defined and naturally acts as
\begin{equation}
     U_\Rv \Ps (\xv) = \Ps (\xv + \Rv)
\end{equation}
on single-particle wave functions.

Equation (\ref{comu}) implies that the $U_{\av_j}$ have a joint system of
eigenfunctions. If $\Ps$ is such an eigenfunction, then
\begin{equation}
     U_{\av_j} \Ps (\xv) = \Ps (\xv + \av_j) = \om (\av_j) \Ps (\xv) \epp
\end{equation}
Since $U_{\av_j}$ is unitary, $|\om (\av_j)| = 1$ implying that
$\exists \, v_j \in {\mathbb R}$ such that $\om (\av_j) = \re^{\i 2 \p v_j}$.
Let $\Rv = \ell \av_1 + m \av_2 + n \av_3$, $\kv = v_1 \bv_1 + v_2 \bv_2
+ v_3 \bv_3$. Then
\begin{equation} \label{blochform}
     \Ps (\xv + \Rv) = U_\Rv \Ps (\xv) = \re^{\i 2 \p(\ell v_1 + m v_2 + n v_3)} \Ps(\xv)
        = \re^{\i \<\kv, \Rv\>} \Ps(\xv) \epp
\end{equation}

Thus, for every common eigenfunction $\Ps$ of the three generators $U_{\av_j}$
of lattice translations $\exists\, \kv \in {\mathbb R}^3$ such that
(\ref{blochform}) holds for all $\Rv \in B$. The vector $\kv$ is a triple
of quantum numbers characterizing the eigenstates of the lattice
translation operators in very much the same manor as the momentum $\pv$
is a triple of quantum numbers that characterize the eigenstates of the
operator of infinitesimal translations, the momentum operator. For this
analogy $\kv$ is called the lattice momentum.

Let $\gv \in \overline B$, $\Rv \in B$. Then $\<\gv, \Rv\> = m 2 \p$ for some
$m \in {\mathbb Z}$ and $\re^{\i \<\kv + \gv, \Rv\>} = \re^{\i \<\kv, \Rv\>}$.
This means that $\kv$ and $\kv + \gv$ characterize the same eigenstate
of the lattice translation operator, or that the lattice momentum is defined
only modulo reciprocal lattice vectors. For this reason we may restrict the
domain of definition of $\kv$ to any unit cell of the reciprocal lattice. 
This domain is conventionally taken as the first Brillouin zone, which
explains the importance of the latter.

In the language of group theory the lattice momenta $\kv$ label the
irreducible representations of the Bravais lattice. Since the Bravais
lattice is an Abelian group, all of its irreducible representations
must be one-dimensional. They act by multiplication with complex numbers
as can be seen in equation (\ref{blochform}).

The Hamiltonian of the solid (\ref{hamsol}) is invariant under any
infinitesimal translation. For this reason the center of mass
momentum of the solid is conserved. In nature the translation
symmetry is affected by a mechanism called spontaneous symmetry
breaking. After separating the center of mass motion the ground
state of the Hamiltonian (\ref{hamsol}) is less symmetric than
the Hamiltonian itself. Instead of the full translation symmetry
it exhibits a discrete translation symmetry with an underlying
Bravais lattice $B$.  Effective Hamiltonians describing the dynamics
above the ground state again have the reduced symmetry described
by a Bravais lattice. The simplest example for a class of such
effective Hamiltonians is the Hamiltonian a single electron in a
lattice periodic potential. The study of this class of Hamiltonians
is called band theory. We will have a closer look at it below.
In any case, single-particle Hamiltonians $H$ which are invariant
under the action of a Bravais lattice,
\begin{equation}
     [H, U_{\av_j}] = 0 \epc \qd j = 1, 2, 3 \epc
\end{equation}
play an important role in solid state physics. As we have seen,
their wave functions can be labeled by the lattice momentum
quantum numbers. This is the statement of Bloch's theorem.
\begin{theorem}
Bloch \cite{Bloch29}. The eigenfunctions of a single-particle Hamiltonian $H$,
periodic with respect to a Bravais lattice $B$, can be labeled
by lattice momenta $\kv \in BZ$, where $BZ \subset \overline B$
is the Brillouin zone associated with $B$. An eigenfunction
$\Ps_\kv$ of $H$ then has the following properties with respect to
translations by Bravais lattice vectors,
\begin{equation}
     \Ps_\kv (\xv + \Rv) = \re^{\i \<\kv, \Rv\>} \Ps_\kv (\xv)
\end{equation}
for all $\Rv \in B$.
\end{theorem}
Let $\Ps_{\kv, \a}$ an eigenstate with lattice momentum $\kv$
of a lattice periodic Hamiltonian~$H$. Here we denote all other
quantum numbers needed to specify the state by $\a$. According
to Bloch's theorem $u_{\kv, \a} (\xv) =  \re^{- \i \<\kv, \xv\>}
\Ps_{\kv, \a} (\xv)$ is a lattice periodic function, $u_{\kv, \a}
(\xv + \Rv) = u_{\kv, \a} (\xv)$. This implies the following
corollary to Bloch's theorem.
\begin{corollary}
The eigenfunctions of a lattice periodic single-particle Hamiltonian
$H$ are of the form
\begin{equation}
     \Ps_{\kv, \a} (\xv) = \re^{\i \<\kv, \xv\>} u_{\kv, \a} (\xv) \epc
\end{equation}
where $\kv \in BZ$ is a lattice momentum vector and $u_{\kv, \a}$
is a lattice periodic function.
\end{corollary}
Hence, we may think of the eigenfunctions of a lattice periodic
Hamiltonian as of amplitude-modulated plane waves, for which the
modulation has the periods of the corresponding Bravais lattice.

\subsection{Periodic boundary conditions}
\label{sec:pbcs}
For the calculation of thermodynamic quantities in the framework
of statistical mechanics (in particular) it is necessary to count
states. For this reason we prefer systems of finite size which
have a discrete spectrum. In solid state physics this can be
enforced by introducing `boundaries' (by putting the system into
a box). After having counted the states one considers the limit,
when the systems size goes to infinity (the thermodynamic limit).

In general, boundaries are incompatible with lattice translations
induced by a Bravais lattice. They break the translational
symmetry and invalidate Bloch's theorem. A way out of this dilemma
is by employing periodic boundary conditions.

For a $d$-dimensional systems periodic boundary conditions can
be realized by starting with a parallelepiped and identifying opposite
faces. This amounts to bending the parallelepiped to a torus
in $d+1$ dimensions. For this reason periodic boundary conditions
are also sometimes called toroidal boundary conditions.

Imposing periodic boundary conditions on a 1d system of $L$ sites
with lattice constant $a$ we find for a state with lattice momentum
$k$ that
\begin{equation}
     \Ps_k (x + aL) = \Ps_k (x) = \re^{\i k a L} \Ps_k (x) \epp
\end{equation}
For this to hold, the lattice momentum must be restricted to
the values
\begin{equation}
     k = m \frac{2 \p}{a L} \epc
\end{equation}
where $m \in {\mathbb Z}$ in such a way that $k$ lies in the
Brillouin zone. The reciprocal lattice is generated by $b = 2\p/a$.
The Brillouin zone is the interval $BZ = [-\p/a, \p/a)$, and
$k \in BZ\ \Leftrightarrow\ -L/2 \le m < L/2$. Thus, for every
$L \in {\mathbb N}$ there are $L$ inequivalent $k$s in the
Brillouin zone.

The argument is similar in any number of dimensions. In order to
obtain the lattice momentum quantization condition, e.g.\ in 3d,
we expand $\kv$ in a basis of primitive vectors of the reciprocal
lattice, $\kv = v_1 \bv_1 + v_2 \bv_2 + v_3 \bv_3$. Then, for a
state $\Ps_\kv$ of lattice momentum~$\kv$,
\begin{equation}
     \Ps_\kv (\xv + L \av_j) = \re^{\i \<\kv, L \av_j\>} \Ps_\kv (\xv)
        = \Ps_\kv (\xv) \epc
\end{equation}
requiring that $\<\kv, L \av_j\> = 2\p v_j L = 2 \p m_j$ for
$m_j \in {\mathbb Z}$. Restricting $\kv$ to the first Brillouin zone
means
\begin{equation}
     v_j = \frac{m_j}{L} \mod 1\ \Leftrightarrow\
           - \frac{L}{2} \le m_j < \frac{L}{2} \epc
\end{equation}
i.e., there are $L^3$ inequivalent lattice momenta in the first
Brillouin zone. Let us rephrase this statement in the following form.
\begin{lemma}
There are as many lattice momenta in the Brillouin zone that are
compatible with periodic boundary conditions as unit cells in the
crystal.
\end{lemma}
\begin{remark}
In 3d periodic boundary conditions cannot be physically realized.
However, the density of states of a macroscopically large system
($\sim 10^{23}$ particles) is practically independent of the
boundary conditions. As long as we are not interested in the
boundaries themselves, periodic boundary conditions are justified.
\end{remark}
\begin{remark}
Mathematically periodic boundary conditions imply that we are
dealing with functions that are periodic with periods of the
large parallelepiped spanned by $L \av_1$, $L \av_2$, $L \av_3$.
Hamiltonians must be defined in a way that is compatible with
this periodicity. On the corresponding space of states the
generators $U_{\av_j}$ of translations by primitive vectors turn
into generators of the cyclic group of order $L$,
\begin{equation}
     U_{\av_j}^L = \id \epc \qd j = 1, 2, 3 \epc
\end{equation}
which equivalently might have served as a starting point for
introducing periodic boundary conditions.
\end{remark}

\mysection{Phonons -- spectrum and states}
\subsection{The Hamiltonian of the lattice vibrations in harmonic approximation}
In lecture~\ref{lect:bo} we have discussed the Born-Oppenheimer
approximation. We have seen that, in this approximation, the
motion of the much heavier ions decouples from the motion
of the electrons up to fourth order in an expansion in
the deviations $\uv$ of the ion positions from their equilibrium
values $\Rv^{(0)}$. In most of this lecture we shall be dealing
with ideal solids. By definition, the equilibrium positions of
the ions $\Rv^{(0)}$ in an ideal solid are the points of a crystal
lattice. Solids in nature can be very close to ideal. The idealization
of a perfect crystal is a good starting point to describe real solids.

Consider a crystal with Bravais lattice $B$ and $N$ ions per
unit cell. In such a crystal it makes sense to label the
coordinates of the vector $\uv$ as $u^\a (\Rv)$, where
$\Rv \in B$ and
\begin{equation}
     \a = (r,j) \in I = \{1, \dots, N\} \times \{x, y, z\} \epp
\end{equation}
Here $r = 1, \dots, N$ counts the ions in a unit cell, and
$j = x, y, z$ denotes their Cartesian coordinates. For
the dimensionless ion masses (cf.\ Section~\ref{sec:adiadecoupling})
we introduce the notation
\begin{equation}
     \m_\a = \m_{(r,j)} = \frac{M_r}{M}
\end{equation}
Then the operator $T_\uv$, equation (\ref{deftu}), of the kinetic
energy of the ions takes the form
\begin{equation} \label{phononkineticenergy}
     T_\uv = - \2 \sum_{\Rv \in B} \sum_{\a \in I}
                  \frac{1}{\m_\a} \6_{u^\a (\Rv)}^2 \epp
\end{equation}

We shall treat the motion of the ions within the harmonic
approximation (\ref{lathamharmonic}). This will allow us
to develop a rather simple and general theory which
nevertheless describes many of the experimental observations
quite accurately. Within the harmonic approximation the
potential energy of the ions can be written as
\begin{equation} \label{phononpotentialenergy}
     V(\uv) = \2 \sum_{\Rv, \Sv \in B} \sum_{\a, \be \in I}
              u^\a (\Rv) K_{\a \be} (\Rv, \Sv) u^\be (\Sv) \epc
\end{equation}
where
\begin{equation} \label{deffm}
     K_{\a \be} (\Rv, \Sv) =
        \frac{\6^2 \e_0 (\Xv)}{\6 X^\a (\Rv) \6 X^\be (\Sv)}
	      \biggr|_{\Xv = \Rv^{(0)}}
\end{equation}
is the so-called force matrix (see (\ref{harmpot})).

Thus, the Hamiltonian of the lattice vibrations in harmonic
approximation is
\begin{equation} \label{hharmvib}
     H = - \2 \sum_{\Rv \in B} \sum_{\a \in I}
              \frac{1}{\m_\a} \6_{u^\a (\Rv)}^2
         + \2 \sum_{\Rv, \Sv \in B} \sum_{\a, \be \in I}
              u^\a (\Rv) K_{\a \be} (\Rv, \Sv) u^\be (\Sv) \epp
\end{equation}
It is a quadratic form in the position operators of the
ions and the corresponding derivatives. We will diagonalize
this quadratic form. This will reduce the spectral problem
of the Hamiltonian to the spectral problem of independent
harmonic oscillators describing the quantized normal modes
of the ideal harmonic solid. In order to control the number
of the normal modes, we will employ periodic boundary
conditions as introduced in the previous lecture.
\subsection{Implications of the translational invariance}
Fix a set of primitive vectors $\{\av_1, \av_2, \av_3\}
\subset B$. Define the action of $B$ on functions $v$ on $B$
obeying periodic boundary conditions by
\begin{equation}
     U_{\av_j} v (\Rv) = v (\Rv + \av_j) \epp
\end{equation}
The space spanned by such functions is a finite dimensional
vector space. For the action of $U_{\av_j}$ on this space there
is an $L \in {\mathbb N}$ such that $U_{\av_j}^L = \id$,
$j = 1, 2, 3$. Thus, if $\om_j$ is an eigenvalue of $U_{\av_j}$
we must have $|\om_j| = 1$ (because of unitarity) and $\om_j^L = 1$.
It follows that
\begin{equation}
     \om_j = \re^\frac{\i m_j 2 \p}{L} = \re^{\i \<\kv, \av_j\>}
\end{equation}
for some $\kv = (m_1 \bv_1 + m_2 \bv_2 + m_3 \bv_3)/L \in BZ$.
As we have seen in the previous lecture there are altogether
$L^3$ such vectors.

It is easy to find the corresponding eigenfunctions of the
shift operators. If $v_\kv$ is an eigenfunctions with lattice
momentum $\kv$, and $\Rv = \ell \av_1 + m \av_2 + n \av_3 \in B$,
then
\begin{multline} \label{shiftev}
     v_\kv (\Rv) = v_\kv (\ell \av_1 + m \av_2 + n \av_3) =
              U_{\av_1}^\ell U_{\av_2}^m U_{\av_3}^n v_\kv (0) \\[1ex]
            = \om_1^\ell \om_2^m \om_3^n v_\kv (0)
	    = \re^{\i \<\kv, \Rv\>} v_\kv (0)
\end{multline}
for all $\Rv \in B$. This determines $v_\kv$ up to normalization.
We fix the normalization by setting
\begin{equation} \label{fouriernorm}
     v_\kv (0) = \frac{1}{\sqrt{L^3}} \epp
\end{equation}
All joint eigenfunctions of the $U_{\av_j}$ are of this form, and
all these functions are joint eigenfunctions of the $U_{\av_j}$.
Hence, they form a basis of on the space of complex valued functions
on $B$.
\begin{remark}
\begin{enumerate}
\item
We have just invented the (discrete) Fourier transformation.
\item
With the choice (\ref{fouriernorm}) the usual Hermitian scalar
product of two eigenfunctions takes the values
\begin{equation} \label{fourierortho}
     \<v_\kv, v_\qv\> = \sum_{\Rv \in B} v_\kv^* (\Rv) v_\qv (\Rv)
        = \frac{1}{L^3} \sum_{\Rv \in B} \re^{- \i \<\kv - \qv, \Rv\>}
	= \de_{\kv, \qv}
\end{equation}
for any two $\kv, \qv \in BZ$.
\end{enumerate}
\end{remark}
\subsection{Block diagonalization of the force matrix}
We will first of all diagonalize the force matrix 
$K_{\a \be} (\Rv, \Sv)$ defined in (\ref{deffm}). For this purpose
and for later use as well we list its main properties.
\begin{enumerate}
\item
Symmetry. From its very definition as a second derivative matrix
and from the commutativity or the partial derivatives we get at once
that
\begin{equation} \label{kmsymmetry}
     K_{\a \be} (\Rv, \Sv) = K_{\be \a} (\Sv, \Rv) \epp
\end{equation}
\item
Translation symmetry.
\begin{equation}
     K_{\a \be} (\Rv + \av, \Sv + \av) = K_{\a \be} (\Rv, \Sv) \qd
        \forall\ \av \in B \epp
\end{equation}
Inserting here $\av = - \Sv$ we obtain
\begin{equation}
     K_{\a \be} (\Rv, \Sv) = K_{\a \be} (\Rv - \Sv, 0)
        = K_{\a \be} (\Rv - \Sv) \epc
\end{equation}
where the second equation is a definition. The meaning is that
forces between ions depend only on the relative positions of
unit cells.
\item
In crystal lattices with inversion symmetry we have in addition
that
\begin{equation}
     K_{\a \be} (\Rv) = K_{\a \be} (- \Rv) \epp
\end{equation}
Combining this with (\ref{kmsymmetry}) we see that the force
matrix in crystal lattices with inversion symmetry is symmetric
in every unit cell,
\begin{equation}
     K_{\a \be} (\Rv) = K_{\be \a} (\Rv) \epp
\end{equation}
\end{enumerate}
The force matrix defines an operator on the space of functions
with periodic boundary conditions on $B$,
\begin{equation}
     \widehat K_{\a \be} f (\Rv) = \sum_{\Sv \in B} K_{\a \be} (\Rv - \Sv) f (\Sv) \epp
\end{equation}
It is easy to see that $\widehat K_{\a \be}$ commutes with the
shift operators $U_{\av_j}$, $j = 1, 2, 3$,
\begin{multline}
     U_{\av_j} \widehat K_{\a \be} f (\Rv)
        = \sum_{\Sv \in B} K_{\a \be} (\Rv + \av_j - \Sv) f (\Sv) \\
        = \sum_{\Sv \in B} K_{\a \be} (\Rv - \Sv) f (\Sv + \av_j)
	= \widehat K_{\a \be} U_{\av_j} f (\Rv) \epc \qd
	  \Leftrightarrow\ [\widehat K_{\a \be}, U_{\av_j}] = 0 \epp
\end{multline}
Hence, $U_{\av_j}$ and $\widehat K_{\a \be}$ possess a joint system
of eigenfunctions. Since the $v_\kv$, equation (\ref{shiftev}),
form already an orthonormal basis of non-degenerate eigenfunctions,
they must be also eigenfunctions of $\widehat K_{\a \be}$,
\begin{equation} \label{fmonfourier}
     \widehat K_{\a \be} v_\kv (\Rv) = \k_{\a \be} (\kv) v_\kv (\Rv) \epp
\end{equation}

In order to block diagonalize the quadratic form
(\ref{phononpotentialenergy}) we expand the displacements $u^\a (\Rv)$
into their Fourier modes,
\begin{equation} \label{dispfourier}
     u^\a (\Rv) = \sum_{\kv \in BZ} \x_\kv^\a v_\kv (\Rv) \epp
\end{equation}
Using that $u^\a (\Rv) \in {\mathbb R}$ and $v_\kv^* = v_{-\kv}$
we see that
\begin{equation} \label{fourierreality}
     {\x_\kv^\a}^* = \x_{-\kv}^\a \epp
\end{equation}
Inserting (\ref{dispfourier}) into (\ref{phononpotentialenergy})
and making use of (\ref{fmonfourier}), (\ref{fourierreality}) we
obtain
\begin{multline}
     V(\uv) = \2 \sum_{\a, \be \in I} \bigl\<u^\a, \widehat K_{\a \be} u^\be\bigr\>
       = \2 \sum_{\a, \be \in I} \sum_{\qv, \kv \in BZ} {\x_\qv^\a}^* \x_\kv^\be
            \bigl\<v_\qv, \widehat K_{\a \be} v_\kv\bigr\> \\
       = \2 \sum_{\qv \in BZ} \sum_{\a, \be \in I}
	    \x_{-\qv}^\a \k_{\a \be} (\qv) \x_\qv^\be \epc
\end{multline}
the block diagonal form of the potential energy.
\subsection{Transformation of the kinetic energy}
Let us now apply the same transformation to the kinetic energy
operator. First of all
\begin{equation}
     \<u^\a, v_{-\kv}\> = \sum_{\qv \in BZ} \x_{-\qv}^\a \<v_\qv, v_{-\kv}\>
        = \x_\kv^\a = \sum_{\Sv \in B} u^\a (\Sv) v_{-\kv} (\Sv) \epc
\end{equation}
implying that
\begin{equation}
     \frac{\6}{\6 u^\a (\Rv)}
        = \sum_{\kv \in BZ} \frac{\6 \x_\kv^\a}{\6 u^\a (\Rv)}
	                    \frac{\6}{\6 \x_\kv^\a}
        = \sum_{\kv \in BZ} v_{-\kv} (\Rv) \frac{\6}{\6 \x_\kv^\a} \epp
\end{equation}
It follows that
\begin{multline}
     T_\uv = - \2 \sum_{\a \in I} \frac{1}{\m_\a} \sum_{\Rv \in B}
                  \sum_{\qv, \kv \in BZ} v_{-\qv} (\Rv) v_{-\kv} (\Rv)
		        \frac{\6}{\6 \x_\qv^\a} \frac{\6}{\6 \x_\kv^\a} \\[1ex]
           = - \2 \sum_{\a \in I} \frac{1}{\m_\a}
                  \sum_{\qv, \kv \in BZ} \<v_\qv, v_{-\kv}\>
		        \frac{\6}{\6 \x_\qv^\a} \frac{\6}{\6 \x_\kv^\a}
           = - \2 \sum_{\qv \in BZ} \sum_{\a \in I} \frac{1}{\m_\a}
		        \frac{\6}{\6 \x_{-\qv}^\a} \frac{\6}{\6 \x_\qv^\a} \epp
\end{multline}
Here we have used (\ref{fourierortho}) in the second equation.
We see that the lattice Fourier transformation has diagonalized
$T_\uv$.

In the next step we want to completely diagonalize the
force matrix while keeping the diagonal form of the kinetic
energy operator. To achieve the latter goal, we first
rescale the complex coordinates, setting
\begin{equation}
     \h_\qv^\a = \sqrt{\m_\a}\, \x_\qv^\a\ \then\
     \frac{\6}{\6 \h_\qv^\a} = \frac{1}{\sqrt{\m_\a}} \frac{\6}{\6 \x_\qv^\a} \epp
\end{equation}
Further defining
\begin{equation}
     \widetilde \k_{\a \be} (\qv) = \frac{\k_{\a \be} (\qv)}{\sqrt{\m_\a \m_\be}}
\end{equation}
we obtain the following form of the Hamiltonian (\ref{hharmvib}),
\begin{equation}
     H = \2 \sum_{\qv \in BZ} \biggl\{ - \sum_{\a \in I}
            \frac{\6}{\6 \h_{-\qv}^\a} \frac{\6}{\6 \h_\qv^\a}
	  + \sum_{\a, \be \in I} \h_{-\qv}^\a \widetilde \k_{\a \be} (\qv) \h_\qv^\be
	    \biggr\} \epp
\end{equation}
\subsection{\boldmath Properties of the matrix $\widetilde \k (\qv)$}
Before we can proceed we have to understand the properties
of the matrix $\widetilde \k (\qv)$.
\begin{enumerate}
\item
$\widetilde \k (\qv)$ is Hermitian, since
\begin{align}
     \k_{\a \be}^* (\qv) & = \bigl\<v_\qv, \widehat K_{\a \be} v_\qv\bigr\>^* \notag \\[1ex]
        & = \sum_{\Rv, \Sv \in B} v_\qv (\Rv) K_{\a \be} (\Rv - \Sv) v_{-\qv} (\Sv) \notag \\
        & = \sum_{\Rv, \Sv \in B} v_\qv (\Rv) K_{\be \a} (\Sv - \Rv) v_{-\qv} (\Sv) \notag \\
        & = \sum_{\Rv, \Sv \in B} v_{-\qv} (\Rv) K_{\be \a} (\Rv - \Sv) v_\qv (\Sv)
          = \bigl\<v_\qv, \widehat K_{\be \a} v_\qv\bigr\> = \k_{\be \a} (\qv) \epp
\end{align}
Here we have used the symmetry of the force matrix in the third
equation.
\item
$\widetilde \k (\qv)$ is non-negative. This follows, since the potential
energy $V$ is assumed to have a total minimum for $\uv = 0$ with $V = 0$.
Hence
\begin{equation}
     V (\uv) = \2 \sum_{\a, \be \in I}
                  \h_{-\qv}^\a \widetilde \k_{\a \be} (\qv) \h_\qv^\be \ge 0 \epp
\end{equation}
\item
$\widetilde \k (\qv)$ and $\widetilde \k (-\qv)$ are similar matrices.
First of all
\begin{equation}
     \k_{\a \be}^* (\qv)
        = \sum_{\Rv, \Sv \in B} v_\qv (\Rv) K_{\a \be} (\Rv - \Sv) v_{-\qv} (\Sv)
	= \k_{\a \be} (- \qv) \epc
\end{equation}
since the force matrix is real. Then also
\begin{equation}
     \widetilde \k_{\a \be}^* (\qv) = \widetilde \k_{\a \be} (- \qv) \epp
\end{equation}

According to (i) and (ii) the matrix $\widetilde \k (\qv)$ can be
diagonalized by a unitary transformation and has a non-negative
spectrum $\{\om_\a^2 (\qv)\}_{\a \in I}$ for every $\qv \in BZ$.
Let $\{\yv_\a (\qv)\}_{\a \in I}$ the set of corresponding orthonormal
eigenvectors. Then
\begin{equation}
     \widetilde \k^* (\qv) \yv_\a (\qv)^* = \om_\a^2 (\qv) \yv_\a (\qv)^*
        = \widetilde \k (-\qv) \yv_\a (\qv)^*
\end{equation}
implying that $\{\om_\a^2 (\qv)\}_{\a \in I}$ is the spectrum of
$\widetilde \k (-\qv)$. Hence, $\widetilde \k (\qv)$ and
$\widetilde \k (-\qv)$ are similar matrices.

Since, on the other hand,
\begin{equation}
     \widetilde \k (-\qv) \yv_\a (-\qv) = \om_\a^2 (-\qv) \yv_\a (-\qv)
\end{equation}
by definition of the eigenvectors and eigenvalues, the identification
\begin{equation} \label{ytreversal}
     \yv_\a (-\qv) = \yv_\a (\qv)^*
\end{equation}
(which is one possible choice of indexing the eigenvectors of
$\widetilde \k (-\qv)$ once the eigenvectors of $\widetilde \k (\qv)$
are given) implies that
\begin{equation} \label{omegaeven}
     \om_\a^2 (\qv) = \om_\a^2 (-\qv) \epp
\end{equation}
\end{enumerate}

\subsection{Exercise 6. Classical harmonic chain with various boundary
conditions}
The lectures on lattice vibrations will be accompanied by a set
of exercises on the classical harmonic chain with broken translation
invariance. We shall study the influence of fixed and open boundary
conditions and of a mass defect. There is a good deal to learn from these
exercises, namely something about the irrelevance of the boundary
conditions as far as bulk thermodynamic properties are concerned,
but also something about the typical effects of impurities, such
as the appearance of localized states and impurity levels inside
the band gap.

We start with an important technical device, the transfer matrix, and
with the effect of fixed and open boundary conditions. For this purpose
consider $N$ masses $m_1,\ldots,m_N$ coupled to a linear, harmonic chain
by $N-1$ springs with force constants $k>0$. Periodic boundary conditions
are realized by an additional identical spring connecting the masses $m_1$
and $m_N$. For fixed boundary conditions, the masses $m_1$ and $m_N$ are
coupled with springs of spring constants $k$ to a rigid wall, while open
boundaries are realized if the masses $m_1$ and $m_N$ are not at all
coupled to each other.

In this exercise we shall consider the particular case of equal masses 
$m_1=\ldots=m_N=m$. We want to analyze the harmonic chain by the so-called
transfer matrix method. Upon slight modifications, it is possible to treat
the different boundary conditions in a similar way.
\begin{enumerate}
\item
Find the equations of motion for the deviations $x_n$ from the equilibrium
positions in the periodic case. Employing the ansatz $x_n(t)=\re^{i\omega t}x_n$
the equations of motion imply an eigenvalue problem of the form
$T\xv=\Omega^2\xv$ with $\xv=(x_1,\ldots,x_N)^t$. Obtain the
matrix $T$, and show how $\Omega$ depends on $\omega,m$ and $k$.
\item
Show, that, with the substitution $\psi(n)=x_n,\varphi(n)=\psi(n-1)$, the
eigenvalue problem in (i) can be reformulated as
\[
     \begin{pmatrix}\psi(n+1)\\\varphi(n+1)\end{pmatrix}
        = L_n(\Omega^2)\begin{pmatrix}\psi(n)\\\varphi(n)\end{pmatrix},\
	  L_n(\Omega^2)=L(\Omega^2)=\begin{pmatrix}2-\Omega^2 & -1\\1 & 0\end{pmatrix}.
\]
In the periodic case $L_n(\Omega^2)$ is independent of the site index. Calculate
the eigenvalues and the corresponding eigenvectors of $L(\Omega^2)$. Because
$L(\Omega^2)$ acts like a translation operator, it is useful to write the
eigenvalues in the form $\re^{\pm i\kappa}$. Diagonalize the equation
\[
     \begin{pmatrix}\psi(n+1)\\\varphi(n+1)\end{pmatrix}
        =L^n(\Omega^2)\begin{pmatrix}\psi(1)\\\varphi(1)\end{pmatrix} \epp
\]
How can $\Omega^2$ be expressed in terms of $\kappa$?
\item
The periodic boundary conditions turn into $\psi(N+1)=\psi(1)$ and
$\varphi(N+1)=\varphi(1)$. From this determine all possible eigenfrequencies $\omega$!
\item
Which modification is required for fixed boundaries? Determine all possible
eigenfrequencies $\omega$ in this case.
\item
Show that the modifications necessary for open boundaries lead to the equation
\[
     (1-\Omega^2,-1)L^{N-2}(\Omega^2)\begin{pmatrix}1-\Omega^2\\1\end{pmatrix}=0 \epp
\]
With this, calculate again all possible eigenfrequencies $\omega$. What is the
physical meaning of the solution $\omega=0$? How are the eigenfrequencies of the
open chain and of the chain with fixed boundaries connected with each other?
\end{enumerate}

\mysection{Phonons -- spectrum and states continued}
\subsection{Reduction of the Hamiltonian to a diagonal quadratic form}
For every $\qv \in BZ$ we define a unitary $3N \times 3N$ matrix
\begin{equation}
     Y (\qv) = (\yv_{(1,x)} (\qv), \dots, \yv_{(N,z)} (\qv)) \epp
\end{equation}
This matrix diagonalizes $\widetilde \k (\qv)$,
\begin{equation}
     \widetilde \k (\qv) = 
        Y(\qv) \diag\bigl(\om_{(1,x)}^2, \dots, \om_{(N,z)}^2\bigr) Y^+ (\qv)
\end{equation}
and, because of (\ref{ytreversal}), has the property that
\begin{equation} \label{matytreversal}
     Y(\qv)^* = Y(- \qv) \epp
\end{equation}
Letting
\begin{equation}
     \xv_\qv = Y^+ (\qv)\, {\bm \h}_\qv\
        \Leftrightarrow\ x_\qv^\a = \bigl\<\yv_\a (\qv), {\bm \h}_\qv\bigr\>
\end{equation}
we see that
\begin{equation}
     \sum_{\a, \be \in I} \h_{-\qv}^\a \widetilde \k_{\a \be} (\qv) \h_\qv^\be
        = {\bm \h}_{- \qv}^t \widetilde \k (\qv) {\bm \h}_\qv
	= \sum_{\a \in I} \om_\a^2 (\qv) x_{- \qv}^\a x_\qv^\a \epc
\end{equation}
while for the kinetic term
\begin{multline}
     \sum_{\a \in I} \frac{\6}{\6 \h_{-\qv}^\a} \frac{\6}{\6 \h_\qv^\a}
        = \sum_{\a, \be, \g \in I}
	  \frac{\6 x_{- \qv}^\be}{\6 \h_{-\qv}^\a} \frac{\6 x_\qv^\g}{\6 \h_\qv^\a}
	  \frac{\6}{\6 x_{-\qv}^\be} \frac{\6}{\6 x_\qv^\g} \\
        = \sum_{\a, \be, \g \in I}
	  {Y^+}_\a^\be (- \qv) {Y^+}_\a^\g (\qv)
	  \frac{\6}{\6 x_{-\qv}^\be} \frac{\6}{\6 x_\qv^\g}
        = \sum_{\a \in I}
	  \frac{\6}{\6 x_{-\qv}^\a} \frac{\6}{\6 x_\qv^\a} \epp
\end{multline}
Here we have used (\ref{matytreversal}) in the last equation.
Altogether, we have transformed $H$ into a diagonal quadratic form,
\begin{equation} \label{hamdiagqform}
     H = \2 \sum_{\qv \in BZ} \sum_{\a \in I}
               \biggl\{ - \frac{\6}{\6 x_{-\qv}^\a} \frac{\6}{\6 x_\qv^\a}
	                + \om_\a^2 (\qv) x_{- \qv}^\a x_\qv^\a \biggr\} \epp
\end{equation}
The term in the bracket can be interpreted as the Hamiltonian
of a 1d harmonic oscillator with `complex coordinates'.
\subsection{Zero modes}
\label{sec:zeromodes}
Before going on we have to discuss the question whether the
functions $\om_\a (\qv)$ can be zero. What we can say is
that there are always at least three values of $\a$ for which
$\om_\a (0) = 0$. These special `modes' are connected with the
center of mass motion of the solid. Their existence can be
inferred from the translation invariance of the Hamiltonian
(\ref{hharmvib}) which is inherited from the full Hamiltonian
(\ref{hsolnatunits}) of the solid. For (\ref{hharmvib}) 
translation invariance means invariance under the transformation
\begin{equation}
    u^{(r,j)} (\Rv) \mapsto u^{(r,j)} (\Rv) + \eps^j \epc
\end{equation}
for all $\eps^j \in {\mathbb R}$ and every $(r, j) \in I$.
This transformation leaves the kinetic energy (\ref{phononkineticenergy})
trivially invariant. For the potential energy (\ref{phononpotentialenergy})
we infer that
\begin{multline}
     \6_{\eps^j} V \bigr|_{\eps^j = 0} = \\
        \2 \sum_{\Rv, \Sv \in B} \sum_{r, s = 1}^N \sum_{\ell = x,y,z}
	   \bigl(K_{(r,j)(s,\ell)} (\Rv - \Sv) u^{(s,\ell)} (\Sv) +
	         u^{(r,\ell)} (\Rv) K_{(r,\ell) (s,j)} (\Rv - \Sv)\bigr) \\
        = \sum_{\Rv, \Sv \in B} \sum_{r, s = 1}^N \sum_{\ell = x,y,z}
	   K_{(r,j)(s,\ell)} (\Rv - \Sv) u^{(s,\ell)} (\Sv) = 0
\end{multline}
for arbitrary $u^{(s,\ell)} (\Sv) \in {\mathbb R}$. Here we have
used the symmetry (\ref{kmsymmetry}) of the force matrix in the
second equation. Setting all but one of the displacements equal to
zero and this one equal to one we obtain the relation
\begin{equation} \label{kmtransinv}
     \sum_{\Rv \in B} \sum_{r = 1}^N K_{(r,j)(s,\ell)} (\Rv) = 0
\end{equation}
for the force matrix, which holds for all $(s,\ell) \in I$ and
$j = x, y, z$. On the other hand
\begin{equation}
     \k_{\a \be} (0) = \bigl\<v_0, \widehat K_{\a \be} v_0\bigr\>
        = \frac{1}{L^3} \sum_{\Rv, \Sv \in B} K_{\a \be} (\Rv - \Sv)
        = \sum_{\Rv \in B} K_{\a \be} (\Rv) \epp
\end{equation}
Setting $\a = (r, j), \be = (s, \ell)$ summing over $r$ and using
(\ref{kmtransinv}) we conclude that
\begin{equation} \label{kappatransinv}
     \sum_{r = 1}^N \sqrt{\m_r \m_s} \; \widetilde\k_{(r,j)(s,\ell)} (0) = 0
\end{equation}
for all $(s,\ell) \in I$ and $j = x, y, z$. Thus, there are three
independent linear relations between the rows of the matrix
$\widetilde \k (0)$ which therefore has at least a threefold eigenvalue
zero.  We may order the spectrum of $\widetilde \k_{\a \be} (0)$ in such
a way that the corresponding eigenvectors are $\yv_{(1,j)} (0)$, $j =
x, y, z$. The corresponding `normal coordinates' are $x_0^{(1,j)}
= \<\yv_{(1,j)} (0),{\bm \h}_0\>$. Using this notation, the Hamiltonian
(\ref{hamdiagqform}) splits into
\begin{equation} \label{hamdiagcmsplit}
     H = T_{\rm cm} +
         \2 \sum_{\substack{\qv \in BZ, \a \in I\\ (\qv, \a) \ne (0, (1, j))}}
               \biggl\{ - \frac{\6}{\6 x_{-\qv}^\a} \frac{\6}{\6 x_\qv^\a}
	                + \om_\a^2 (\qv) x_{- \qv}^\a x_\qv^\a \biggr\} \epc
\end{equation}
where
\begin{equation}
     T_{\rm cm} = - \2 \sum_{j = x, y, z} \frac{\6^2}{\6 \bigl(x_0^{(1,j)}\bigr)^2}
\end{equation}
can be interpreted a the kinetic energy of the center of mass
motion of the crystal.

The center of mass motion of the crystal is unbounded.
If there were any other `zero modes', i.e., a higher than
threefold degeneracy of the eigenvalue zero, then there would
be another eigenvector $\yv_\a (\qv)$ corresponding to
unbounded motion. This would necessarily involve an unbounded
relative motion of different parts of the crystal, meaning
that the crystal would disintegrate. In the following we
shall exclude this possibility and concentrate on stable
crystals. We shall also discard the center of mass motion.
Then we remain with the Hamiltonian of the proper lattice
vibrations which we denote
\begin{equation} \label{phham}
     H_{\rm ph} = H - T_{\rm cm} =
         \2 \sum_{(\qv, \a) \in Q}
               \biggl\{ - \frac{\6}{\6 x_{-\qv}^\a} \frac{\6}{\6 x_\qv^\a}
	                + \om_\a^2 (\qv) x_{- \qv}^\a x_\qv^\a \biggr\} \epp
\end{equation}
Here we have introduced the notation
\begin{equation}
     Q = \bigl\{ (\qv, \a) = BZ \times I \big| \om_\a (\qv) \ne 0\bigr\} \epp
\end{equation}
The subindex `$\rm ph$' refers to `phonon' which is the
name of a quantized normal mode of the lattice.

\subsection{Diagonalization of the phonon Hamiltonian}
To accomplish a complete diagonalization of the Hamiltonian 
$H_{\rm ph}$ we introduce the operators
\begin{subequations}
\label{aadagger}
\begin{align}
     a_\qv^\a & = \sqrt{\frac{\om_\a (\qv)}{2}} x_\qv^\a
                + \frac{1}{\sqrt{2 \om_\a (\qv)}} \frac{\6}{\6 x_{- \qv}^\a} \epc \\[1ex]
     {a^+}_\qv^\a & = \sqrt{\frac{\om_\a (\qv)}{2}} x_{- \qv}^\a
                    - \frac{1}{\sqrt{2 \om_\a (\qv)}} \frac{\6}{\6 x_\qv^\a}
\end{align}
\end{subequations}
for all $(\qv, \a) \in Q$. They satisfy the commutation relations
(exercise: check it!)
\begin{equation} \label{heisenbergalg}
     \bigl[a_\qv^\a, a_\kv^\be\bigr] = 0
        = \bigl[{a^+}_\qv^\a, {a^+}_\kv^\be\bigr] \epc \qd
     \bigl[a_\qv^\a, {a^+}_\kv^\be\bigr] = \de_{\qv, \kv} \de^\a_\be \epp
\end{equation}
Inverting (\ref{aadagger}) we obtain
\begin{equation} \label{xdxaadagger}
     x_\qv^\a = \frac{a_\qv^\a + {a^+}_{- \qv}^\a}{\sqrt{2 \om_\a (\qv)}} \epc \qd
     \frac{\6}{\6 x_\qv^\a} = \sqrt{\frac{\om_\a (\qv)}{2}}
                              \bigl( a_{- \qv}^\a - {a^+}_\qv^\a\bigr) \epp
\end{equation}

The latter equation implies that
\begin{multline} \label{localphham}
     \om_\a^2 (\qv) x_{- \qv}^\a x_\qv^\a
        - \frac{\6}{\6 x_{- \qv}^\a} \frac{\6}{\6 x_{\qv}^\a} \\
	= \frac{\om_\a (\qv)}{2} \bigl(a_{- \qv}^\a + {a^+}_\qv^\a\bigr)
	                         \bigl(a_\qv^\a + {a^+}_{- \qv}^\a\bigr)
	- \frac{\om_\a (\qv)}{2} \bigl(a_\qv^\a - {a^+}_{- \qv}^\a\bigr)
	                         \bigl(a_{- \qv}^\a - {a^+}_\qv^\a\bigr) \\[1ex]
        = \om_\a (\qv) \bigl({a^+}_{- \qv}^\a a_{- \qv}^\a +
	                     {a^+}_\qv^\a a_\qv^\a + 1\bigr) \epc
\end{multline}
whenever $\om_\a (\qv) \ne 0$. Here we have used the evenness of the
functions $\om_\a$, equation (\ref{omegaeven}), and the commutation
relations (\ref{heisenbergalg}). Inserting (\ref{localphham}) into
(\ref{phham}) and using once more that $\om_\a$ is an even function
of $\qv$ we arrive at
\begin{equation} \label{phhamaadagger}
     H_{\rm ph} = \sum_{(\qv, \a) \in Q}
                        \om_\a (\qv) \Bigl({a^+}_\qv^\a \, a_\qv^\a + \tst{\2}\Bigr) \epp
\end{equation}
Thus, $H_{\rm ph}$ is decomposed into a sum of independent harmonic
oscillators.

\subsection{Creation and annihilation operators in terms of
the original displacements}
Going step by step backwards, we express the creation and annihilation
operators of the phonons in term of the original displacement
variables and their associated momentum operators
\begin{equation}
     p^\a (\Rv) = - \i \frac{\6}{\6 u^\a (\Rv)} \epp
\end{equation}
We obtain
\begin{align} \label{xintermsofu}
     x_\qv^\a & = \sum_{\be \in I} {Y^+}^\a_\be (\qv) \h^\be_\qv
                = \sum_{\be \in I} {Y^+}^\a_\be (\qv) \sqrt{\m_\be}\, \x_\qv^\be \notag \\
              & = \sum_{\be \in I} {Y^+}^\a_\be (\qv) \sqrt{\m_\be}\, \< u^\be, v_{- \qv}\>
	        = \frac{1}{\sqrt{L^3}}
		  \sum_{\Rv \in B} \sum_{\be \in I} \re^{- \i \<\qv, \Rv\>}
		  \sqrt{\m_\be}\, {Y^+}^\a_\be (\qv) u^\be (\Rv)
\end{align}
and
\begin{align}
     \frac{\6}{\6 x_{- \qv}^\a} &
        = \sum_{\be \in I} \frac{\6 \h_{- \qv}^\be}{\6 x_{- \qv}^\a}
	                   \frac{\6}{\6 \h_{- \qv}^\be}
        = \sum_{\be \in I} \frac{Y^\be_\a (- \qv)}{\sqrt{\m_\be}} 
	                   \frac{\6}{\6 \x_{- \qv}^\be}
        = \sum_{\be \in I} \frac{{Y^+}^\a_\be (\qv)}{\sqrt{\m_\be}} 
	                   \frac{\6}{\6 \x_{- \qv}^\be} \notag \\
        & = \frac{1}{\sqrt{L^3}}
	    \sum_{\Rv \in B} \sum_{\be \in I} \re^{- \i \<\qv, \Rv\>}
	    \frac{{Y^+}^\a_\be (\qv)}{\sqrt{\m_\be}}\, \i p^\be (\Rv) \epp
\end{align}
Inserting the latter two equations into the definitions
(\ref{aadagger}) of the annihilation and creation operators
we obtain
\begin{subequations}
\label{aadintermsofu}
\begin{align}
     a_\qv^\a & = \frac{1}{\sqrt{L^3}}
		  \sum_{\Rv \in B} \sum_{\be \in I} \re^{- \i \<\qv, \Rv\>}
		  {Y^+}^\a_\be (\qv)
		  \Biggl[ \sqrt{\frac{\m_\be \om_\a (\qv)}{2}} u^\be (\Rv)
		         + \frac{\i p^\be (\Rv)}{\sqrt{2 \m_\be \om_\a (\qv)}}
			 \Biggr] \epc \\ \label{adagger}
     {a^+}_\qv^\a & = \frac{1}{\sqrt{L^3}}
		  \sum_{\Rv \in B} \sum_{\be \in I} \re^{\i \<\qv, \Rv\>}
		  Y^\be_\a (\qv)
		  \Biggl[ \sqrt{\frac{\m_\be \om_\a (\qv)}{2}} u^\be (\Rv)
		         - \frac{\i p^\be (\Rv)}{\sqrt{2 \m_\be \om_\a (\qv)}}
			 \Biggr] \epp
\end{align}
\end{subequations}
In this form it is obvious that $a_\qv^\a$ and ${a^+}_\qv^\a$ are
mutually adjoint operators.
\subsection{Construction of the eigenstates}
Let
\begin{equation} \label{groundx}
     \Ps_0 (\uv) = \exp\biggl\{- \2 \sum_{(\kv, \be) \in Q}
                                 \om_\be (\kv) x_{- \kv}^\be x_\kv^\be\biggr\} \epp
\end{equation}
Then
\begin{multline}
     \frac{\6 \Ps_0 (\uv)}{\6 x_{- \qv}^\a}
        = \biggl( - \frac{\om_\a (\qv) x_\qv^\a}{2}
                  - \frac{\om_\a (- \qv) x_\qv^\a}{2} \biggr) \Ps_0 (\uv)
        = - \om_\a (\qv) x_\qv^\a \Ps_0 (\uv) \\
	\Leftrightarrow\ a_\qv^\a \Ps_0 (\uv) = 0
\end{multline}
for all $(\qv, \a) \in Q$. Hence,
\begin{equation} \label{phgsenergy}
     H_{\rm ph} \Ps_0 (\uv) = E_0 \Ps_0 (\uv) \epc \qd
     E_0 = \2 \sum_{(\qv, \a) \in Q} \om_\a (\qv) \epp
\end{equation}
$\Ps_0$ is the ground state, since it is the ground state for
every 1d harmonic oscillator in the sum (\ref{phhamaadagger}).

Comparing (\ref{groundx}) and (\ref{phham}) and recalling the
original definition (\ref{phononpotentialenergy}) of the harmonic
potential, we obtain the ground state wave function as a function
of the displacements $\uv$ of the ions,
\begin{equation} \label{phonongs}
     \Ps_0 (\uv) = \re^{- V(\uv)}
        = \exp \biggl\{ - \2 \sum_{\Rv, \Sv \in B} \sum_{\a, \be \in I}
	                     u^\a (\Rv) K_{\a \be} (\Rv, \Sv) u^\be (\Sv) \biggr\} \epc
\end{equation}
which is a natural generalization of the 1d case.

A general phonon state is generated by the multiple action of
phonon creation operators ${a^+}_\qv^\a$ on the ground state
which, for this reason, is also sometimes called the phonon
vacuum (exercise: repeat the construction of excited states
for the 1d harmonic oscillator based on the Heisenberg algebra
(\ref{heisenbergalg})). Such states are parameterized by
maps $Q \rightarrow {\mathbb N}_0$, $(\qv, \a) \mapsto n_\qv^\a$.
Accordingly we shall denote them as
\begin{equation} \label{phonones}
     \Ps_n (\uv) = \prod_{(\qv, \a) \in Q}
                   {\bigl({a^+}_\qv^\a}\bigr)^{n_\qv^\a} \; \Ps_0 (\uv) \epp
\end{equation}
It follows from the commutation relations (\ref{heisenbergalg})
that
\begin{equation} \label{phononspectrum}
     H_{\rm ph} \Ps_n (\uv) = E_n \Ps_n (\uv) \epc \qd
     E_n = \sum_{(\qv, \a) \in Q} \om_a (\qv) \bigl(n_\qv^\a + \tst{\2}\bigr) \epp
\end{equation}

The eigenvectors $\yv_\a (\qv)$ of $\widetilde \k (\qv)$, together
with the corresponding eigenfrequencies $\om_\a (\qv)$, determine
the creation operators ${a^+}_\qv^\a$, equation (\ref{adagger}),
since $Y_\a^\be (\qv) = (\yv_\a)^\be$. In analogy with the expression
for the quantized electro-magnetic field we shall call them polarization
vectors.

Let us summarize the insight we have gained so far in the following
\begin{theorem}
In order to obtain the spectrum (\ref{phononspectrum}) and
the eigenstates (\ref{phonongs}), (\ref{phonones}) of the
vibrational motion of the ions in a solid in harmonic approximation,
it suffices to calculate the dispersion relations $\om_\a (\qv)$
and the polarization vectors $\yv_\a (\qv)$. For this purpose one
first calculates the matrix
\begin{equation} \label{fmfourier}
     \widetilde \k_{\a \be} (\qv)
        = \sum_{\Rv \in B} \frac{\re^{- \i \<\qv,\Rv\>} K_{\a \be} (\Rv)}
	                        {\sqrt{\m_\a \m_\be}}
\end{equation}
and then its eigenvectors $\yv^\a (\qv)$ and eigenvalues
$\om_\a^2 (\qv)$.
\end{theorem}
The input here is the force matrix. In applications it comes
from quantum chemical calculations or from simple heuristic
models. Note that in (\ref{fmfourier}) every matrix element
$\widetilde \k_{\a \be} (\qv)$ is represented as a (finite)
Fourier series (cf.\ section~\ref{subsec:fourier}) defining
it as a periodic
function in reciprocal space with periods in $\overline B$. We
expect the forces between ions to decay rapidly with distance
and the convergence of the Fourier series (\ref{fmfourier})
in the thermodynamic limit to be uniform, implying that the
limit function is differentiable in $\qv$.

\mysection{Phonons -- examples and general properties}
\subsection{Example 1 -- the harmonic chain}
To start with we reconsider the harmonic chain within the framework
of the general theory. This is a 1d problem with one ion per
unit cell, thus no indices $\a, \be$ are required and $\m = 1$.
Denoting the lattice spacing by $a$ we obtain
\begin{equation}
     q = \frac{n}{L} \cdot \frac{2\p}{a} \epc \qd R = \ell a
\end{equation}
for the quantized lattice momenta $q$ and the Bravais lattice
vectors $R$.

The model force matrix is
\begin{equation}
     K(R,S) = K(\ell a, m a) = K\bigl((\ell - m) a\bigr)
        = \om_0^2 (2 \de_{\ell, m} - \de_{\ell, m+1} - \de_{\ell, m-1}) \epc
\end{equation}
where we have to keep the periodic boundary conditions in mind.
This is a $1 \times 1$ matrix. The polarization vector is $y = 1$
and
\begin{multline}
     \widetilde \k (q) = \sum_\ell \re^{- \i \frac{n 2 \p}{L} \ell} \om_0^2 \,
                         (2 \de_{\ell, 0} - \de_{\ell, 1} - \de_{\ell, -1})
		       = \om_0^2 2 \bigl(1 - \cos(2 \p n/L)\bigr) \\
		       = 2 \om_0^2 \bigl(1 - \cos(q a)\bigr)
		       = 4 \om_0^2 \sin^2 (qa/2) = \om^2 (q) \epp
\end{multline}
It follows that
\begin{equation}
     \om (q) = 2 \om_0 |\sin(qa/2)|
\end{equation}
which is called the dispersion relation of the harmonic chain
with nearest-neighbour interactions.
\subsection{Example 2 -- diatomic chain with alternating forces}
\begin{figure}
\begin{center}
\includegraphics[width=.95\textwidth]{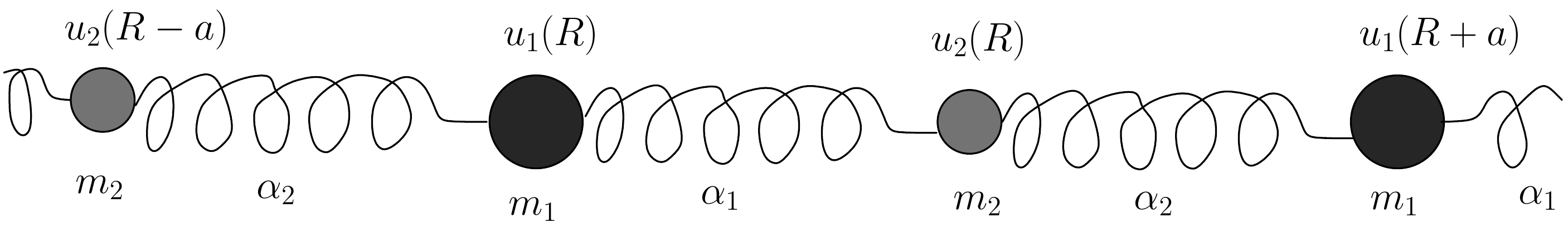}
\caption{\label{fig:diatomic_chain} Sketch of a diatomic chain
with alternating masses $m_1$ and $m_2$ and alternating force
constants $\a_1$ and $\a_2$.}
\end{center}
\end{figure}
The classical model for this configuration is given by the
equations of motion
\begin{align}
     \m_1 \ddot u_1 (R) = \a_1 \bigl(u_2 (R) - u_1 (R)\bigr)
                          - \a_2 \bigl(u_1 (R) - u_2 (R - a)\bigr) \epc \notag \\[1ex]
     \m_2 \ddot u_2 (R) = \a_2 \bigl(u_1 (R + a) - u_2 (R)\bigr)
                          - \a_1 \bigl(u_2 (R) - u_1 (R)\bigr) \epc
\end{align}
where
\begin{equation}
     \m_j = \frac{2 m_j}{m_1 + m_2} \epc \qd j = 1, 2 \epc
\end{equation}
are the dimensionless masses and $\a_1$, $\a_2$ dimensionless
force constants (see Figure~\ref{fig:diatomic_chain}). The
corresponding force matrix is ($\Fv = - \grad V$)
\begin{equation}
     K(R,S) = \begin{pmatrix}
                 (\a_1 + \a_2) \de_{R,S} & - \a_1 \de_{R,S} - \a_2 \de_{R,S+a} \\[1ex]
		 - \a_1 \de_{R,S} - \a_2 \de_{R,S-a} & (\a_1 + \a_2) \de_{R,S}
              \end{pmatrix} \epp
\end{equation}
It follows that
\begin{equation}
     \k (q) = \begin{pmatrix}
                 \a_1 + \a_2 & - \a_1 - \a_2 \re^{- \i q a} \\[1ex]
		 - \a_1 - \a_2 \re^{\i q a} & \a_1 + \a_2
              \end{pmatrix}
\end{equation}
and
\begin{equation}
     \widetilde \k (q)
          = \begin{pmatrix}
               \frac{\a_1 + \a_2}{\m_1} &
	       \frac{- \a_1 - \a_2 \re^{- \i q a}}{\sqrt{\m_1 \m_2}} \\[1ex]
	       \frac{- \a_1 - \a_2 \re^{\i q a}}{\sqrt{\m_1 \m_2}} &
	       \frac{\a_1 + \a_2}{\m_2}
            \end{pmatrix} \epp
\end{equation}
As it should be, this is a Hermitian $2 \times 2$ matrix.
We have to calculate its eigenvalues and eigenvectors.

For the eigenvalues $\la_\pm$ of a $2 \times 2$ matrix $A$
we have the general formula (exercise: check it!)
\begin{equation}
      \la_\pm = \frac{\tr A}{2} \pm \sqrt{\biggl(\frac{\tr A}{2}\biggr)^2 - \det A} \epp
\end{equation}
For the matrix $\widetilde \k (\qv)$ we calculate
\begin{subequations}
\begin{align}
     \tr \widetilde \k (q) & = \frac{2 (\a_1 + \a_2)}{\m_1 \m_2} \epc \\[1ex]
     \det \widetilde \k (q) &
        = \frac{\a_1 \a_2}{\m_1 \m_2} 4 \sin^2 \Bigl(\frac{qa}{2}\Bigr) \epp
\end{align}
\end{subequations}
It follows that
\begin{multline} \label{dispdiat}
     \om_\pm^2 (q) = \frac{\a_1 + \a_2}{\m_1 \m_2} \pm
        \sqrt{\biggl(\frac{\a_1 + \a_2}{\m_1 \m_2}\biggr)^2
	      - \frac{\a_1 \a_2}{\m_1 \m_2} 4 \sin^2 \Bigl(\frac{qa}{2}\Bigr)} \\[1ex]
        = \frac{(\a_1 + \a_2)(m_1 + m_2)^2}{4 m_1 m_2}
	  \Biggl\{1 \pm \sqrt{1 - \frac{4 \a_1 \a_2}{(\a_2 + \a_2)^2}
	                          \frac{4 m_1 m_2}{(m_1 + m_2)^2}
				  \sin^2 \Bigl(\frac{qa}{2}\Bigr)} \Biggr\} \epp
\end{multline}
Recalling that $4 x y \le (x + y)^2\ \Leftrightarrow\ 0 \le (x - y)^2$
for all $x, y \in {\mathbb R}$ we may conclude that
\begin{equation}
     0 < \frac{4 \a_1 \a_2}{(\a_2 + \a_2)^2}\frac{4 m_1 m_2}{(m_1 + m_2)^2} \le 1
\end{equation}
as it must be for $\om_\pm^2 (q)$ to be real.
\begin{figure}
\begin{center}
\includegraphics[width=.75\textwidth]{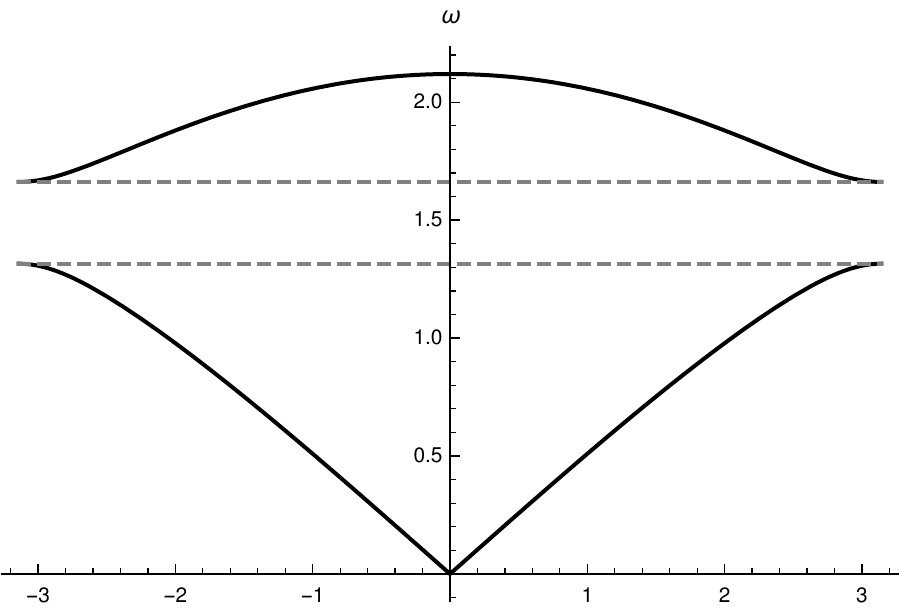}
\raisebox{70pt}{$\begin{array}{l} \om_+ \\[15pt] \om_- \\[92pt] q \end{array}$}
\caption{\label{fig:dispersion_diatomic} Dispersion relation (\ref{dispdiat})
of the diatomic chain. Upper branch $\om_+ (q)$, lower branch $\om_- (q)$.
Parameters $m_1 = 1,5$, $m_2 = 2,0$, $\a_1 = 1,3$, $\a_2 = 0,9$, $a = 1$.
The lower branch is called the acoustic branch, the upper branch is called
the optical branch. The branches are separated by the `band gap' which is the
distance between the dashed lines.
}
\end{center}
\end{figure}

The two branches $\om_\pm$ of the dispersion relation are sketched in
Figure~\ref{fig:dispersion_diatomic}. Note that the lower branch is
going to zero linearly as $q$ goes to zero,
\begin{equation}
     \om_- (q) = v_s q + {\cal O}(q^2) \epc \qd
          v_s = \sqrt{\frac{\a_1 \a_2}{2(\a_1 + \a_2)}} \, a \epp
\end{equation}
This is the limit of long wave lengths. For this reason the branch
is called the acoustic branch, $v_s$ is the sound velocity.

The branch $\om_+$ is called the optical branch. The normal modes at
small $q$ in the acoustic branch correspond to motions when the
two atoms in the unit cell move in phase, while small values of
$q$ in the optical branch correspond to motions when the
two atoms move against each other. This can be seen by looking at
the polarization vectors, which we leave as an exercise. As we
see in Figure~\ref{fig:dispersion_diatomic} the frequencies in
the optical branch are higher that in the acoustic branch. In real
solids typical optical branches correspond to frequencies in the
infrared.

The numbers
\begin{equation}
     W_\pm = \max \om_\pm (q) - \min \om_\pm (q)
\end{equation}
are called the band width of the optical and acoustic branches.
They quantify the ranges of available frequencies or bands.
In real solids, like in our Figure~\ref{fig:dispersion_diatomic},
optical bands have typically smaller band widths than acoustic
bands. The number
\begin{equation}
     g = \min \om_+ (q) - \max \om_- (q)
\end{equation}
is called the band gap. It corresponds to a range of `forbidden
frequencies'. As we shall see the existence of bands and
band gaps explain many of the characteristic feature of
solids observed in experiments.

It is interesting to see, how the monoatomic chain of example 1
is recovered for $\a_1 = \a_2 = \a$ and $m_1 = m_2$. For this special
choice of parameters
\begin{equation}
     \om_\pm^2 (q) = 2 \a \bigl(1 \pm \cos\bigl(\tst{\frac{qa}{2}}\bigr)\bigr) \epc
\end{equation}
since $|q| \le \p/a$. This can be rewritten as
\begin{subequations}
\begin{align}
     \om_-^2 (q) & = 4 \a \sin^2 \bigl(\tst{\frac{qa}{4}}\bigr) \epc \\[1ex]
     \om_+^2 (q) & = 4 \a \cos^2 \bigl(\tst{\frac{qa}{4}}\bigr)
        = 4 \a \sin^2 \bigl( (q \pm \tst{\frac{2\p}{a}}) \tst{\frac{a}{4}}\bigr) \epp
\end{align}
\end{subequations}
We see that $\om_+$ is the same function as $\om_-$, shifted
by $\pm \frac{2\p}{a}$. The band gaps have vanished, and by
shifting $\om_+$ on the interval $[-\p/a,0]$ by $2\p/a$ to the
right and the same function on the interval $[0,\p/a]$ by
$2\p/a$ to the left, we obtain the function $\om_-$
on the doubled Brillouin zone $[-2\p/a,2\p/a]$. Thus, we
have two equivalent descriptions, two branches of the dispersion
relation on the original Brillouin zone $[-\p/a,\p/a]$ or
one branch on $[-2\p/a,2\p/a]$. The doubling of the Brillouin
zone corresponds to a bisection of the unit cell in the
Bravais lattice (exercise: draw the pictures!).

\subsection{General properties of phonons in a 3d lattice}
\begin{enumerate}
\item
There are $3 N$ branches $\om_\a (\qv)$ of the dispersion relation
for a crystal lattice with $N$ atoms per unit cell.
\item
In the thermodynamic limit the lattice momenta $\qv$ densely
fill the Brillouin zone, and the matrix $\widetilde \k (\qv)$
as defined in equation (\ref{fmfourier}) becomes a continuously
differentiable function of $\qv$. Then, due to the implicit function
theorem, the functions $\om_\a^2 (\qv)$ become differentiable
functions of $\qv$. As also follows from (\ref{fmfourier})
they are naturally extended as periodic functions on the
reciprocal lattice $\overline B$.
\item
There are precisely three acoustic branches for which $\om_\a (0) = 0$.
All other branches are optical branches with $\min \om_\a > 0$.
Since any continuous functions assumes its extremum on compact
sets, every phonon band has finite band width.
\item
Of the three acoustic branches one has longitudinal, the others
have transversal polarization. In general the longitudinal
acoustic modes are faster than the transversal acoustic modes
(reversing force is larger for pressure waves than for shear waves,
becomes clear when thinking about transition to fluid).
\item
The dispersion relations of the phonons are invariant under the
action of the point group $R_0$ of the crystal,
\begin{equation}
     \om_\a (\qv) = \om_\a (G \qv)
\end{equation}
for all $(\qv, \a) \in Q$, for all $G \in R_0$.
\end{enumerate}
\begin{proof}
Let $G \in R_0 \subset O(3)$. Then $G$ acts naturally,
as a rotation or a rotation followed by an inversion, on the
Bravais lattice and on the displacements $u^{(r,j)} (\Rv)$,
$j = x, y, z$, of the individual ions from their equilibrium
positions. The latter action combines into the action
of a representation $D$ of $R_0$ on the vectors of displacement
$\uv$,
\begin{equation}
     {u'}^\a (G \Rv) = \sum_{\be \in I} D^\a_\be (G) u^\be (\Rv) \epp
\end{equation}
This transformation leaves the kinetic energy and the potential
energy of the Hamiltonian (\ref{phham}) of the harmonic crystal
separately invariant, as all ions simultaneously undergo the same
$O(3)$ transformation, which does not affect their relative
displacements. The crystal lattice is not necessarily invariant
under this transformation, but the effect on the crystal lattice
is at most a translation, since any point group operation can be
seen as a combination of a space group operation and a translation.
This is another way of understanding the point group invariance of
the Hamiltonian (\ref{phham}).

Let us work out the consequences of the invariance for the
representation $D$. First of all, setting $\Rv' = G \Rv$,
\begin{equation}
     \frac{\6}{\6 u^\a (\Rv)}
        = \sum_{\be \in I} \frac{\6 {u'}^\be (\Rv')}{\6 u^\a (\Rv)}
	                   \frac{\6}{\6 {u'}^\be (\Rv')}
        = \sum_{\be \in I} D^\be_\a (G) \frac{\6}{\6 {u'}^\be (\Rv')} \epc
\end{equation}
and hence
\begin{multline}
     T_\uv = - \2 \sum_{\Rv \in B} \sum_{\a \in I}
                  \frac{1}{\m_\a} \frac{\6^2}{\6 {u^\a (\Rv)}^2} \\[1ex]
           = - \2 \sum_{\Rv \in B}
	          \sum_{\a \in I} \frac{1}{\m_\a}
               \sum_{\be, \g \in I} D^\be_\a (G) D^\g_\a (G)
	          \frac{\6}{\6 {u'}^\be (\Rv')} \frac{\6}{\6 {u'}^\g (\Rv')} \\[1ex]
           = - \2 \sum_{\Rv \in B} \sum_{\be, \g \in I}
	          \Biggl[\sum_{\a \in I} \frac{D^\be_\a (G) D^\g_\a (G)}{\m_\a}\Biggr]
	          \frac{\6}{\6 {u'}^\be (\Rv)} \frac{\6}{\6 {u'}^\g (\Rv)} \epc
\end{multline}
where we have used the invariance of $B$ under $R_0$ in the third
equation. Form invariance of the kinetic energy now means that
\begin{equation}
     \sum_{\a \in I} \frac{D^\be_\a (G) D^\g_\a (G)}{\m_\a}
        = \frac{\de^{\be \g}}{\m_\be} \epp
\end{equation}
Similarly, the potential energy transforms like
\begin{multline}
     V(\uv) = \2 \sum_{\Rv, \Sv \in B} \sum_{\a, \be \in I}
              u^\a (\Rv) K_{\a \be} (\Rv - \Sv) u^\be (\Sv) \\[1ex]
            = \2 \sum_{\Rv, \Sv \in B} \sum_{\a, \be \in I}
              u^\a (G \Rv) K_{\a \be} \bigl(G(\Rv - \Sv)\bigr) u^\be (G \Sv) \\[1ex]
            = \2 \sum_{\Rv, \Sv \in B} \sum_{\a, \be \in I}
              {u'}^\a (\Rv)
	      \Biggl[\sum_{\g, \de \in I} D^\g_\a (G)
	             K_{\g \de} \bigl(G(\Rv - \Sv)\bigr) D^\de_\be (G)
		     \Biggr] {u'}^\be (\Sv) \epp
\end{multline}
Then form invariance of this expression implies that
\begin{equation} \label{fmpttransform}
     K_{\a \be} (\Rv - \Sv) =
        \sum_{\g, \de \in I} D^\g_\a (G)
	      K_{\g \de} \bigl(G(\Rv - \Sv)\bigr) D^\de_\be (G) \epp
\end{equation}

It follows that
\begin{multline} \label{kappapttrans}
     \sum_{\g, \de \in I} D^\g_\a (G) \k_{\g \de} (G \qv) D^\de_\be (G)
        = \sum_{\Rv \in B} \re^{- \i \<G \qv, \Rv\>} 
          \sum_{\g, \de \in I} D^\g_\a (G) K_{\g \de} (\Rv) D^\de_\be (G) \\[1ex]
        = \sum_{\Rv \in B} \re^{- \i \<G \qv, G \Rv\>} 
          \sum_{\g, \de \in I} D^\g_\a (G) K_{\g \de} (G \Rv) D^\de_\be (G)
	= \k_{\a \be} (\qv) \epp
\end{multline}
Here we have used (\ref{fmpttransform}) and the fact that $\<G \qv, G \Rv\> =
\<\qv,\Rv\>$ in the last equation. Equation (\ref{kappapttrans}) is
equivalent to saying that
\begin{equation}
     \widetilde \k_{\a \be} (\qv) = 
        \sum_{\g, \de \in I} {\cal D}^\g_\a (G)
	      \widetilde \k_{\g \de} (G \qv) {\cal D}^\de_\be (G) \epc
\end{equation}
where
\begin{equation}
     {\cal D}^\a_\be (G) = \sqrt{\frac{\m_\a}{\m_\be}} D^\a_\be (G) 
\end{equation}
and therefore
\begin{equation}
     \sum_{\be \in I} {\cal D}^\a_\be (G) {{\cal D}^t}^\be_\g (G)
        = \sum_{\be \in I} {\cal D}^\a_\be (G) {\cal D}^\g_\be (G)
        = \sqrt{\m_\a \m_\g} \sum_{\be \in I} \frac{D^\a_\be (G) D^\g_\be (G)}{\m_\be}
	= \de^{\a \g} \epp
\end{equation}
This means that $\widetilde \k (\qv)$ and $\widetilde \k (G \qv)$ are
similar matrices which implies our claim.
\end{proof}

\subsection{Exercise 7. Linear chain with a mass defect}
For equal masses the solution of the periodic chain in Exercise 6 leads to
the acoustic phonons of the simple one-dimensional lattice. If the masses
are different, however, there is, in general, no simple closed solution,
not even of the one-dimensional problem.

In the following we shall study the influence of a mass defect. We assume that
all masses but one are equal to $m$ and that the remaining mass is equal
to $m (1 + \m)$. As in Exercise 6 the ansatz of an harmonic time dependence
leads to an eigenvalue problem of the form
\begin{equation} \label{eig}
  \binom{\psi(1)}{\varphi(1)} =\binom{\psi(N+1)}{\varphi(N+1)}
     = L_N\cdots L_1 \binom{\psi(1)}{\varphi(1)} \epc \qd
       L_n=\binom{2-\Omega_n^2 \quad -\!\!1}{ 1 \qquad\quad 0} \epp
\end{equation}
Here $\psi(n)$ and $\varphi(n)$ are defined as in Exercise 6, and
$\Omega_n^2= m_n\omega^2/k$ with $m_1 = m(1+\mu)$ and $m_j = m$ for $j=2,\dots,N$. 

\begin{enumerate}
\item
Solve the eigenvalue problem \eqref{eig} and show that $\kappa$ in
$\Omega^2=4\sin^2({\kappa/2})$ satisfies the transcendental equation
\begin{equation} \label{tanGl}
  \tan{\left(\frac{N\kappa}{2}\right)}+\mu\tan{\left(\frac{\kappa}{2}\right)}=0
\end{equation}
or $\re^{\i \k} = \pm 1$ or $\re^{\i N \k} = 1$.
\item
In order to solve \eqref{tanGl} graphically we transform it into a more
convenient form. For this purpose set $z=e^{i\kappa}$. First prove that
\begin{equation}
  \frac{1}{z^N-1} = \frac{1}{\prod_{j=1}^{N}(z-z_j)}
    = \frac{1}{N}\sum_{j=1}^{N}\frac{z_j}{z-z_j}
\end{equation}
and obtain an analogous relation for $z_j^{-1}$. Here the $z_j$, $j=1,\ldots,N$,
are the $N$th roots of unity. Then show that \eqref{tanGl} is equivalent to
\begin{equation} \label{OmegaGl}
     \frac{N}{\mu}+\sum_{j=1}^{N}\frac{\Omega^2}{\Omega^2-\Omega_j^2} = 0 \epc
\end{equation}
where the $\Omega_j$ are the eigenfrequencies of the problem with equal masses.
Now discuss (\ref{OmegaGl}) graphically. Which equations for the frequencies
do you obtain for the particular cases when $\mu = -1$ or $\mu = \infty$?
\item
How many solution do you obtain from (\ref{OmegaGl})? Where are the missing
solutions and what is their interpretation?
\end{enumerate}
\subsection{Exercise 8. Mass defect in the thermodynamic limit }
In exercise 7 the consequence of a mass defect on the spectrum of the harmonic
chain with periodic boundaries was analyzed. All eigenfrequencies except the
translation mode are either determined by 
\begin{equation}\label{bestgl}
  \sum_{j=1}^N \frac{\Omega^2}{\Omega^2 - \Omega_j^2} = -\frac{N}{\mu},
\end{equation}
wherein $N$ is the number of masses and $\Omega_j^2$ are the eigenvalues of the
periodic chain without mass defect, or they agree with one of the $\Omega_j^2$.
For a negative mass defect $-1<\mu<0$ there is one mode, i.e.\ one solution
of \eqref{bestgl}, outside of the domain $[0,4]$ of the $\Omega_j^2$.
\begin{enumerate}
\item
Calculate the eigenfrequency of this mode in the thermodynamic limit
$N\to\infty$ directly from \eqref{bestgl}.
\item
Consider the infinite harmonic chain with $m_j=m,j\neq 0$ and $m_0=m(1+\mu)$. The
deviations from the equilibrium positions of the masses are denoted $x_n$. Use the ansatz
\begin{equation}
x_n(t)=\exp(-q\left|n\right|-i\omega t)
\end{equation}  
and determine how $q$ and $\omega$ have to be chosen, to get a solution of the equation
of motion of the chain. Why does this solution only make sense for negative mass defects?
\end{enumerate}
\subsection{Exercise 9. Harmonic oscillations of a two-dimensional lattice}
Consider a two-dimensional square lattice composed of identical ions of mass $m$
under periodic boundary condition. Every ion interacts with nearest and next-nearest
neighbours. The spring constants of the harmonic potential are given by $\beta_1$
for nearest neighbours and by $\beta_2$ for next-nearest neighbours. All other
interactions are assumed to be negligible. Furthermore all motions of the ions are
confined to the lattice plane. Set up the force-matrix $K_{\alpha \beta}(\Rv,\Sv)$
and compute $\kappa(\qv)$. Diagonalize $\kappa(\qv)$. How does the frequency depend
on the wave vector $\qv$? Plot the dispersion relation in $(q,0)$- and in
$(q,q)$-direction.%

\mysection{Statistical mechanics of the harmonic crystal}
\subsection{Partition function and free energy}
The thermodynamic properties of the harmonic crystal are completely
determined by the canonical partition function
\begin{equation} \label{phpartfun}
     Z_{\rm ph} = \tr \Bigl\{ \re^{- \frac{H_{\rm ph}}{T}} \Bigr\}
                = \re^{- \frac{F_{\rm ph}}{T}} \epp
\end{equation}
Here $F_{\rm ph}$ is the free energy of the harmonic crystal,
and we are using units such that the Boltzmann constant $k_B = 1$.
Inserting (\ref{phhamaadagger}) into (\ref{phpartfun}) we obtain
\begin{multline} \label{evalpf}
     Z_{\rm ph} = \prod_{(\qv, \a) \in Q} \tr \biggl\{
        \exp\biggl\{- \frac{\om_\a (\qv)}{T}
	             \Bigl({a^+}_\qv^\a \, a_\qv^\a + \tst{\2}\Bigr) \biggr\} \biggr\} \\
        = \prod_{(\qv, \a) \in Q} \biggl\{ \re^{- \frac{\om_\a (\qv)}{2T}}
	      \sum_{k=0}^\infty \re^{- \frac{\om_\a (\qv) k}{T}} \biggr\}
	= \prod_{(\qv, \a) \in Q} \frac{1}{2 \sh \bigl(\frac{\om_\a (\qv)}{2T}\bigr)} \epp
\end{multline}
It follows that
\begin{equation} \label{phfreeenergy}
     F_{\rm ph} = E_0 + T \sum_{(\qv, \a) \in Q}
                          \ln \Bigl(1 - \re^{- \frac{\om_\a (\qv)}{T}}\Bigr) \epp
\end{equation}
Here we have used (\ref{phpartfun}) and (\ref{phgsenergy}).
\subsection{The Bose-Einstein distribution}
From (\ref{evalpf}) and (\ref{phfreeenergy}) we also conclude that
\begin{multline}
     \frac{\6 F_{\rm ph}}{\6 \om_\a (\qv)} =
        - \frac{T}{Z_{\rm ph}} \frac{\6 Z_{\rm ph}}{\6 \om_\a (\qv)} = 
	\frac{\tr \Bigl\{
	          \Bigl({a^+}_\qv^\a \, a_\qv^\a + \tst{\2}\Bigr)
		  \re^{- \frac{H_{\rm ph}}{T}} \Bigr\}}
	     {\tr \Bigl\{ \re^{- \frac{H_{\rm ph}}{T}} \Bigr\}} \\[1ex]
     = \bigl\<{a^+}_\qv^\a \, a_\qv^\a + \tst{\2}\bigr\>
     = \2 + \frac{1}{\re^\frac{\om_\a (\qv)}{T} - 1} \epp
\end{multline}
Recalling that
\begin{equation}
     \hat n_\qv^\a = {a^+}_\qv^\a \, a_\qv^\a
\end{equation}
is the occupation number operator that measures the occupancy of
the mode $(\qv, \a)$, we may interpret
\begin{equation} \label{boseeinstein}
     \< \hat n_\qv^\a \> = \frac{1}{\re^\frac{\om_\a (\qv)}{T} - 1}
\end{equation}
as a function of $T$ as the average occupancy of the mode
$(\qv, \a)$ in a canonical ensemble of temperature $T$. Seen as a
functions of $\qv$ and $\a$ this functions measures the distribution
at temperature $T$ of the quanta of vibration energy over the modes.
Equation (\ref{boseeinstein}) defines the famous Bose-Einstein
distribution.

In physics we call excitations of a `quantum field', that carry energy
and momentum, particles. The particles associated with the quantized
lattice vibrations are called phonons. Within the approximation
of the harmonic crystal the phonons do not interact. According
to (\ref{boseeinstein}) they form an ideal gas of non-conserved Bosons,
very much like the photons which are the quanta of the electro-magnetic
field. If the number of phonons would be conserved, a chemical potential
that would control their number would appear in (\ref{boseeinstein}).

\subsection{The density of states}
In the theory of ideal quantum gases the free energy and all
derived thermodynamic quantities in the thermodynamic limit are usually
written as functionals of the density of states. We would like to
briefly recall its definition and its use in approximating sums
like in (\ref{phfreeenergy}) by integrals. Usually the latter is done
in a two-step procedure. In step one sums over lattice momenta are
converted into integrals. In step two integrals over momenta are
transformed into integrals over energies by means of the density
of states.

Recall that the volume of the Brillouin zone is (cf.\ equation
(\ref{volbz})) $V_R = (2 \p)^3/V_u$, where $V_u$ is the volume
of the unit cell, and that there are $L^3$ lattice momenta in
the Brillouin zone, if we introduce periodic boundary conditions
as described in section~\ref{sec:pbcs}. Then the volume of the
crystal is $V = L^3 V_u$. The volume per lattice momentum in the
Brillouin zone or volume element is
\begin{equation}
     \rd^3 q = \frac{V_R}{L^3} = \frac{(2\p)^3}{V_u L^3} = \frac{(2\p)^3}{V} \epp
\end{equation}
Thus, the free energy (\ref{phfreeenergy}) can be approximated as
\begin{equation} \label{freeenergyintbz}
     F_{\rm ph} = E_0 + \frac{VT}{(2\p)^3}
                        \sum_{\a \in I} \int_{BZ} \rd^3 q \:
                        \ln \Bigl(1 - \re^{- \frac{\om_\a (\qv)}{T}}\Bigr) \epp
\end{equation}

We would like to rewrite the integral on the right hand side
of (\ref{freeenergyintbz}) as an integral over energies. For
this purpose we define the counting function
\begin{equation} \label{countfun}
     f_\a (\om) = \frac{1}{V} \sum_{\qv \in BZ} \Th(\om - \om_\a (\qv))\
		  \xrightarrow[\scriptscriptstyle V \rightarrow \infty]{\mspace{72.mu}}\
                  \frac{1}{(2\p)^3} \int_{BZ} \rd^3 q \: \Th (\om - \om_\a (\qv)) \epc
\end{equation}
where $\a \in I$ and $\Th$ is the Heaviside step function. The density
of states of the $\a$th phonon branch is equal to the number of
states in $[\om, \om + \D \om]$ divided by $V \D \om$ or, for
$\D \om \rightarrow 0$,
\begin{equation} \label{densstatesbranch}
     g_\a (\om) = f_\a' (\om)
        = \frac{1}{(2\p)^3} \int_{BZ} \rd^3 q \: \de (\om - \om_\a (\qv)) \epp
\end{equation}
Introducing the total density of states as
\begin{equation} \label{deftotalds}
     g(\om) = \sum_{\a \in I} g_\a (\om)
\end{equation}
we can rewrite the expression (\ref{freeenergyintbz}) for the free
energy in the form
\begin{equation} \label{phfreeends}
     F_{\rm ph} = E_0 + VT \int_0^\infty \rd \om \:
                           g(\om) \ln \bigl(1 - \re^{- \frac{\om}{T}}\bigr) \epp
\end{equation}
This form will be the starting point for the discussion of
the specific heat of the harmonic crystal below.

Before entering this discussion we will have a closer look at 
the density of states $g (\om)$. The integral on the right hand
side of (\ref{densstatesbranch}) can be interpreted as
representing the `area' of a surface $S(\om)$ in reciprocal
space, where $\om$ is implicitly defined by
\begin{equation}
     \s(\qv_0) = \om - \om_\a (\qv_0) = 0 \epp
\end{equation}
In the vicinity of this surface we introduce local coordinates
\begin{equation}
     \D \qv = \nv \D q_\perp + \tv_1 \D q_{1 \parallel} + \tv_2 \D q_{2 \parallel} \epc
\end{equation}
where $\nv$ is a unit vector normal to the surface and $\tv_j
\D q_{1 \parallel}$, $j = 1, 2$, are parallel unit vectors. Then
\begin{equation}
     \s (\qv_0 + \D \qv) = - \<\grad \om_\a (\qv_0), \D \qv\> + \dots
        = - \|\grad \om_\a (\qv_0)\| \D q_\perp + \dots
\end{equation}
and thus
\begin{multline}
     g_\a (\om) = \frac{1}{(2\p)^3}
                  \int_{S(\om)} \rd S \int \rd \D q_\perp
		                \de(- \|\grad \om_\a (\qv_0)\| \D q_\perp + \dots) \\[1ex]
                = \frac{1}{(2\p)^3}
		  \int_{S(\om)} \frac{\rd S}{\|\grad \om_\a (\qv_0)\|} \epp
\end{multline}
This formula can be used to actually calculate the density
of states, given the dispersion relations of the phonons.
\subsection{Van-Hove singularities}
\label{sec:vanh}
The density of states $g(\om)$ exhibits singularities where
$\grad \om = 0$. These are called van-Hove singularities.
Their character depends on the space dimension and is
different for 3d, 2d and 1d lattices. We can understand them
by expanding $\om$ in the vicinity of a critical point $\kv_0$,
$\grad \om (\kv_0) = 0$, where
\begin{equation}
     \om = \om_0 + \2 \bigl\< (\kv - \kv_0), \G (\kv - \kv_0)\bigr\>
\end{equation}
and $\G$ is the matrix of the second derivatives of $\om$ at $\kv_0$.
Setting $\xv = \kv - \kv_0$ we have $\grad \om = \G (\kv - \kv_0)
= \G \xv$. Hence, for values of $\om$ close to a critical
point,
\begin{equation}
     g_\a (\om) = \frac{1}{(2\p)^3}
		  \int_{\<\xv,\G \xv\> = 2 (\om - \om_0)}
		  \frac{\rd S}{\|\G \xv\|} \epp
\end{equation}
The matrix $\G$ can be diagonalized by an orthogonal transformation.
Such a transformation leaves the surface element $\rd S$ invariant.
Hence, we may assume that $\G = \diag(\g_1, \g_2, \g_3)$. Thus,
\begin{equation}
     g_\a (\om) = \frac{1}{(2\p)^3}
		  \int_{\substack{\g_1 x_1^2 + \g_2 x_2^2 + \g_3 x_3^2 \\ = 2 (\om - \om_0)}}
		  \frac{\rd S}{\sqrt{\g_1^2 x_1^2 + \g_2^2 x_2^2 + \g_3^2 x_3^2}} \epp
\end{equation}

When discussing this integral, we have distinguish several cases.
The critical point may be a minimum, a maximum or a saddle point
depending on the signature of $\G$ which is defined as
\begin{equation}
     \sign \G = \diag (\sign \g_1, \sign \g_2, \sign \g_3) \epp
\end{equation}
We denote the four different cases by $m = (+, +, +)$, $M = (-, -, -)$,
$S_1 = (+, +, -)$ and $S_2 = (+, -, -)$. Clearly $m$ corresponds
to a minimum, $M$ to a maximum, and $S_1$, $S_2$ to two kinds of
saddle points. In 3d the singularities are of square root type in
all four cases. The details will be worked out in
exercise 10.
\begin{theorem}
Van Hove \cite{vanHove53}. In 3d the density of states $g$ has at
least one singularity of type $S_1$, and one of type $S_2$. The
derivative at the upper edge of the spectrum is $- \infty$.
\end{theorem}
\subsection{Exercise 10. Van-Hove singularities}
The density of states per unit volume of the $\a$th phonon branch in a crystal
lattice in $d$ dimensions is given by 
\begin{equation}
  g_\a (\omega) = \frac{1}{(2\pi)^d} \int_{S_\a (\omega)}
                \frac{\rd S}{\|\grad \om_\a (\kv)\|} \epc
\end{equation}
where $S_\a (\om)$ is the surface $\om_\a = \text{const.}$ in the reciprocal
space. The total density of states per unit volume is the sum over all
branches.

Determine the four distinct types of singularities of the density of states
for space dimension $d=3$. For this purpose expand $\omega_\a$ close to
a critical point $\omega_0$,
\begin{equation}
    \om_\a - \om_0 = \g_1 x_1^2 + \g_2 x_2^2 + \g_3 x_3^2 \epc
\end{equation}
and discuss the four different cases associated with different choices
of the relative sign of the coefficients $\g_1$, $\g_2$ and $\g_3$.
Hint: the saddle point cases require the introduction of a cut-off.
\enlargethispage{1ex}

\subsection{Exercise 11. Density of states in linear chains}
The dimensionless dispersion relations of the monoatomic linear chain and of the linear
chain with alternating masses with $0\leq \kappa < 2\pi$ and ratio of masses
$\mu=m/M$ are given by
\begin{align}
  \Omega^2(\kappa) & = 4\sin^2{(\kappa/2)} \qquad \text{and}\\
  \Omega^2(\kappa) & = 1+\mu\pm\sqrt{(1+\mu)^2-4\mu\sin^2{(\kappa/2)}} \epp
\end{align}
Calculate and sketch the densities of states. Which types of singularities do appear?

\mysection{Specific heat of the harmonic crystal}
Given the free energy of the phonons as a functional of the density
of states (\ref{phfreeends}) we can calculate their internal
energy,
\begin{multline} \label{phinten}
     E_{\rm ph} = F_{\rm ph} + T S_{\rm ph}
        = F_{\rm ph} - T \frac{\6 F_{\rm ph}}{\6 T} \\[1ex]
	= E_0 - V T^2 \frac{\6}{\6 T}
	        \int_0^\infty \rd \om \: g(\om) \ln \bigl(1 - \re^{- \frac{\om}{T}}\bigr)
	= E_0 + V \int_0^\infty \rd \om \: \frac{g(\om) \om}{\re^{\frac{\om}{T}} - 1} \epp
\end{multline}
Here $S_{\rm ph}$ in the first equation is the entropy of the
harmonic lattice vibrations. The integrand on the right hand side
of the last equation has a clear interpretation as the ``density of
states $g(\om)$ $\times$ energy $\om$ $\times$ thermal occupation.''

The quantity that is measured in experiments is the specific heat
\begin{equation}
     C_V = \frac{\6 E_{\rm ph}}{\6 T}
         = V \int_0^\infty \rd \om \: g(\om)
	     \biggl(\frac{\om/2T}{\sh(\om/2T)}\biggr)^2 \epp
\end{equation}
It is a linear functional of the density of states. Due to
(\ref{deftotalds}) the specific heat of the phonons is the
sum of the contributions from all branches of the dispersion
relation. From any model for the force matrix we can calculate
the dispersion $\om_\a$, then the density of states $g_\a$
and finally the specific heat by means of the above equation.

\subsection{Low-temperature specific heat}
Many bulk characteristic properties of solids at low and high
temperatures are rather universal. A prime example, which we
shall consider now, is provided by the contribution of the
phonons to the specific heat. In order to understand its low-$T$
behaviour we use (\ref{phinten}) to present it as
\begin{equation} \label{cvforlow}
     C_V = - V  \frac{\6}{\6 T} T^2 \frac{\6}{\6 T}
	        \int_0^\infty \rd \om \: g(\om) \ln \bigl(1 - \re^{- \frac{\om}{T}}\bigr) \epp
\end{equation}
The density of states has a Taylor series expansion around
$\om = 0$. Let $\om_0 > 0$ be its radius of convergence and
fix $\de$ such that $0 < \de < \om_0$. $\then$
\begin{multline} \label{bosesommerfeld}
     \int_0^\infty \rd \om \: g(\om) \ln \bigl(1 - \re^{- \frac{\om}{T}}\bigr)
        = \int_0^\de \rd \om \: g(\om) \ln \bigl(1 - \re^{- \frac{\om}{T}}\bigr)
	  + {\cal O} \bigl(T^\infty\bigr) \\[1ex]
        = \sum_{n=0}^\infty \frac{g^{(n)} (0) T^{n+1}}{n!}
	     \int_0^{\de/T} \rd x \: x^n \ln\bigl(1 - \re^{-x}\bigr)
	  + {\cal O} \bigl(T^\infty\bigr) \epp
\end{multline}
Using the Taylor expansion of the logarithm we can estimate the
integral on the right hand side as
\begin{multline} \label{zetagamma}
     \int_0^{\de/T} \rd x \: x^n \ln\bigl(1 - \re^{-x}\bigr)
        = \int_0^\infty \rd x \: x^n \ln\bigl(1 - \re^{-x}\bigr)
	  + {\cal O} \bigl(T^\infty\bigr) \\[1ex]
        = - \sum_{k=1}^\infty \frac{1}{k} \int_0^\infty \rd x \: x^n \re^{-kx}
	  + {\cal O} \bigl(T^\infty\bigr)
        = - \sum_{k=1}^\infty \frac{1}{k^{n+2}} \int_0^\infty \rd x \: x^n \re^{-x}
	  + {\cal O} \bigl(T^\infty\bigr) \\[1ex]
        = - \z(n+2) \G(n+1) + {\cal O} \bigl(T^\infty\bigr) \epc
\end{multline}
where $\G$ is the gamma function and $\z$ Riemann's zeta function.
Inserting (\ref{zetagamma}) and (\ref{bosesommerfeld}) into
(\ref{cvforlow}) we obtain the low-$T$ expansion of the specific heat,
\begin{equation} \label{phcvlowt}
     C_V = V \sum_{n=0}^\infty g^{(n)} (0) (n + 1)(n + 2) \z(n + 2) T^{n+1} \epc
\end{equation}
which holds up to exponentially small corrections in the temperature.

Equation (\ref{phcvlowt}) holds separately for every phonon branch. If we
replace $g$ by $g_\a$, we obtain the contribution $C_{V, \a}$ of
the phonon branch number $\a$ to the specific heat. For every optical
branch the density of states at $\om = 0$ is identically zero and
so are all the coefficients in its Taylor expansion around this point.
Thus,
\begin{equation}
     C_{V, \a} = {\cal O} \bigl(T^\infty\bigr)
\end{equation}
for every optical branch. In other words, the low-$T$ specific heat of
the phonons is entirely determined by the three acoustic branches.

Consider an acoustic phonon branch with (isotropic) sound velocity $v$.
Then
\begin{equation}
     \om (\qv) = v \|\qv\|
\end{equation}
for small $\qv$. The corresponding counting function (\ref{countfun})
for small $\om$ is
\begin{equation}
     f(\om) = \frac{1}{(2\p)^3} \int_{BZ} \rd^3 q \: \Th (\om - v\|\qv\|)
            = \frac{1}{2\p^2} \int_0^{\om/v} \rd q \: q^2 \epp
\end{equation}
Hence, for each acoustic branch with sound velocity $v_\a$, the density
of states at small $\om$ is
\begin{equation}
     g(\om) = \frac{\om^2}{2 \p^2 v^3} \epp
\end{equation}
According to (\ref{phcvlowt}) the contribution to the specific heat is
\begin{equation}
     C_{V, \a} = \frac{2 \p^2 V T^3}{15 v_\a^3} + {\cal O} \bigl(T^4\bigr) \epp
\end{equation}
Here we have used that $\z(4) = \p^4/90$. Summing over the three
acoustic branches we obtain the total low-$T$ specific heat of
the ions in harmonic approximation,
\begin{equation} \label{phtthreelaw}
     C_V = \frac{2 \p^2 V T^3}{5 \<v^3\>_{\rm h}}
           + {\cal O} \bigl(T^4\bigr) \epc
\end{equation}
where $\<\cdot\>_{\rm h}$ stands for the harmonic mean of the
sound velocities.

Equation (\ref{phtthreelaw}) is the $T^3$ law for the specific heat
of insulators (in conductors the electrons contribute significantly
to $C_V$). In simple solids it holds up to temperatures of about
$10 {\rm K}$. Notice that (\ref{phtthreelaw}) implies that the
low-$T$ specific heat can be determined by measuring the sound
velocities, or, taking it the other way round, that the average
sound velocity can be obtained from a specific heat measurement.

\subsection{Internal energy and specify heat at high temperatures}
Recall that, for $|x| < 2 \p$,
\begin{equation} \label{xctgx}
     \frac{x}{\re^x - 1} = 1 - \frac{x}{2}
                           + \sum_{n=1}^\infty \frac{B_{2n}}{(2n)!} x^{2n} \epc
\end{equation}
where the $B_{2n}$ are the Bernoulli numbers. This allows us to 
derive a convergent high-$T$ expansion of the internal energy of
the phonon gas. Inserting (\ref{xctgx}) into (\ref{phinten}) we obtain
\begin{multline} \label{phintenhightea}
     E_{\rm ph} =
        E_0 + V T \int_0^\infty \rd \om \: g(\om) 
	           \biggl\{ 1 - \frac{\om}{2T}
		            + \sum_{n=1}^\infty \frac{B_{2n}}{(2n)!}
			    \Bigl(\frac{\om}{T}\Bigr)^{2n} \biggr\} \\[1ex]
        = E_0 + 3 N_{\rm at} \biggl\{ T - \frac{\<\om\>_g}{2}
		            + \sum_{n=1}^\infty \frac{B_{2n}}{(2n)!}
                              \frac{\<\om^{2n}\>_g}{T^{2n - 1}} \biggr\} \epp
\end{multline}
Here $N_{\rm at} = N L^3$ is the total number of ions and
$\< \cdot \>_g$ denotes the average with respect to the probability
density
\begin{equation}
     p_g (\om) = \frac{g(\om)}{\int_0^\infty \rd \om \: g(\om)} \epp
\end{equation}
Thus, $\<\om^m\>_g$ is the $m$th moment of the probability density.
In the derivation of (\ref{phintenhightea}) we have used that
\begin{equation} \label{phdsnorm}
     \int_0^\infty \rd \om \: g(\om)
        = \lim_{\om \rightarrow + \infty} \sum_{\a \in I} f_\a (\om)
	= \frac{3 N L^3}{V}
\end{equation}
which follows from (\ref{countfun}).

Equation (\ref{phintenhightea}) implies the high-$T$ series
representation
\begin{equation}
     C_V = - 3 N_{\rm at} \sum_{n=0}^\infty \frac{(2n - 1) B_{2n}}{(2n)!}
             \<\om^{2n}\>_g T^{- 2n}
	 = 3 N_{\rm at} - \frac{N_{\rm at} \<\om^2\>_g}{4 T^2}
	   + {\cal O} \bigl(T^{-4}\bigr)
\end{equation}
for the specific heat. Here we have used that $B_0 = 1$ and
$B_2 = 1/6$. The leading order contribution corresponds to
the Dulong-Petit law
\begin{equation}
     C_V = 3 N_{\rm at}
\end{equation}
indicating a constant heat capacity at high temperature, whence
the name `heat capacity'. Note that this high-temperature limit
is approached monotonously from below.

\subsection{Debye interpolation}
We have seen that the specific heat of the phonon gas shows
universal behaviour at low and high temperatures. Let us seek
for a simple model interpolating between these two universal
regimes. For this purpose we take the low-frequency form of
the density of states and cut it off at a frequency $\om_D$
in such a way that the normalization condition (\ref{phdsnorm})
is satisfied. Let us further introduce an effective sound velocity
\begin{equation}
     \overline v = \bigl\<v^3\bigr\>_{\rm h}^\frac{1}{3} \epp
\end{equation}
Then our model density of states is
\begin{equation}
     g_D (\om) = \frac{3 \om^2}{2 \p^2 \overline{v}^3} \Th(\om_D - \om) \epc
\end{equation}
where the cut-off frequency $\om_D$ is fixed by the condition
\begin{equation}
     \int_0^\infty \rd \om \: g_D (\om)
	= \frac{\om_D^3}{2 \p^2 \overline{v}^3}
	= \frac{3 N_{\rm at}}{V}
\end{equation}
implying that
\begin{equation}
     \om_D = \overline{v} \biggl(\frac{6 \p^2 N_{\rm at}}{V}\biggr)^\frac{1}{3}
           = N^\frac{1}{3} \overline v \, r_R \epc
\end{equation}
where
\begin{equation}
     r_R = \biggl(\frac{V_R}{4 \p/3}\biggr)^\frac{1}{3}
\end{equation}
is a reciprocal length parameter sometimes called the `radius
of the Brillouin zone'. $g_D$ is called Debye density of
states and the frequency $\om_D$ the Debye frequency. These
notions go back to the Dutch-American noble laureate Peter
Debye \cite{Debye12}.

The internal energy for the Debye density of states is
\begin{equation}
     E_{\rm ph} = E_0 + \frac{9 N_{\rm at} T^4}{\om_D^3}
                        \int_0^\frac{\om_D}{T} \rd \om \: \frac{\om^3}{\re^\om - 1} \epp
\end{equation}
Defining the Debye function
\begin{equation}
     D(x) = \frac{3}{x^3} \int_0^x \rd t \: \frac{t^3}{\re^t - 1}
\end{equation}
we may recast the internal energy as
\begin{equation}
     E_{\rm ph} = E_0 + 3 N_{\rm at} T D(\om_D/T) \epc
\end{equation}
while the corresponding specific heat takes the form
\begin{equation}
     C_V = 3 N_{\rm at}
           \biggl\{D \Bigl(\frac{\om_D}{T}\Bigr)
                           - \frac{\om_D}{T} D' \Bigl(\frac{\om_D}{T}\Bigr) \biggr\} \epp
\end{equation}
\begin{figure}
\begin{center}
\includegraphics[width=.62\textwidth]{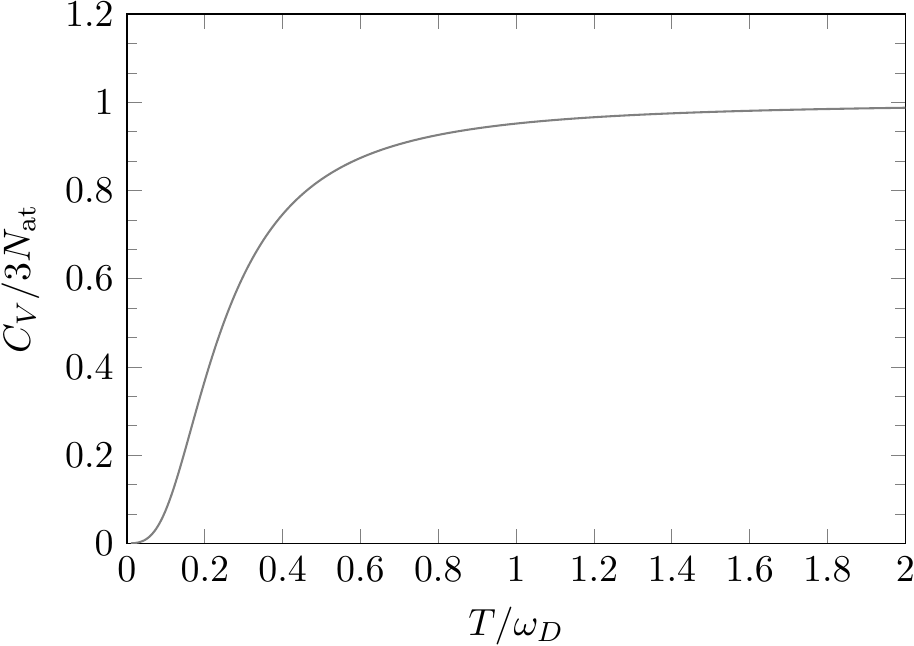}
\caption{\label{fig:cv_debye} Debye model of the specific heat of
the phonon gas in solids.
}
\end{center}
\end{figure}%
This is called the Debye formula or the Debye interpolation
formula. According to this formula the specific heat of the
phonon gas is a universal function of $\om_D/T$ (see
Figure~\ref{fig:cv_debye}). With respect to the phonon gas
contribution to the specific heat different solids are
distinguished by a single parameter, the Debye frequency.
Because it is simple and at the same time covers the basic
features of the temperature dependence of the specific
heat, the Debye model is prevailing in experimental solid
state physics. 
\subsection{The anomaly of the harmonic crystal}
Usually the basic tasks in statistical physics are to derive
the specific heat and the equation of state, expressing the
pressure $p = - \6 F/\6 V$ as a function of $T$ and $V$. For
the harmonic crystal one can find the statement that the
pressure is temperature independent,
\begin{equation} \label{pdoesnotdependont}
     \Bigl(\frac{\6 p}{\6 T}\Bigr)_V = 0 \epc
\end{equation}
in several places in the textbook literature. This should be
related to the scale invariance of the equations of motion in
harmonic potentials, but at the moment I do not know of any
really convincing derivation in the general case.

Equation (\ref{pdoesnotdependont}) implies several thermodynamic
anomalies. Since
\begin{equation}
     \Bigl(\frac{\6 V}{\6 T}\Bigr)_p
        = - \frac{\bigl(\frac{\6 p}{\6 T}\bigr)_V}
	         {\bigl(\frac{\6 p}{\6 V}\bigr)_T} = 0 \epc
\end{equation}
it follows, for instance, that the thermal expansion coefficient
\begin{equation}
     \a = \frac{1}{V} \Bigl(\frac{\6 V}{\6 T}\Bigr)_p = 0 \epp
\end{equation}
This contradicts our experience that solids usually expand
when they are heated. We conclude that the thermal expansion of
solids is a feature that must be attributed the higher order
corrections to the harmonic crystal. Another anomaly implied by
(\ref{pdoesnotdependont}) is, for instance, $C_V = C_p$, the
coincidence of the specific heats at constant pressure and
constant volume.

The simplest model system for which we can verify
(\ref{pdoesnotdependont}) is a classical chain of particles
of equal masses $m$ that interact with their nearest neighbours
through a pair potential $V(r)$, where $r$ is the distance between
the particles. Let us consider $N+1$ such particles with
coordinates $x_n$, $n = 0, \dots, N$, and nearest-neighbour
distances $r_n = x_n - x_{n-1}$, $n = 1, \dots, N$. Then the
Hamiltonian of the system is
\begin{equation} \label{hamnnsprings}
     H = \sum_{n=0}^N \frac{p_n^2}{2m} + \sum_{n=1}^N V(r_n) \epc
\end{equation}
where the $p_n$ are the momenta canonically conjugate to the $x_n$.
Note that this system has no periodic boundary conditions.
Instead of applying periodic boundary conditions for the
interaction we apply an external mechanic pressure to control
the length of the system. This can be achieved by adding a
term $(x_N - x_0)p$ to the Hamiltonian (\ref{hamnnsprings}).
As a configuration space $K$ for the particles we consider a ring
of finite length $L$ on which the `springs' connecting
neighbouring particles can be arbitrarily expanded by moving
them several times relative to each other around the ring, but
the center of mass of all particles is confined to the ring.
This construction is necessary in order to regularize the
integral
\begin{equation} \label{znnchain}
     Z = \frac{1}{(N+1)!}
         \int_{{\mathbb R}^{N+1}} \frac{\rd^{N+1} p}{(2\p \hbar)^{N+1}}
         \int_K \rd^{N+1} x \:
	 \exp \biggl\{- \frac{1}{T} \bigl(H + (x_N - x_0) p\bigr) \biggr\}
\end{equation}
representing the classical partition function. It would be
otherwise divergent. By construction
\begin{equation}
     - T \6_p \ln Z = \< x_N - x_0\> = \ell
\end{equation}
is the average length of the system at given $T$ and $p$ which
justifies the interpretation of $p$ as the pressure.

It is not difficult to calculate the integral on the right hand
side of (\ref{znnchain}). The momentum integration reduces to
Gaussian integrals, whereas the remaining integrals over the
configuration space can be dealt with after introducing Jacobi
coordinates
\begin{equation}
     s = \frac{1}{N+1} \sum_{n=0}^N x_n \epc \qd r_n = x_n - x_{n-1} \epc
         \qd n = 1, \dots, N \epp
\end{equation}
This transformation is linear, and it is easy to see that
its Jacobi matrix equals one. Taking this into account we obtain
\begin{equation}
     Z =  \frac{L}{(N+1)!} \biggl(\frac{mT}{2 \p \hbar^2}\biggr)^\frac{N+1}{2}
          \biggl(\int_{- \infty}^\infty \rd r \: \re^{- (V(r) + pr)/T} \biggr)^N \epp
\end{equation}
It follows that
\begin{equation}
     \ell = - N T \6_p
        \ln \biggl(\int_{- \infty}^\infty \rd r \: \re^{- (V(r) + pr)/T} \biggr) \epp
\end{equation}

In general the latter expression does depend on $T$, but if we
insert the harmonic potential
\begin{equation}
     V(r) = \frac{m \om_0^2}{2} (r - a)^2 \epc
\end{equation}
where $a$ is the `rest length' of the potential, then
\begin{equation}
     \ell = N \biggl(a - \frac{p}{m \om_0^2}\biggr)
\end{equation}
which we interpret as the thermodynamic length of the system.
If we now solve for the pressure $p$, we end up with
\begin{equation}
      p = m \om_0^2 \biggl(a - \frac{\ell}{N}\biggr)
\end{equation}
which, indeed, does not depend on $T$. Note that, unlike the
pressure of gases, the pressure here can have either sign,
a negative pressure contracting the chain, if it is expanded
over its equilibrium length.

\mysection{Neutron scattering}
\subsection{Thermal neutrons}
Neutron scattering experiments in solids are performed with
so-called thermal neutrons. These are neutrons with energies
between $5$ - $10\, {\rm meV}$ ($T \simeq 60$ - $1000 {\rm K}$,
$\la \simeq 0,4$ - $0,1\, {\rm nm}$). Thermal neutrons are one of
the most important probes for studying the properties of
solids and fluids. This has several reasons.
\begin{enumerate}
\item
The neutron is electrically neutral. It penetrates deeply into
the solid, can come close to the atomics nuclei and is
scattered by nuclear forces.
\item
Because of their relatively large mass the de Broglie wave length
of thermal neutrons is of the order of magnitude of atomic
distances in solids and fluids.
\item
Their energy is of the order of magnitude of the energy of
elementary excitations in solids. If a neutron is
scattered inelastically its relative change of energy is
generally large enough to be resolved experimentally. Hence,
neutrons do not only resolve the structure of solids, but
can be also used to measure their excitation spectra, e.g.\
the dispersion relations of phonons.
\item
Neutrons carry a magnetic moment which is sensitive against
intra-atomic magnetic fields. Therefore they can be used to
measure short-wavelength magnetic structures (antiferromagnetism!)
and magnetic excitations.
\end{enumerate}
The traditional source of neutrons are nuclear reactors.
Note that neutrons which come out of a nuclear reactor have
a Maxwell velocity distribution leading to a neutron current
per unit time for neutrons with velocity $v$ of
\begin{equation} \label{neutroncurrent}
     \Ph (v) \rd v = \frac{\r}{\sqrt{\p}} \Bigl(\frac{m}{2T}\Bigr)^\frac{3}{2}
                     v^3 \re^{- \frac{m v^2}{2T}} \rd v \epc
\end{equation}
where $m$ is the neutron mass and $\r$ their density
(exercise: derive (\ref{neutroncurrent})).
 
\subsection{Scattering experiments and cross sections}
\begin{figure}
\begin{center}
\includegraphics[width=\textwidth]{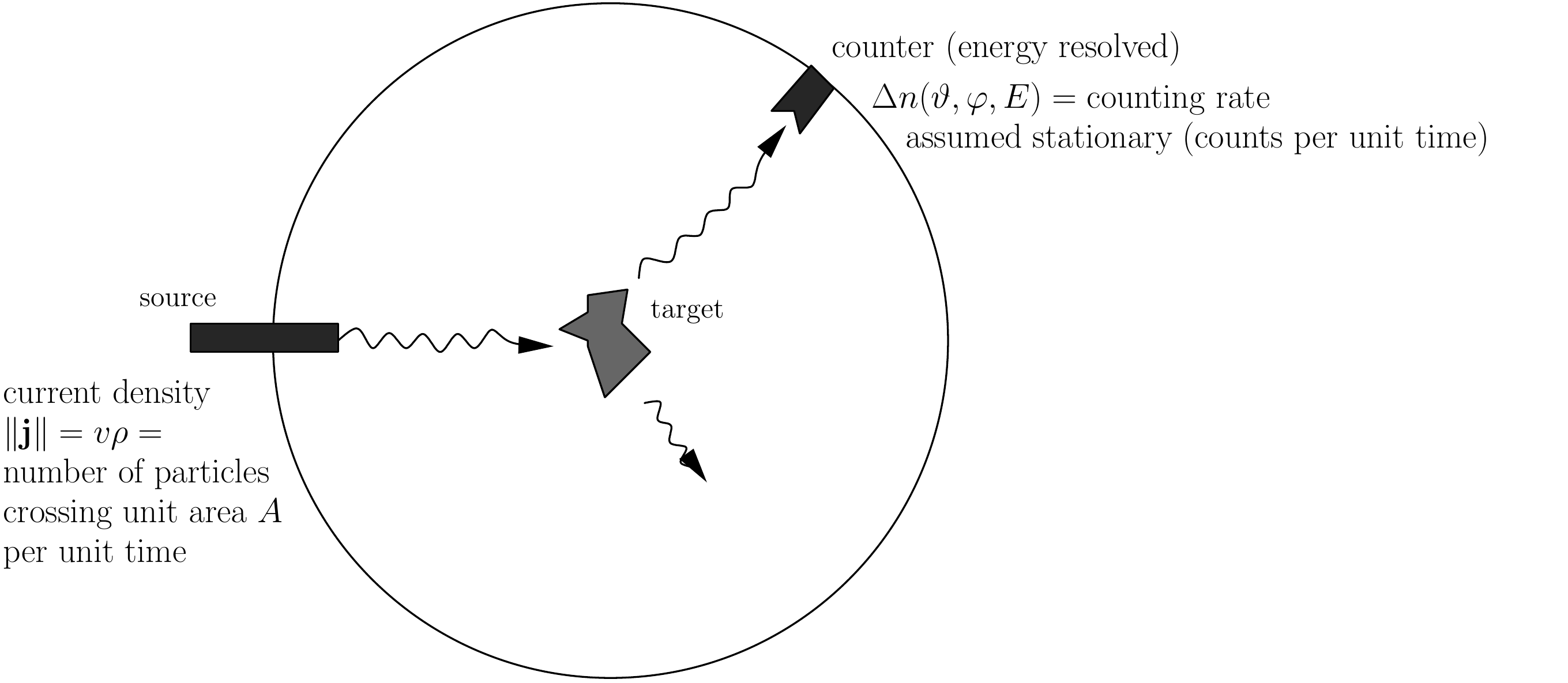}
\caption{\label{fig:cross_section} 
Schematic view of a typical scattering experiment. An incident
particle current hits a target and is scattered in all directions.
The currents of the scattered particles through a sphere centered 
around the target are measured. They determine the differential
cross section.
}
\end{center}
\end{figure}
The typical geometry of a scattering experiment is sketched in
Figure~\ref{fig:cross_section}. A stationary current of incident
particles impinges on a target and is scattered. Detectors in
equal distance from the target measure the current scattered in
direction $(\dh, \ph)$, where $\dh$, $\ph$ are spherical coordinates.
We denote the counting rate at $(\dh, \ph)$ for particles with
energies in an infinitesimal interval around $E$ by $\D n(\dh, \ph, E)$.
In a stationary situation the counting rate is expected to
be proportional to the current density $\|\jv\|$ of the incident
current $\jv$. Hence, the counting rate normalized by the current
density of the incident current,
\begin{equation}
     \D \s = \frac{\D n}{\|\jv\|} \epc
\end{equation}
is expected to be independent of the current and to characterize
the target. If we further divide by the solid angle
$\D \Om = \sin (\dh) \D \dh \D \ph$ and by the width $\D E$
of the energy interval, we obtain a quantity called the 
differential cross section,
\begin{equation} \label{defdiffcross}
     \frac{\D \s}{\D \Om \D E} = \frac{\D n}{{\D \Om \D E} \|\jv\|} \epp
\end{equation}
In the limit of infinite energy and angle resolution this function
is denoted
\begin{equation}
     \frac{\rd^2 \s}{\rd \Om \rd E} (\dh, \ph, E) \epp
\end{equation}

\subsection{Cross section and transition rate}
We shall assume that the incident neutron beam is monochromatic
and coherent. Quantum mechanically it is then represented by a
plane wave
\begin{equation}
     \Ps (\xv) = \frac{\re^{\i \<\kv, \xv\>}}{\sqrt{V}}
\end{equation}
with wave vector $\kv$ and `normalization volume $V$' (we shall use
units such that $\hbar = 1$). The corresponding current is
\begin{equation} \label{curdens}
     \jv (\xv) = - \frac{\i}{2m}
                   \bigl(\Ps^* (\xv) \6_\xv \Ps (\xv)
		         - \Ps(\xv) \6_\xv \Ps^* (\xv)\bigr)
               = \frac{\kv}{m V} \epp
\end{equation}

The interaction of the neutrons with the target determines the
transition rate $P(\kv, \kv')$ for a transition from an incoming
wave with wave vector $\kv$ to a scattered wave with wave vector
$\kv'$ under the influence of the perturbation caused by the target
(recall the time-dependent perturbation theory from the quantum
mechanics lecture). Since $\frac{\rd^3 k'}{(2\p)^3/V}$ is the
number of states in the volume $\rd^3 k'$ around $\kv'$ we can
express the counting rate as
\begin{equation} \label{counttrans}
     \rd n = P(\kv, \kv') \frac{\rd^3 k'}{(2\p)^3/V}
           = P(\kv, \kv') \frac{V k'^2 \rd k' \rd \Om}{(2\p)^3}
           = \frac{m V}{(2\p)^3} P(\kv, \kv') k' \rd E \rd \Om \epp
\end{equation}
Using (\ref{curdens}) and (\ref{counttrans}) in (\ref{defdiffcross})
we have expressed the differential cross section in terms of
the transition rate,
\begin{equation} \label{cstr}
     \frac{\rd^2 \s}{\rd \Om \rd E} (\dh, \ph, E)
        = \frac{k'}{k} \frac{(mV)^2}{(2 \p)^3} P(\kv, \kv') \epp
\end{equation}
This formula connects the basic measurable quantity on the left
hand side with a quantity depending on the microscopic properties
of the target on the right hand side.

Let us recall how the transition rate appears in quantum mechanics.
It is usually first encountered in the context of time-dependent
perturbation theory and Fermi's `golden rule' which states in our
case that
\begin{equation} \label{fermigold}
     P(\kv, \kv') = 2 \p \sum_f \de(\e_i - \e_f) |\<\Ps_f, U \Ps_i\>|^2
\end{equation}
is the rate for transitions $\Ps_i \rightarrow \Ps_f$, where $\Ps_i =
\Ps_\kv (\xv) \Ph_i$ is a fixed initial state and $\Ps_f =
\Ps_{\kv'} (\xv) \Ph_f$ runs through all possible final states
with fixed $\kv'$. Here $\Ph_i$ and $\Ph_f$ denote eigenstates
of the ions (recall that the neutrons interact with the ions!).
Hence, the energies of initial and final states $\e_i$ and $\e_f$
are
\begin{equation}
     \e_i = E_i + \frac{k^2}{2m} \epc \qd
     \e_f = E_f + \frac{k'^2}{2m} \epc \qd
\end{equation}
where $E_i$ and $E_f$ are the energies of the ionic states. Introducing
the notation $\om = (k^2 - k'^2)/2m$ the energy difference in
(\ref{fermigold}) takes the form
\begin{equation}
     \e_i - \e_f = E_i - E_f + \om \epp
\end{equation}
The potential $U$ that describes the interaction of the neutrons 
with the ions is of the form
\begin{equation}
     U(\xv) = \sum_{\Rv \in B} \sum_{r=1}^N
                 u \bigl(\xv - \xv_r (\Rv)\bigr) \epc
\end{equation}
where the $\xv_r (\Rv)$ are the position vectors of the ions.
In order to simplify the notation we suppress from now on
the sum over the ion positions in the unit cell. It will
always come with the sum over the Bravais lattice an can be
restored at any later stage if necessary.

\subsection{Form factor and structure factor}
Substituting the explicit form of the potential and of the
factorized wave functions we can calculate the matrix elements in
(\ref{fermigold}),
\begin{multline}
     \<\Ps_i, U \Ps_f\> =
        \frac{1}{V} \sum_{\Rv \in B} \int \rd^3 x \re^{- \i \<\qv, \xv\>}
	  \int \biggl[\prod_{\Sv \in B} \rd^3 x(\Sv)\biggr] \:
	     \Ph_i^*\, u \bigl(\xv - \xv (\Rv)\bigr)\, \Ph_f \\
        = \frac{1}{V} \sum_{\Rv \in B} \int \rd^3 x \re^{- \i \<\qv, \xv\>} u(\xv)
	  \int \biggl[\prod_{\Sv \in B} \rd^3 x(\Sv)\biggr]\:
	     \Ph_i^* \re^{- \i \<\qv, \xv(\Rv)\>} \Ph_f \epc
\end{multline}
where $\qv = \kv - \kv'$. We see that the integrals factorize into
a factor that depends on the interaction potential between the
neutrons and the ions and a factors that depends on the crystal
structure. Denoting
\begin{subequations}
\begin{align}
     & b(\qv) = \int \rd^3 x \re^{- \i \<\qv, \xv\>} u(\xv) \epc \\
     & \bigl\< \Ph_i, \re^{- \i \<\qv, \xv(\Rv)\>} \Ph_f \bigr\> =
	  \int \biggl[\prod_{\Sv \in B} \rd^3 x(\Sv)\biggr]\:
	     \Ph_i^* \re^{- \i \<\qv, \xv(\Rv)\>} \Ph_f
\end{align}
\end{subequations}
we obtain the following factorized form of the transition rate,
\begin{equation} \label{transitionrate}
     P(\kv, \kv') = \frac{2 \p}{V^2} |b(\qv)|^2 
        \sum_f \de(\e_i - \e_f)
	\Bigl|\sum_{\Rv \in B}
	      \bigl\< \Ph_i, \re^{- \i \<\qv, \xv(\Rv)\>} \Ph_f \bigr\>\Bigr|^2 \epp
\end{equation}
Here $|b(\qv)|^2$ is called the atomic form factor, since it depends
only on the interaction of the individual ions with the neutrons.
The sum over $f$ divided by $N_{\rm at}$ is called the dynamic
structure factor $S(\qv, \om)$. It contains the information about
the structure and the dynamics of the ions. Inserting the expression
(\ref{transitionrate}) for the transition rate into (\ref{cstr})
we obtain a formula for the differential cross section,
\begin{equation}
     \frac{\rd^2 \s}{\rd \Om \rd E} (\dh, \ph, E)
        = \frac{k'}{k} \frac{m^2}{(2 \p)^2} |b(\qv)|^2
	  \sum_f \de(\e_i - \e_f)
	  \Bigl|\sum_{\Rv \in B} \bigl\< \Ph_i,
	        \re^{- \i \<\qv, \xv(\Rv)\>} \Ph_f \bigr\>\Bigr|^2 \epp
\end{equation}

Note that this formula is very general.
\begin{enumerate}
\item
With slight modifications it holds for fluids as well.
\item
It solely relies on Fermi's golden rule (= time dependent perturbation
theory + Born approximation).
\end{enumerate}

\subsection{Rewriting the structure factor}
Using the `Fourier inversion formula'
\begin{equation}
     \de (\om) = \int_{-\infty}^\infty \frac{\rd t}{2 \p} \: \re^{- \i \om t}
\end{equation}
we obtain the following expression for the dynamic structure factor.
\begin{multline} \label{sqomzero}
     S(\qv, \om) =
        \frac{1}{N_{\rm at}}
	\int_{-\infty}^\infty \frac{\rd t}{2 \p} \:
	   \sum_f \re^{- \i (E_i - E_f + \om) t} \mspace{-4mu}
	   \sum_{\Rv, \Sv \in B} \mspace{-4mu}
	      \bigl\< \Ph_i, \re^{- \i \<\qv, \xv(\Sv)\>} \Ph_f \bigr\>
	      \bigl\< \Ph_f, \re^{\i \<\qv, \xv(\Rv)\>} \Ph_i \bigr\> \\[1ex]
        = \frac{1}{N_{\rm at}}
	  \int_{-\infty}^\infty \frac{\rd t}{2 \p} \:
	  \re^{- \i \om t} \sum_{\Rv, \Sv \in B}
	      \bigl\< \Ph_i, \re^{- \i \<\qv, \xv(\Sv)\>}
	              \re^{\i \<\qv, \xv(\Rv,t)\>} \Ph_i \bigr\> \epp
\end{multline}
Here we have used that
\begin{multline}
     \re^{- \i (E_i - E_f) t}
        \bigl\< \Ph_f, \re^{\i \<\qv, \xv(\Rv)\>} \Ph_i \bigr\> \\[1ex] =
        \bigl\< \Ph_f, \re^{\i H t} \re^{\i \<\qv, \xv(\Rv)\>}
	        \re^{- \i H t} \Ph_i \bigr\> =
        \bigl\< \Ph_f, \re^{\i \<\qv, \xv(\Rv, t)\>} \Ph_i \bigr\> \epc
\end{multline}
where $H$ is an effective Hamiltonian for the ion-ion interaction.
The expectation value under the sum on the right hand side of (\ref{sqomzero})
is called a dynamical two-point correlation function. The appearance
of such a correlation function is generic. As we shall see with further
examples below, most spectroscopic experiments and transport experiments in
solids measure two- or four-point correlation functions.

\mysection{Neutron scattering continued}
\subsection{Disorder}
So far we have assumed that the scattering potentials of all nuclei are
equal. This is only true if the solid (the fluid) consists of a single
isotope and if all nuclear spins are aligned (nuclear spin ferromagnet)
or zero. In reality the latter conditions are almost never satisfied.
We typically rather have to deal with mixtures of different isotopes
and with nuclear spin paramagnets. This means that we have to modify
our above derivation, replacing
\begin{equation}
     u(\xv - \xv(\Rv)) \rightarrow u_\Rv (\xv - \xv(\Rv)) \epc \qd
     b(\qv) \rightarrow b_\Rv (\qv) \epp
\end{equation}
In particular $b_\Rv (\qv)$ remains under the sum over the Bravais lattice
in the expression for the differential cross section, which now reads
\begin{multline}
     \frac{\rd^2 \s}{\rd \Om \rd E} (\dh, \ph, E) \\[1ex]
        = \frac{k'}{k} \frac{m^2}{(2 \p)^2}
	  \int_{-\infty}^\infty \frac{\rd t}{2 \p} \:
	  \re^{- \i \om t} \sum_{\Rv, \Sv \in B}
	      b_\Rv (\qv) b_{\Sv}^* (\qv)
	      \bigl\< \Ph_i, \re^{- \i \<\qv, \xv(\Sv)\>}
	              \re^{\i \<\qv, \xv(\Rv,t)\>} \Ph_i \bigr\> \epp
\end{multline}

In order to simplify this expression again, we assume that the
potentials $u_\Rv$ are randomly distributed. We shall indicate
the disorder average by brackets $\< \cdot \>_{\rm d}$. We
assume that the distribution underlying the average is such
that the mean value of $b_\Rv (\qv)$ is translation invariant
and, for this reason, use the notation
\begin{equation}
     \<b(\qv)\>_{\rm d} = \<b_\Rv (\qv)\>_{\rm d} \epp
\end{equation}
It further reasonable to assume that potentials at different
lattice sites are uncorrelated, implying that
\begin{equation}
     \<b_\Rv (\qv) b_{\Sv}^* (\qv)\>_{\rm d} = 
        \<b_\Rv (\qv)\>_{\rm d} \<b_{\Sv}^* (\qv)\>_{\rm d} = 
        |\<b(\qv)\>_{\rm d}|^2
\end{equation}
for $\Rv \ne \Sv$. Let us further introduce the averaged coherent
and incoherent atomic form factors
\begin{equation}
     \s_{\rm coh} = \frac{m^2}{\p} |\<b(\qv)\>_{\rm d}|^2 \epc \qd
     \s_{\rm inc} = \frac{m^2}{\p} \bigl(\<|b(\qv)|^2\>_{\rm d}
                                        - |\<b(\qv)\>_{\rm d}|^2 \bigr) \epp
\end{equation}
Averaging the differential cross section over the disorder then
results in
\begin{multline}
     \frac{\rd^2 \s}{\rd \Om \rd E} (\dh, \ph, E)
        = \frac{k'}{k}
	  \int_{-\infty}^\infty \frac{\rd t}{2 \p} \:
	  \re^{- \i \om t} \biggl\{ \sum_{\Rv, \Sv \in B}
	  \frac{\s_{\rm coh}}{4 \p}
	     \bigl\< \Ph_i, \re^{- \i \<\qv, \xv(\Sv)\>}
	             \re^{\i \<\qv, \xv(\Rv,t)\>} \Ph_i \bigr\> \\[1ex]
	  + \sum_{\Rv \in B}
	    \frac{\s_{\rm inc}}{4 \p}
	    \bigl\< \Ph_i, \re^{- \i \<\qv, \xv(\Rv)\>}
	            \re^{\i \<\qv, \xv(\Rv,t)\>} \Ph_i \bigr\> \biggr\} \epp
\end{multline}

So far we have assumed that the ions are in a pure initial state
$\Ph_i$. This is not realistic in experiments on a macroscopic
sample. In a macroscopic sample the ions have to be described by
a density matrix. If the crystal can exchange energy with its
environment it will be the density matrix of the canonical ensemble
in which each state $\Ph_i$ is occupied with probability $\re^{- E_i/T}/Z$,
where $Z$ is the canonical partition function. Performing the
canonical ensemble average and denoting the canonical expectation
values by $\< \cdot \>_T$ we obtain the final formula for the
theory of neutron scattering.
\subsection{Summary of the theory of neutron scattering -- general case}
First of all the dynamic structure factor at finite temperature
becomes
\begin{equation} \label{sqom}
     S(\qv, \om)
        = \frac{1}{N_{\rm at}}
	  \int_{-\infty}^\infty \frac{\rd t}{2 \p} \:
	  \re^{- \i \om t} \sum_{\Rv, \Sv \in B}
	      \bigl\<\re^{- \i \<\qv, \xv(\Sv)\>}
	              \re^{\i \<\qv, \xv(\Rv,t)\>}\bigr\>_T \epp
\end{equation}
The differential cross section splits into a coherent part and
an incoherent part,
\begin{equation}
     \frac{\rd^2 \s}{\rd \Om \rd E}
        = \frac{\rd^2 \s}{\rd \Om \rd E}\biggr|_{\rm coh}
	  + \frac{\rd^2 \s}{\rd \Om \rd E}\biggr|_{\rm inc} \epc
\end{equation}
where
\begin{subequations}
\begin{align}
     \frac{\rd^2 \s}{\rd \Om \rd E}\biggr|_{\rm coh} & =
        \frac{\s_{\rm coh}}{4 \p} \frac{k'}{k} N_{\rm at} S(\qv, \om) \epc \\[1ex]
     \frac{\rd^2 \s}{\rd \Om \rd E}\biggr|_{\rm inc} & =
        \frac{\s_{\rm inc}}{4 \p} \frac{k'}{k}
	  \int_{-\infty}^\infty \frac{\rd t}{2 \p} \:
	  \re^{- \i \om t} \sum_{\Rv \in B}
	      \bigl\<\re^{- \i \<\qv, \xv(\Rv)\>}
	              \re^{\i \<\qv, \xv(\Rv,t)\>}\bigr\>_T \epp
\end{align}
\end{subequations}

These formulae are widely used in order to analyze the data of
neutron scattering experiments. Before specializing them to
crystal structures we would like make a few general comments.
\begin{enumerate}
\item
In our derivation of (\ref{sqom}) we have used the letters $\Rv$
and $\Sv$ in $\xv (\Rv)$, $\xv(\Sv)$ as mere particle labels,
which is the reason why equation (\ref{sqom}), if properly read,
also defines the dynamic structure factor of a fluid or a glass.
When we used the notation $\Rv \in B$ we had in mind to apply
the formula to a mono-atomic (simple) lattice, but so far $B$
was rather an index set which may refer to any labeling of the
ions. Similar formulae hold, in particular, for a lattice with
basis.
\item
The factors $\s_{\rm coh}$ and $\s_{inc}$ depend, in the experimentally
relevant range of neutron wave lengths of the order of inter-atomic
distances only weakly on $\qv$, since the potentials $u_\Rv$ vary on
a scale of the size of the nuclei. For this reason $\s_{\rm coh}$ and
$\s_{\rm inc}$ can often be treated as constants.
\item
The most interesting contribution to the differential cross section
is the dynamic structure factor $S(\qv, \om)$. It is completely
determined by the properties of the sample and independent of the
properties of the neutrons.
\item
The incoherent cross section sums contributions from the same
nuclei at different times. It carries no information about
the structure of the sample. If all $b_\Rv (\qv)$ are identical,
the incoherent part vanishes. Its occurance is a direct consequence
of the disorder in the system.
\end{enumerate}
\subsection{Theory of neutron scattering -- specialization to crystal structures}
What happens if we specialize (\ref{sqom}) to a mono-atomic crystal
is, that the Hamiltonian is then invariant under the action of
the Bravais lattice $B$, implying that
\begin{equation}
     \bigl\<\re^{- \i \<\qv, \xv(\Sv)\>} \re^{\i \<\qv, \xv(\Rv,t)\>}\bigr\>_T
        = \bigl\<\re^{- \i \<\qv, \xv(0)\>} \re^{\i \<\qv, \xv(\Rv - \Sv,t)\>}\bigr\>_T \epp
\end{equation}
Inserting this into (\ref{sqom}) and also using that,
$\xv (\Rv) = \Rv + \uv(\Rv)$, where $\uv(\Rv)$ is the
deviation from the equilibrium position at $\Rv$, we obtain
\begin{equation}
     S(\qv, \om)
        = \int_{-\infty}^\infty \frac{\rd t}{2 \p} \:
	  \re^{- \i \om t} \sum_{\Rv \in B} \re^{\i \<\qv,\Rv\>}
	      \bigl\<\re^{- \i \<\qv, \uv(0)\>}
	              \re^{\i \<\qv, \uv(\Rv,t)\>}\bigr\>_T \epp
\end{equation}
The dynamic structure factor of the ions in a crystal lattice
is the spatio-temporal Fourier transform of the dynamical
two-point function $\bigl\<\re^{- \i \<\qv, \uv(0)\>}
\re^{\i \<\qv, \uv(\Rv,t)\>}\bigr\>_T$.

Similarly, the incoherent part of the differential cross section
simplifies to
\begin{equation}
     \frac{\rd^2 \s}{\rd \Om \rd E}\biggr|_{\rm inc} =
        \frac{\s_{\rm inc}}{4 \p} \frac{k'}{k} N_{\rm at}
	  \int_{-\infty}^\infty \frac{\rd t}{2 \p} \:
	  \re^{- \i \om t} \bigl\<\re^{- \i \<\qv, \uv(0)\>}
	                          \re^{\i \<\qv, \uv(0,t)\>}\bigr\>_T \epc
\end{equation}
the Fourier transformation in time of the auto correlation
function $ \bigl\<\re^{- \i \<\qv, \uv(0)\>}
\re^{\i \<\qv, \uv(0,t)\>}\bigr\>_T$.

\subsection{Dynamic structure factor of the harmonic crystal}
The harmonic approximation brings more simplifications about. If
$A$ and $B$ are any two operators that depend linearly on the
deviations $\uv (\Rv)$ of the ions from their equilibrium position
and linearly on the conjugate momentum $\pv(\Rv)$, then
\begin{equation} \label{anexpidentity}
     \bigl\<\re^A \re^B\bigr\>_T = \re^{\2 \<A^2 + 2AB + B^2\>_T} \epc
\end{equation}
where the canonical averages are calculated with the Hamiltonian
of the harmonic crystal (see exercise~\ref{ex:mermin}). Applying
this formula to $A = - \i \<\qv, \uv(0)\>$ and $B = \i \<\qv, \uv(\Rv,t)\>$
and denoting
\begin{equation} \label{defdebyewaller}
     2 W(\qv) = \bigl\<\<\qv, \uv(0)\>^2\bigr\>_T
              = \bigl\<\<\qv, \uv(\Rv,t)\>^2\bigr\>_T \epc
\end{equation}
we obtain the following formula for the dynamic structure factor
of the harmonic crystal,
\begin{equation} \label{sqomharm}
     S(\qv, \om)
        = \frac{\re^{- 2 W(\qv)}}{2 \p} \int_{-\infty}^\infty \rd t \:
	  \re^{- \i \om t} \sum_{\Rv \in B} \re^{\i \<\qv,\Rv\>}
	       \re^{\< \<\qv, \uv(0)\> \<\qv, \uv(\Rv,t)\> \>_T} \epp
\end{equation}
The function $W(\qv)$ is called the Debye-Waller factor.

\subsection{Expansion into multi-phonon processes}
A harmonic crystal has $3 N_{\rm at}$ eigenstates which can be
excited independently. Denote the corresponding occupation numbers
$n_1, \dots, n_{3 N_{\rm at}}$. Transition matrix elements can be
classified according to the differences in occupation numbers
between the initial state $\{n_i\}$ and final state $\{n_i'\}$.
\begin{description}
\item[Elastic processes.]
$n_i = n_i'$, all occupation numbers remain unaltered, no exchange
of energy between crystal and neutron.
\item[Single-phonon processes.]
$\exists\ \a \in \{1, \dots, 3 N_{\rm at}\}$ such that $n_i = n_i'\
\forall\ i \ne \a, n_\a' = n_\a \pm 1$, the occupation number of
a single mode is changed due to the scattering process.
\item[Multi-phonon processes.]
The occupation numbers of several modes are changed.
\end{description}
We shall see below that the $n$th term in the expansion
\begin{equation} \label{multiphexp}
     \re^{\< \<\qv, \uv(0)\> \<\qv, \uv(\Rv,t)\> \>_T} =
        1 + \bigl\< \dots \bigr\>_T + \2 \bigl\< \dots \bigr\>_T^2 + \dots
\end{equation}
can be identified with the $n$-phonon processes.

\subsection{Coherent elastic neutron scattering}
The contribution of the first term in (\ref{multiphexp}) to the
dynamic structure factor is
\begin{equation} \label{sqomelast}
     S_0 (\qv, \om)
        = N_{\rm at} \re^{- 2 W(\qv)} \de (\om)
	  \sum_{\Kv \in \overline B} \de_{\Kv, \qv} \epp
\end{equation}
Recall that the argument of the delta function is $\om =
(k^2 - k'^2)/2m$. Hence, the condition $\om = 0$ imposed
by the delta function in (\ref{sqomelast}) is the condition
of energy conservation for the scattered neutrons which
justifies the interpretation of the zeroth order term as
the elastic scattering term. Since $k, k' > 0$, the
zeroth order structure factor (\ref{sqomelast}) can only
be non-vanishing if $k = k'$. Consequentially, we obtain
the expression
\begin{equation}
     \frac{\rd^2 \s}{\rd \Om \rd E}\biggr|_{\rm coh} =
        \frac{\s_{\rm coh}}{4 \p} N_{\rm at}^2
        \re^{- 2 W(\kv - \kv')} \de (\om)
	\sum_{\Kv \in \overline B} \de_{\Kv, \kv - \kv'}
\end{equation}
for the elastic contribution to the coherent differential
cross section of the harmonic crystal. The sum over Kronecker
deltas on the right hand side means that elastic scattering
takes place, if the wave vectors $\kv$ of the incident neutron
beam and $\kv'$ of the scattered beam differ by a reciprocal
lattice vector $\Kv$,
\begin{equation} \label{bragg}
     \kv - \kv' = \Kv \in \overline B \epp
\end{equation}
This is the famous Bragg condition which also holds in
$X$-ray diffraction experiments.
\subsection{Interpretation of the Bragg condition}
The interpretation of the Bragg condition in reciprocal space
is depicted in Figure~\ref{fig:bragg_reci}.%
\begin{figure}
\begin{center}
\setlength{\tabcolsep}{3em}
\begin{tabular}{@{}cc@{}}
\raisebox{2pt}{\includegraphics[width=.28\textwidth]{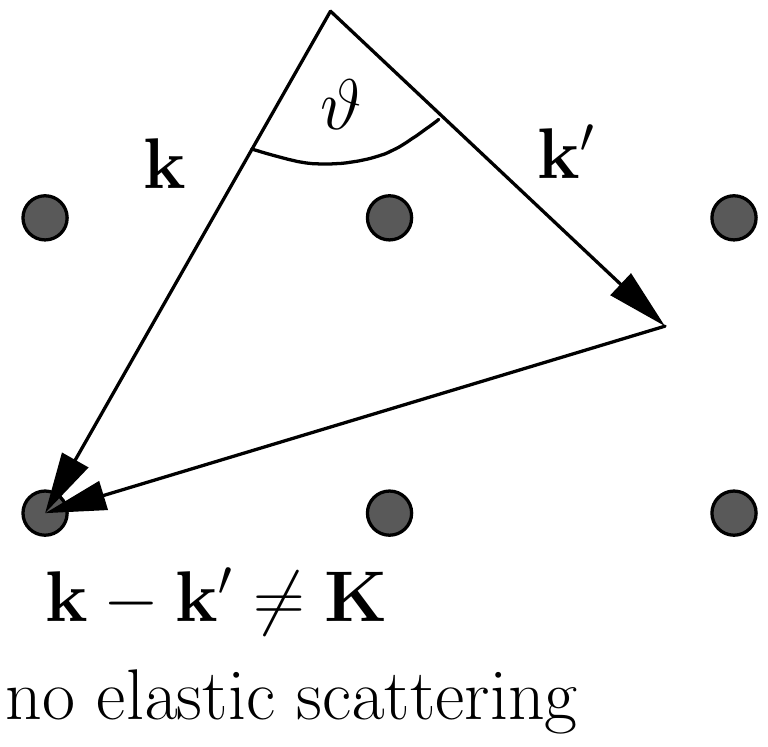}} &
\includegraphics[width=.52\textwidth]{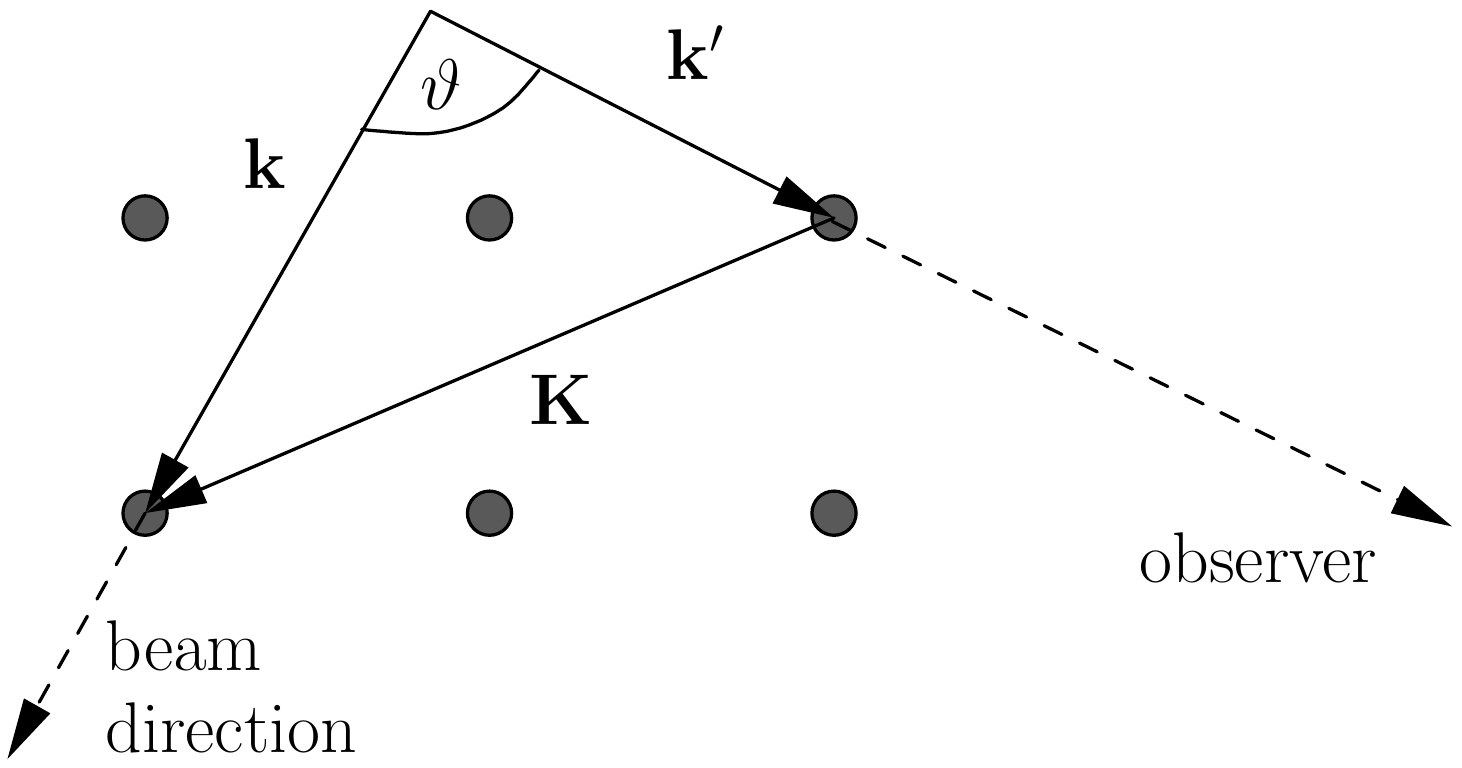}
\end{tabular}
\caption{\label{fig:bragg_reci} 
Illustration of the Bragg condition (\ref{bragg}) in the reciprocal
space.
}
\end{center}
\end{figure}
The Bragg condition is often formulated as a relation between
the distance $d$ of lattice planes in the original Bravais lattice,
the angle $\dh$ between incident and scattered neutron and
the wave length $\la = 2\p/k$ of the incident neutron. Recalling
from section~\ref{subsec:propreclat} that the reciprocal lattice
vector $\Kv$ corresponds to a family of lattice planes perpendicular
to $\Kv$ in such a way that $K = \|\Kv\| = n \cdot 2\p/d$ for some
$n \in {\mathbb N}$ and taking into account that elastic scattering
implies $k = k'$, we conclude that\\
\begin{minipage}{.6\textwidth}
\begin{multline}
     K^2 = k^2 +k'^2 - 2 k k' \cos(\dh) \\ = k^2 4 \sin^2(\dh/2)
         = (n 2\p/d)^2\\[2ex] \then\ n \la = 2 d \sin(\dh/2) \epp
\end{multline}

Relative to the Bravais lattice we can interpret the Bragg condition
in the following way.
\begin{enumerate}
\item
The scattering occurs as if it would happen at the lattice
planes according to the reflection law of geometric optics.
\item
Waves which are reflected by parallel planes interfere constructively,
meaning the difference in optical distance between two `rays'
reflected from consecutive lattice planes must be an integer
multiple of the wave length,
\begin{multline}
     2 \ell - b = 2 \ell - 2a \cos(\dh/2) \\
        = 2 \ell - 2 \ell \cos^2 (\dh/2) = 2 \ell \sin^2 (\dh/2) \\[1ex]
	= 2 d \sin(\dh/2) = n \la \epp
\end{multline}
\end{enumerate}
\end{minipage} \hspace{4ex}
\begin{minipage}{.36\textwidth}
\includegraphics[width=.88\textwidth]{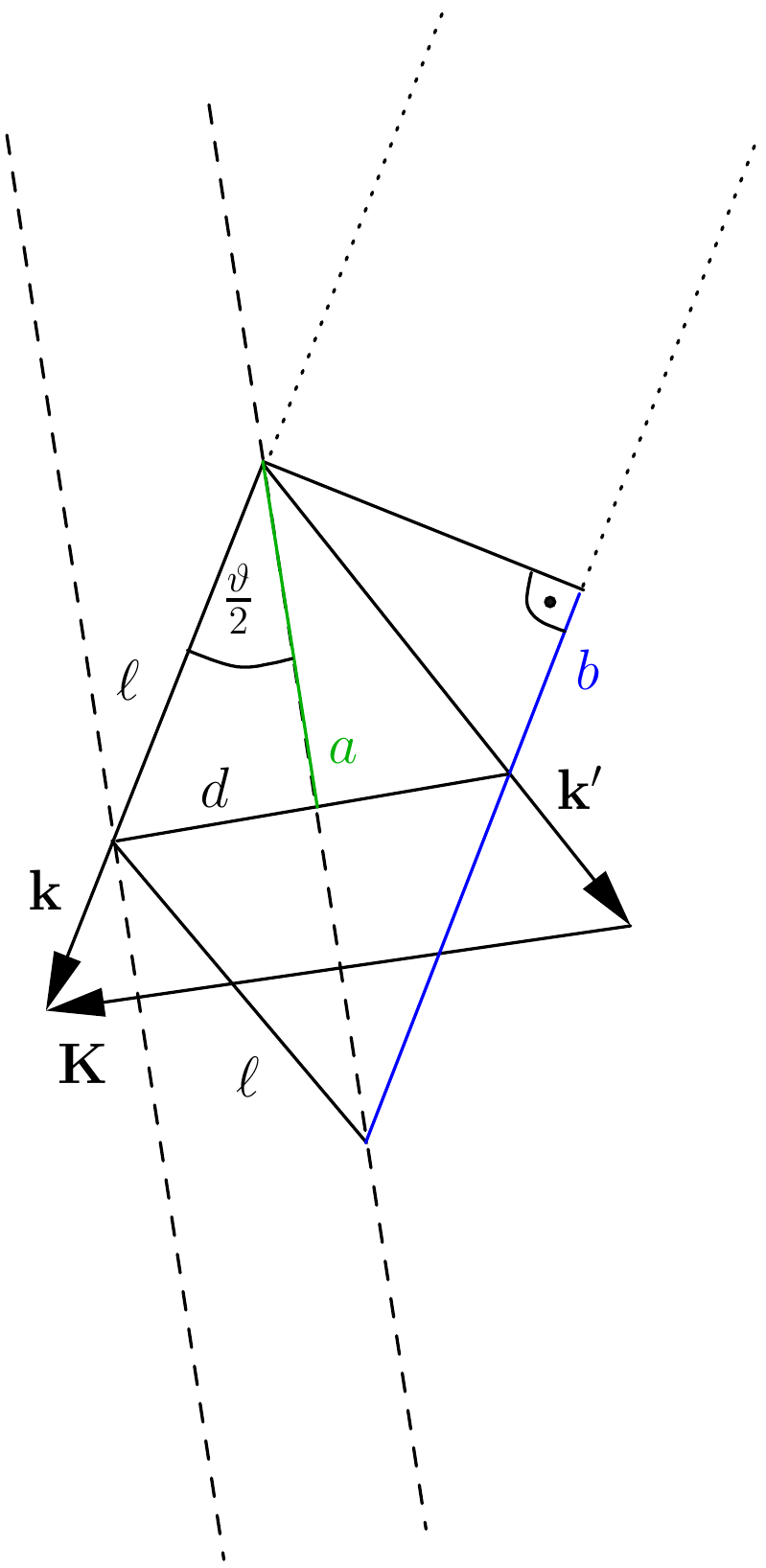}
\end{minipage}

\subsection{Exercise 12. A proof of equation (\ref{anexpidentity})}
\label{ex:mermin}
In the derivation of the formula (\ref{sqomharm}) for the dynamic
structure factor of the harmonic crystal we have used equation
(\ref{anexpidentity}). Here we would like to provide a step-by-step
derivation (cf.\ \cite{Mermin66}).

Consider a system of harmonic oscillators with Hamiltonian
\begin{equation} \label{harmloseoscis}
     H = \sum_i \om_i \bigl(a_i^+ a_i + \tst{\2}\bigr) \epc
\end{equation}
where the $a_i$ and $a_j^+$ are annihilation and creation
operators with commutation relations $[a_i, a_j] = 0 = [a_i^+, a_j^+]$,
$[a_i, a_j^+] = \de_{ij}$. Let
\begin{equation}
     A = \sum_i \bigl(\a_i a_i + \be_i a_i^+\bigr) \epc \qd
     B = \sum_i \bigl(\g_i a_i + \de_i a_i^+\bigr) \epc
\end{equation}
where $\a_i, \be_j, \g_k, \de_l \in {\mathbb C}$. Denote the canonical
ensemble average with Hamiltonian (\ref{harmloseoscis}) by $\< \cdot \>$.
Our goal is to prove that
\begin{equation} \label{expaexpb}
     \bigl\<\re^A \re^B\bigr\> = \re^{\2 \<A^2 + 2AB + B^2\>} \epp
\end{equation}
The proof will be based on the formula
\begin{equation} \label{bkhb}
     \re^A \re^B = \re^{\2 [A,B]} \re^{A + B}
\end{equation}
which should be familiar from the introductory lecture on quantum
mechanics and follows from the Baker-Campbell-Hausdorff formula
(exercise: recall the derivation of both formulae). It holds whenever
$\ad_A [A,B] = \ad_B [A,B] = 0$ (where the adjoint action of an
operator $X$ on any operator $Y$ is defined by $\ad_X Y = [X,Y]$).
\begin{enumerate}
\item
Define
\begin{equation}
     C = \sum_i c_i a_i \epc \qd D = \sum_i d_i a_i^+ \epp
\end{equation}
Use (\ref{bkhb}) to show that
\begin{equation}
     \bigl\<\re^C \re^D\bigr\> = \re^{[C,D]} \bigl\<\re^{D} \re^{C}\bigr\>
        = \re^{[C,D]} \bigl\<\re^{\re^{\ad_{H/T}} C} \re^{D}\bigr\> \epp
\end{equation}
\item
Let
\begin{equation}
     C_n = \re^{n \ad_{H/T}} C \epc \qd n \in {\mathbb N}_0 \epp
\end{equation}
Show that
\begin{equation} \label{oprec}
     \bigl\<\re^C \re^D\bigr\>
        = \re^{[\sum_{k=0}^n C_k, D]} \bigl\<\re^{C_{n+1}} \re^{D}\bigr\> \epp
\end{equation}
\item
Show that (formally) $\lim_{n \rightarrow \infty} C_n = 0$ and
that the series $\sum_{k=0}^\infty C_n$ has a formal limit.
Denoting this limit by $S$ and using that $\bigl\<\re^{D}\bigr\> = 1$,
equation (\ref{oprec}) then implies that
\begin{equation}
      \bigl\<\re^C \re^D\bigr\> = \re^{[S, D]} \epp
\end{equation}
\item
Show in a similar way as above that
\begin{equation}
     \<CD\> = [S,D] \epp
\end{equation}
Thus,
\begin{equation} \label{expbc}
     \bigl\<\re^C \re^D\bigr\> = \re^{\<CD\>} \epp
\end{equation}
\item
Use (\ref{bkhb}) and (\ref{expbc}) to prove (\ref{expaexpb}).
\end{enumerate}

\mysection{Inelastic neutron scattering}
\subsection{Debye-Waller factor}
If we solve equations (\ref{aadintermsofu}) for $u^\a (\Rv)$ we obtain
\begin{equation} \label{uintermsofaad}
     u^\a (\Rv) = \frac{1}{\sqrt{L^3}}
		  \sum_{(\qv, \be) \in Q} \re^{\i \<\qv, \Rv\>} Y^\a_\be (\qv)
		  \frac{a_\qv^\be + {a^+}_{- \qv}^\be}{\sqrt{2 \m_\a \om_\be (\qv)}} \epc
\end{equation}
the deviations of the ions from their equilibrium positions in
a center of mass frame of reference (the zero modes are excluded
in the summation, see section~\ref{sec:zeromodes}).

Let us assume for simplicity that we are dealing with a mono-atomic
lattice. Then the indices $\a, \be$ in (\ref{uintermsofaad}) take
values $x, y, z$, the dimensionless masses reduce to $\m_\a = 1$,
and $L^3 = N_{\rm at}$. The scalar products defining the Debye-Waller
factor (\ref{defdebyewaller}) become
\begin{equation}
     \<\qv, \uv(0)\> = \frac{1}{\sqrt{N_{\rm at}}}
		  \sum_{\substack{\kv \in BZ\\ \kv \ne 0}} \sum_{\a = x, y, z}
		  \frac{\<\qv, \yv_\a (\kv)\>}{\sqrt{2 \om_\a (\kv)}}
		  \bigl(a_\kv^\a + {a^+}_{- \kv}^\a\bigr) \epc
\end{equation}
where the $\yv_\a (\kv)$ are the polarization vectors defined below
(\ref{phononspectrum}). It follows that
\begin{multline} \label{dwfinite}
     2 W(\qv) = \bigl\< \<\qv, \uv(0)\>^2 \bigr\>_T \\[1ex]
        = \frac{1}{N_{\rm at}}
	  \sum_{\substack{\kv, \kv' \in BZ\\ \kv, \kv' \ne 0}} \sum_{\a, \be = x, y, z}
	     \frac{\<\qv, \yv_\a (\kv)\>}{\sqrt{2 \om_\a (\kv)}}
	     \frac{\<\qv, \yv_\be (\kv')\>}{\sqrt{2 \om_\be (\kv')}}
	     \bigl\<\bigl(a_\kv^\a + {a^+}_{- \kv}^\a\bigr)
	     \bigl(a_{\kv'}^\be + {a^+}_{- \kv'}^\be\bigr)\bigr\>_T \\[1ex]
        = \frac{1}{N_{\rm at}}
	  \sum_{\substack{\kv \in BZ\\ \kv \ne 0}} \sum_{\a = x, y, z}
	     \frac{\<\qv, \yv_\a (\kv)\> \<\qv, \yv_\a (- \kv)\>}{2 \om_\a (\kv)}
	     \cth \Bigl( \frac{\om_\a (\kv)}{2 T} \Bigr) \\[1ex]
        = \frac{1}{N_{\rm at}}
	  \sum_{\substack{\kv \in BZ\\ \kv \ne 0}} \sum_{\a = x, y, z}
	     \frac{|\<\qv, \yv_\a (\kv)\>|^2}{2 \om_\a (\kv)}
	     \cth \Bigl( \frac{\om_\a (\kv)}{2 T} \Bigr) \epp
\end{multline}
Here we have used (cf.\ (\ref{boseeinstein}))
\begin{multline}
     \bigl\<\bigl(a_\kv^\a + {a^+}_{- \kv}^\a\bigr)
            \bigl(a_{\kv'}^\be + {a^+}_{- \kv'}^\be\bigr)\bigr\>_T
        = \de^{\a \be} \de_{\kv, - \kv'}
	   \bigl\< a_\kv^\a {a^+}_\kv^\a + {a^+}_\kv^\a a_\kv^\a \bigr\>_T \\[1ex]
        = \de^{\a \be} \de_{\kv, - \kv'} \bigl\< 2 \hat n_\kv^\a + 1\bigr\>_T
        = \de^{\a \be} \de_{\kv, - \kv'} \cth \Bigl( \frac{\om_\a (\kv)}{2 T} \Bigr)
\end{multline}
in the second equation and (\ref{ytreversal}) in the third equation.
Taking the thermodynamic limit in (\ref{dwfinite}) we end up with
\begin{equation} \label{dwtl}
     2 W(\qv) = \frac{V_u}{(2\p)^3} \sum_{\a = x, y, z}
		\int_{BZ} \rd^3 k \:
                \frac{|\<\qv, \yv_\a (\kv)\>|^2}{2 \om_\a (\kv)}
		\cth \Bigl( \frac{\om_\a (\kv)}{2 T} \Bigr) \epp
\end{equation}

Here two remarks are in order.
\begin{enumerate}
\item
In general (\ref{dwtl}) cannot be rewritten by means of the density
of states, since the polarization vectors $\yv_\a (\kv)$ depend
on $\kv$ not necessarily through $\om_\a (\kv)$.
\item
The Debye-Waller factor describes the weakening of the coherent
scattering due to thermal fluctuation and as a function of the
change of momentum $\qv$, i.e.\ as a function of the scattering
angle. It appears as a prefactor $\re^{- 2 W(\qv)}$ in the dynamic
structure factor. Since $W(\qv)$ is quadratic in $\qv$, the
weakening is larger for larger $\qv$. In elastic scattering, where
$\qv$ must equal a reciprocal lattice vector $\Kv$, the scattering is
strong in the direction corresponding to minimal $\Kv$ or
maximal distance between the associated family of lattice planes.
Recall (cf. exercise~\ref{ex:latplan}) that this family has maximal
density of lattice points within a plane.

As a function of the temperature the scattered intensity weakens
with growing $T$, since
\begin{equation}
     \cth \Bigl( \frac{\om_\a (\kv)}{2 T} \Bigr) \sim \frac{2T}{\om_\a (\kv)}
\end{equation}
as $T \rightarrow \infty$, leading to a linear temperature dependence
of $W(\qv)$ and to a suppression of the scattered intensity 
exponentially in $T$.
\end{enumerate}

\subsection{Exercise 13. Debye-Waller factor of a cubic lattice}
For a cubic lattice consisting of $N$ atoms the Debye-Waller factor can
be written as
\begin{equation}
     2W(q) = \frac{V q^2}{N} \int_0^\infty \frac{\rd \omega}{\omega^2}
             E(\omega) g(\omega)\epc
\end{equation}
wherein $g(\omega)$ is the density of states (per branch), $q$
is the difference of momenta and $E(\omega)$ is the average of
the energy of the phonons with frequency $\omega$,
\begin{equation}
  E(\omega) =  \omega \left[ n(\omega,T) + \frac{1}{2} \right]\epc \qd
  n(\omega,T) = \left[ \exp( \omega / T) - 1\right]^{-1} \epp
\end{equation}
\begin{enumerate}
\item
Calculate the Debye-Waller factor for a cubic Bravais lattice by using the Einstein
model for the phonons. This is a model for an optical phonon branch, where all
ions oscillate with the same frequency $\omega_0$; thus all $N$ states are
located at $\omega_0$, and the density of states is given by $g(\omega) =
(N/V) \delta(\omega - \omega_0)$. Determine $2W(q)$ for $ T\ll \omega_0 $ or
$T \gg \omega_0$, respectively.
\item
What is the Debye-Waller factor for the Debye model? What follows for $2W(q)$
in the limits $ T \ll \omega_D $ and $T \gg \omega_D $?
\end{enumerate}

\subsection{Coherent inelastic scattering -- the single-phonon contribution}
Let us now consider the contribution of the second term in (\ref{multiphexp})
to the dynamic structure factor,
\begin{equation} \label{sqomone}
     S_1 (\qv, \om)
        = \frac{\re^{- 2 W(\qv)}}{2 \p} \int_{-\infty}^\infty \rd t \:
	  \re^{- \i \om t} \sum_{\Rv \in B} \re^{\i \<\qv,\Rv\>}
	       \< \<\qv, \uv(0)\> \<\qv, \uv(\Rv,t)\> \>_T \epp
\end{equation}
As we shall see, this term corresponds to the single-phonon contributions.
This can be understood from a closer inspection of the two-point
function on the right hand side. We infer from (\ref{uintermsofaad}) that
\begin{equation}
     \<\qv, \uv(\Rv)\> = \frac{1}{\sqrt{N_{\rm at}}}
		  \sum_{\substack{\kv \in BZ\\ \kv \ne 0}} \sum_{\a = x, y, z} \re^{\i \<\kv,\Rv\>}
		  \frac{\<\qv, \yv_\a (\kv)\>}{\sqrt{2 \om_\a (\kv)}}
		  \bigl(a_\kv^\a + {a^+}_{- \kv}^\a\bigr) \epp
\end{equation}
This evolves in time into
\begin{multline}
     \<\qv, \uv(\Rv,t)\> = \re^{\i Ht} \<\qv, \uv(\Rv)\> \re^{- \i Ht} \\[1ex]
        = \frac{1}{\sqrt{N_{\rm at}}}
		  \sum_{\substack{\kv \in BZ\\ \kv \ne 0}} \sum_{\a = x, y, z} \re^{\i \<\kv,\Rv\>}
		  \frac{\<\qv, \yv_\a (\kv)\>}{\sqrt{2 \om_\a (\kv)}}
		  \bigl(a_\kv^\a \re^{- \i \om_\a (\kv)t}
		        + {a^+}_{- \kv}^\a \re^{\i \om_\a (\kv)t}\bigr) \epp
\end{multline}
Hence, performing a similar calculation as in (\ref{dwfinite}),
\begin{align}
     \< \<& \qv, \uv(0)\> \<\qv, \uv(\Rv,t)\> \>_T \\[1ex]
        & = \frac{1}{N_{\rm at}}
	  \sum_{\substack{\kv, \kv' \in BZ\\ \kv, \kv ' \ne 0}}
	  \sum_{\a, \be = x, y, z} \re^{\i \<\kv,\Rv\>}
	     \frac{\<\qv, \yv_\a (\kv)\>}{\sqrt{2 \om_\a (\kv)}}
	     \frac{\<\qv, \yv_\be (\kv')\>}{\sqrt{2 \om_\be (\kv')}}
	     \mspace{180.mu} \notag \\[-1ex] & \mspace{180.mu} \times
	     \de^{\a \be} \de_{\kv, - \kv'}
	     \bigl(\<\hat n_\kv^\a\>_T \re^{- \i \om_\a (\kv) t} +
	           \<\hat n_\kv^\a + 1\>_T \re^{\i \om_\a (\kv) t}\bigr)
	     \notag \\[1ex]
        & = \frac{1}{N_{\rm at}}
	  \sum_{\substack{\kv \in BZ\\ \kv \ne 0}} \sum_{\a = x, y, z} \re^{\i \<\kv,\Rv\>}
	     \frac{|\<\qv, \yv_\a (\kv)\>|^2}{2 \om_\a (\kv)}
             \biggl\{\frac{\re^{- \i \om_\a (\kv) t}}{\re^{\om_\a (\kv)/T} - 1}
                     + \frac{\re^{\i \om_\a (\kv) t}}{1 - \re^{- \om_\a (\kv)/T}} \biggr\}
	     \notag \\
        & = \frac{1}{N_{\rm at}}
	  \sum_{\substack{\kv \in BZ\\ \kv \ne 0}} \sum_{\a = x, y, z}
	     \frac{|\<\qv, \yv_\a (\kv)\>|^2}{2 \om_\a (\kv)}
             \biggl\{\frac{\re^{\i (\<\kv, \Rv\> - \om_\a (\kv) t)}}
	                  {\re^{\om_\a (\kv)/T} - 1}
                     + \frac{\re^{- \i(\<\kv,\Rv\> - \om_\a (\kv) t)}}
		            {1 - \re^{- \om_\a (\kv)/T}} \biggr\} \epp
\end{align}
Here we have substituted $- \kv$ for $\kv$ in the second sum in
the last equation and have used that $\om_\a (\kv)$ and
$|\<\qv, \yv_\a (\kv)\>|^2$ are even functions of $\kv$.

Inserting the latter equation into (\ref{sqomone}) we obtain the
single-phonon contribution to the dynamic structure factor in the
form
\begin{equation}
     S_1 (\qv, \om) = S_{1, +} (\qv,\om) + S_{1, -} (\qv, \om) \epc
\end{equation}
where
\begin{multline} \label{sqomonep}
     S_{1, +} (\qv, \om)
        = \frac{\re^{- 2 W(\qv)}}{2 \p} \int_{-\infty}^\infty \rd t \:
	  \re^{- \i \om t} \sum_{\Rv \in B} \re^{\i \<\qv,\Rv\>} \\[-1ex]
          \mspace{135.mu} \times \frac{1}{N_{\rm at}}
	  \sum_{\substack{\qv' \in BZ\\ \qv' \ne 0}} \sum_{\a = x, y, z}
	     \frac{|\<\qv, \yv_\a (\qv')\>|^2}{2 \om_\a (\qv')}
                    \frac{\re^{- \i(\<\qv',\Rv\> - \om_\a (\qv') t)}}
		         {1 - \re^{- \om_\a (\qv')/T}} \\[1ex]
        = \re^{- 2 W(\qv)}
	  \sum_{\substack{\qv' \in BZ\\ \qv' \ne 0}} \sum_{\a = x, y, z}
	     \frac{|\<\qv, \yv_\a (\qv')\>|^2}{2 \om_\a (\qv')}
                    \frac{\de(\om - \om_\a (\qv'))}
		         {1 - \re^{- \om_\a (\qv')/T}}
          \sum_{\Kv \in \overline B} \de_{\Kv, \qv - \qv'}
\end{multline}
and
\begin{multline} \label{sqomonem}
     S_{1, -} (\qv, \om) \\
        = \re^{- 2 W(\qv)}
	  \sum_{\substack{\qv' \in BZ\\ \qv' \ne 0}} \sum_{\a = x, y, z}
	     \frac{|\<\qv, \yv_\a (\qv')\>|^2}{2 \om_\a (\qv')}
                    \frac{\de(\om + \om_\a (\qv'))}
		         {\re^{\om_\a (\qv')/T} - 1}
          \sum_{\Kv \in \overline B} \de_{\Kv, \qv + \qv'} \epp
\end{multline}

These are those contributions to the dynamic structure factor, where
precisely one phonon is excited (emitted) or absorbed. We shall denote
the corresponding contributions to the differential cross section
\begin{equation}
     \frac{\rd^2 \s}{\rd \Om \rd E}\biggr|_{\rm coh, 1 \pm} =
        \frac{\s_{\rm coh}}{4 \p} \frac{k'}{k} S_{1, \pm} (\qv, \om) \epp
\end{equation}
Like the elastic cross section these cross sections are `discontinuous'
and describe a pattern of bright spots at certain energy and
momentum transfers. The delta functions and Kronecker deltas in
(\ref{sqomonep}) and (\ref{sqomonem}) force the dispersion relations
of the scattered neutron and the phonons to match in the following
sense.
\begin{description}
\item[Emission of a phonon.]
In $S_{1, +}$ the momenta of the neutron and the phonon involved
in the scattering process must satisfy
\begin{equation}
     \qv = \kv - \kv' = \qv' + \Kv \epc
\end{equation}
where $\kv - \kv'$ is the momentum transfer to the lattice and
$\qv'$ is the momentum taken by the phonon, i.e.\ the phonon
takes the momentum transferred to the lattice and reduced to
the Brillouin zone. For the energies we have the matching condition
\begin{equation} \label{resemission}
     \om = \frac{k^2 - k'^2}{2m} = \om_\a (\qv') = \om_\a (\qv - \Kv)
         = \om_\a (\kv - \kv') \epc
\end{equation}
where we have used the periodicity of $\om_\a$ with respect
to the reciprocal lattice in the last equation. Since $\om_\a \ge 0$
we see that $k \ge k'$ in this process. The lattice absorbs energy,
`a phonon is emitted.'

The temperature dependence of the emission is encoded in the
factor $1/(1 - \re^{- \om_\a (\qv')/T})$ in $S_{1, +}$. This
factor decreases as $T$ decreases but stays finite for
$T \rightarrow 0+$.
\item[Absorption of a phonon.]
For $S_{1, -}$ the momentum balance is
\begin{equation}
     \qv = \kv - \kv' = - \qv' + \Kv
\end{equation}
which we interpret such that the neutron takes momentum $\qv'
- \Kv$. The relation for the exchange of energy between the
scattered neutron an the phonon involved in the process is
\begin{equation} \label{resabsorption}
     - \om = \frac{k'^2 - k^2}{2m} = \om_\a (\qv') = \om_\a (\Kv - \qv)
         = \om_\a (\kv' - \kv) \epp
\end{equation}
Here $k' \ge k$, and energy is absorbed by the neutron, which
is interpreted as a phonon being absorbed in the process.

In this case the temperature dependence comes in through a
factor $1/(\re^{\om_\a (\qv')/T} - 1)$ which vanishes as
$T \rightarrow 0+$. At zero temperature no phonon remains
in the system and there is nothing left to be absorbed.
\end{description}
\subsection{Measuring dispersion relations of phonons}
\begin{figure}
\begin{center}
\includegraphics[width=.82\textwidth]{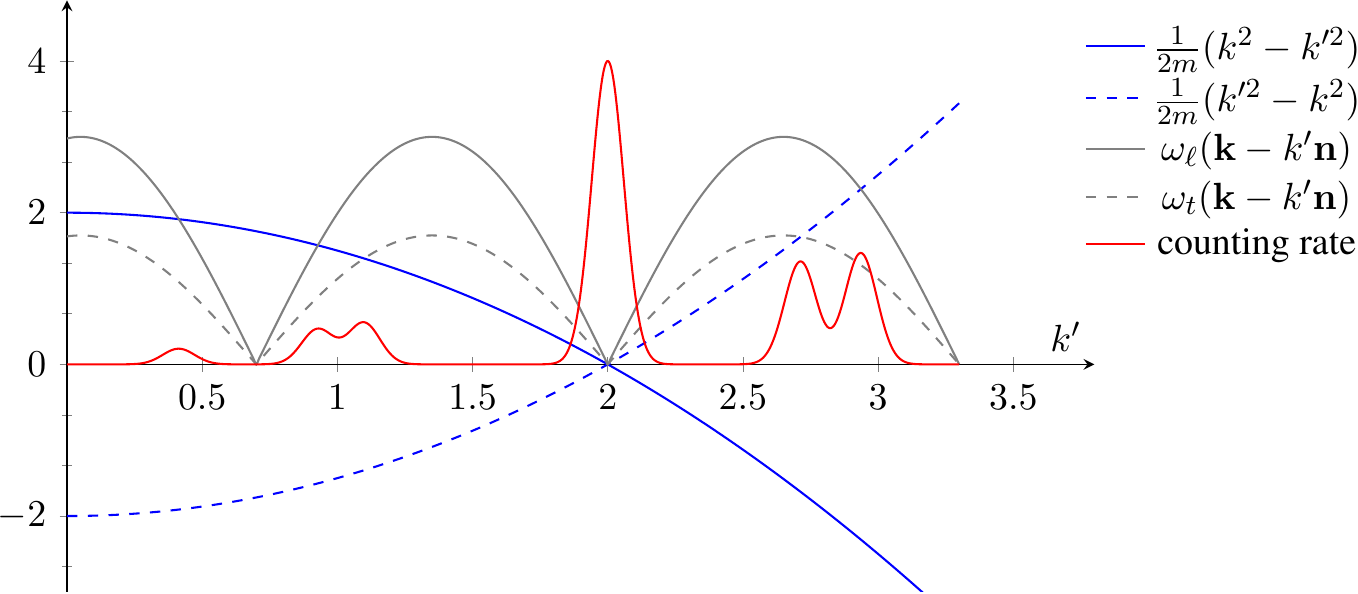}
\caption{\label{fig:inelaneu} 
Schematic picture explaining how phonon dispersion relations
can be inferred from the differential cross section.
}
\end{center}
\end{figure}
Dispersion relations of phonons can be inferred from the
differential cross section. Assume that a beam of mono-energetic
neutrons impinges on a crystal and the velocity distribution
of the neutrons scattered in a fixed direction $\nv$ is
measured . Strong scattering in direction $\nv$ takes place
whenever $\kv' = k' \nv$ satisfies one of the resonance
conditions (\ref{resemission}) or (\ref{resabsorption}) for
emission or absorption. The situation is then as sketched
in Figure~\ref{fig:inelaneu}. In the example we have assumed
a simple lattice with a longitudinal acoustic phonon with
dispersion $\om_\ell$ (gray line) and two degenerate transversal
acoustic phonons with dispersion $\om_t$ (dashed gray line).  
The blue line represents the energy loss of the scattered
neutron if it has momentum $k'$. The scattering is likely
only if there is a phonon of momentum $\kv - k' \nv$ modulo
reciprocal lattice vectors, which can be emitted at this
energy. Similarly, the dashed blue line shows the energy
gain of the neutron if a phonon is absorbed in the scattering
process. One of the resonance conditions is satisfied if
one of the dispersion curves intersects with the blue
or dashed blue line. The red line represents schematically
the corresponding cross-section. If the energy of the incident
neutrons is varied the blue and dashed blue lines sweep
over the dispersion curves of the phonons and the full
dispersion relation in $\nv$ direction is mapped out.

\mysection{Electronic excitations in solids}
\subsection{Hamiltonian of the electrons in adiabatic approximation}
We saw in section~\ref{lect:bo} that the adiabatic principle implies
a decoupling of the lattice and electronic degrees of freedom. To
leading order, the dynamics of the electrons is governed by the
Hamiltonian
\begin{multline} \label{hell}
     H_{\rm el} (\Rv) =
        - \2 \sum_{j=1}^N \6_{\rv_j^\a}^2
	+ \sum_{j=1}^N
	  \underbrace{\sum_{k=1}^L \frac{- Z_k}{\|\rv_j - \Rv_k\|}}_{V_{\rm I} (\rv_j)}
	+ \underbrace{\sum_{1 \le j < k \le N} \frac{1}{\|\rv_j - \rv_k\|}}_{V_{\rm elel}}
	+ M \epc
\end{multline}
where the ions are fixed to their equilibrium positions $\Rv_k$ and
\begin{equation}
     M = \sum_{1 \le j < k \le L} \frac{Z_j Z_k}{\|\Rv_j - \Rv_k\|}
\end{equation}
is the electro-static energy of the ions, the so-called Madelung
energy. The latter plays no role for the dynamics of the electrons.
Thus, to leading order in the adiabatic approximation, the electrons
in a crystal are described by a repulsive Coulomb gas with interaction
$V_{\rm elel}$ that is filled into the periodic potential $V_{\rm I}$
generated by the ions sitting on their equilibrium positions.

Because of the mutual Coulomb interaction of the electrons, this is
still an interacting many-body quantum system, and imperturbable optimism
is required to believe that it can ever be solved exactly or numerically
with sufficient accuracy. At least at the current stage of our
knowledge drastic further approximations are necessary in order to
be able to make any quantitative prediction.

\subsection{Reduction to a single-particle problem}
\label{sec:reductionone}
\begin{enumerate}
\item
Only the valence electrons (electrons outside closed shells) contribute
significantly to the typical properties of solids, because they are
responsible for the chemical bonding and are distributed over the solid.
For this reason we shall interpret $V_{\rm I}$ as the potential of
ions with completely filled shells (Coulomb with reduced charge number $Z_j$;
example sodium, alkali metal, one valence electron, $Z = 1$).
\item
If the electron-electron interaction $V_{\rm elel}$ could be neglected,
the `electronic problem' would be reduced to the problem of a 
single particle in the periodic potential $V_{\rm I}$.
\item
Instead of simply neglecting the electron-electron interaction, we
shall try to split it into an `effective single-particle contribution'
which modifies the periodic potential $V_{\rm I}$ and a residual
`many-body contribution' which for the understanding of certain
quantities can be neglected in many cases, at least in the first
instance. Some ideas about the systematic derivation of an effective
single-particle description will be explained below.
\item
An intuitive explanation of how the reduction to a single-particle
problem is possible is through the notions of `screening' and
`mean fields'. The full many-body Hamiltonian $H_{\rm el}$ is
invariant under the action of the Bravais lattice associated with
the equilibrium positions of the ions, $[H_{\rm el}, U_\Rv] = 0\
\forall\ \Rv \in B$. Hence, following the same reasoning as for
the derivation of Bloch's theorem in section~\ref{sec:bloch},
we see that the lattice momentum $\kv \in BZ$ is a good quantum
number for the many-body eigenfunctions and that those transform like
\begin{equation}
     \Ps_\kv (\rv_1 + \Rv, \dots, \rv_N + \Rv)
          = \re^{\i \<\kv, \Rv\>} \Ps_\kv (\rv_1, \dots, \rv_N) \epc
\end{equation}
$\Rv \in B$, under the action of the Bravais lattice. It follows
that
\begin{equation} \label{manyprobdens}
     |\Ps_\kv (\rv_1 + \Rv, \dots, \rv_N + \Rv)|^2 = |\Ps_\kv (\rv_1, \dots, \rv_N)|^2
\end{equation}
for all $\Rv \in B$. The electronic charge density associated with
this state,
\begin{multline}
     D(\xv) = - e \int \rd^{3N} r \: |\Ps_\kv (\rv_1, \dots, \rv_N)|^2
                  \sum_{j=1}^N \de(\xv - \rv_j) \\
            = - e N \int \rd^{3(N-1)} r \: |\Ps_\kv (\xv, \rv_2, \dots, \rv_N)|^2 \epc
\end{multline}
clearly inherits the invariance under translations by Bravais lattice
vectors, since
\begin{multline}
     D(\xv + \Rv)
        = - e \int \rd^{3N} r \: |\Ps_\kv (\rv_1, \dots, \rv_N)|^2
	      \sum_{j=1}^N \de(\xv - \rv_j + \Rv) \\
        = - e \int \rd^{3N} r \: |\Ps_\kv (\rv_1 + \Rv, \dots, \rv_N +\Rv)|^2
	      \sum_{j=1}^N \de(\xv - \rv_j) = D(\xv) \epc
\end{multline}
where we have used (\ref{manyprobdens}) in the third equation.
We imagine this as the charge density `screening' the ionic
potential $V_{\rm I}$ and combining together with it to an effective
periodic potential which would be felt by an additional electron
inserted as a probe into the solid. This `screened periodic potential'
provides an intuitive single-particle description for the electron
gas in a crystal: a particle moving in the mean field of all
other particles.
\item
Another way of thinking about the existence of an effective
single-particle description of the solid is the following. Consider
the excitations of the full many-particle system (\ref{hell}).
If there are excitations with charge quantum numbers $\pm e$
which do have a dispersion, i.e.\ excitations for which a
definite change of (lattice) momentum of the many-body system always
comes with one and the same definite change of energy, then we
call them quasi-particle excitations or simply particles.
It is always possible to define a single-particle Hamiltonian
(in momentum representation) that has exactly the same dispersion
relation (the same spectrum) as the full many-particle Hamiltonian
(\ref{hell}). This Hamiltonian may be seen as an effective
single-particle Hamiltonian of the full system. How well it
describes the full system depends on the details. Unlike e.g.\
in the harmonic crystal which realizes an ideal gas of
phonons, the two- and multi-particle excitations of the full
electronic system are not just superpositions of single-particle
excitations. Still, the effective interactions between the
quasi-particles may be weak, in which case the effective
single-particle description will give a good description of
at least the thermodynamic properties of the system.

It is beyond the current capabilities of theoretical physics
to prove the existence of quasi particles for the Hamiltonian
(\ref{hell}). But many experiments show that solids generically
admit quasi-particle electronic excitations with charge
quantum numbers $\mp e$. These are called `electrons' and
`holes' (the holes are the solid-state analogues of the
positrons). Experiments show in addition that the interaction
between several electrons or holes or between the electrons and
the holes can often be neglected in the first instance.
\end{enumerate}

\subsection{Particles in a periodic potential}
The above discussion should provide enough motivation to study
the general problem of a particle moving in a periodic
potential. The corresponding Hamiltonian is
\begin{equation} \label{hamper}
     H = - \2 \6_\xv^2 + V(\xv)
\end{equation}
with $\xv \in {\mathbb R}^3$ and $V(\xv) = V(\xv + \Rv)$ for all
$\Rv \in B$. The study of this Hamiltonian will lead us to the
extraordinarily successful band model of solids.

Let us start with some general remarks.
\begin{enumerate}
\item
The one dimensional Kronig-Penney model with potential
\begin{equation}
     V(x) = V_0 \sum_{\ell \in {\mathbb Z}} a \, \de(x - \ell a) \epc
\end{equation}
where $a$ is the lattice constant and $V_0$ the strength of the
interaction, is the only simple but non-trivial model of particles
in a periodic potential which admits a closed analytic solution.
The next simple case $V (x) = - V_0 \cos (2\p x/a)$ involves Mathieu
functions. Its understanding requires already some mathematical
effort.
\item
There is no other choice than trying to understand the general
properties of particles in a periodic potential by the common
means of mathematics or theoretical physics. As far as the mathematical
part is concerned it would be instructive to study Floquet theory
as part of the theory of ordinary differential equations. For
time limitations we refrain from touching this interesting subject
and rather go ahead with typical methods of theoretical physics 
which are based on perturbation theory.
\end{enumerate}

\subsection{Exercise 14. Kronig-Penney model}
The Kronig-Penney model is a simple one-dimensional model for the understanding
of the band structure in solids. Consider an electron of mass $m$ moving in
a periodic potential
\begin{equation}
     V(x) = V_0 \sum_{n=-\infty}^{\infty} a \, \delta(x-na) \epp
\end{equation}
Here $V_0$ is the strength of the potential and $a$ the lattice constant.
Note that either sign of $V_0$ makes sense. For $V_0 > 0$ the potential is
repulsive, for $v_0 < 0$ it is attractive.
\begin{enumerate}
\item
Introduce dimensionless units such that the time-independent Schr\"odinger
equation takes the form
\begin{equation} \label{sg}
  H \varphi(y) = \biggl[ -\frac{\partial^2}{\partial y^2}
                 + 2c \sum_{n=-\infty}^{\infty} \delta(y-n) \biggr] \varphi(y)
                = q^2\varphi(y)\epp
\end{equation}
Which connections exist between $x$ and $y$, $c$ and $V_0$, and between
$q^2$ and the energy $E$?
\item
Because of the periodicity of the potential, the Hamiltonian \eqref{sg} commutes
with the shift operator defined by $T\varphi(y) = \varphi(y+1)$. Thus, $H$ and
$T$ have a common system of eigenfunctions, i.e., the eigenvalue equations
\begin{subequations} \label{simul}
\begin{gather}
  H\varphi(y) = q^2\varphi(y) \epc \\
  T\varphi(y) = \textrm{e}^{i k}\varphi(y)
\end{gather}
\end{subequations}
can be solved simultaneously. Determine the solutions of \eqref{simul} as
a function of $k$ and $q$. Which equation connects $k$ with $q$? For the
calculation it is sufficient to consider the Schr\"odinger equation with
the general solution $\varphi(y) = A \textrm{e}^{i qy} + B \textrm{e}^{-i qy}$
in the interval $[0,1]$. Note that $q$ can take real as well as imaginary values.
\item 
Discuss the dispersion relation $\cos(k) = \cos(q) + (c/q)\sin(q)$ following
from (ii) graphically. Observe that $k$ has to be real in order for $\varphi(y)$ to
be bounded and normalizable. Therefore the eigenstates of $H$ and $T$ are
restricted on certain energy bands.
\end{enumerate}

\subsection{Almost free electrons}
The qualitative features of the motion of particles can be understood
from time-independent perturbation theory. Since $V(\xv)$ is periodic,
it has the Fourier series representation (cf. section~\ref{subsec:fourier})
\begin{equation} \label{elpotfou}
     V(\xv) = \sum_{\substack{\gv \in \overline B\\ \gv \ne 0}}
                 V_\gv \re^{\i \<\gv, \xv\>} \epc \qd
     V_\gv = \frac{1}{V_u} \int_{U} \rd^3 x \:
                         \re^{- \i \<\gv, \xv\>} V(\xv) \epp
\end{equation}
Here we took the liberty to set $V_0 = 0$ which just fixes the
zero point of the energy. In order to be able to apply perturbation
theory we assume that $V(\xv)$ is a weak periodic potential in the
sense that $V_\gv = \la v_\gv$, where $|v_\gv|$ is uniformly bounded
in $\la$ and $|\la| \ll 1$.

Since $[H, U_\Rv] = 0$ for all $\Rv \in B$ the assumptions of the
Bloch theorem (cf.\ section~\ref{sec:bloch}) are fulfilled. Hence,
the eigenfunctions $\Ps_\kv$ of $H$ have definite lattice momentum
$\kv \in BZ$ and are of the form
\begin{equation} \label{elblochfun}
     \Ps_\kv (\xv) = \re^{\i \<\kv, \xv\>} u_\kv (\xv) \epc
\end{equation}
where $u_\kv (\xv) = u_\kv (\xv + \Rv)$ for all $\Rv \in B$ and
thus has a Fourier series representation
\begin{equation}
     u_\kv (\xv) = \sum_{\gv \in \overline B}
                   u_{\kv, \gv} \re^{- \i \<\gv, \xv\>} \epc \qd
     u_{\kv, \gv} = \frac{1}{V_u} \int_{U} \rd^3 x \:
                         \re^{\i \<\gv, \xv\>} u_\kv (\xv) \epp
\end{equation}
Inserting this back into (\ref{elblochfun}) we see that
\begin{equation} \label{elblochfoufun}
     \Ps_\kv (\xv) =
        \sum_{\gv \in \overline B} u_{\kv, \gv} \re^{\i \<\kv - \gv, \xv\>} \epp
\end{equation}

If we substitute this representation into the Schr\"odinger equation
for $H$ and use (\ref{elpotfou}), we obtain
\begin{equation}
     \sum_{\gv'' \in \overline B} \biggl(\2 \|\kv - \gv''\|^2 - \e (\kv)\biggr)
        u_{\kv, \gv''} \re^{- \i \<\gv'',\xv\>} +
     \sum_{\substack{\gv', \gv'' \in \overline B\\ \gv' \ne 0}}
        \la v_{\gv'} u_{\kv, \gv''} \re^{\i \<\gv' - \gv'',\xv\>} = 0 \epp
\end{equation}
Here we multiply by $\re^{\i \<\gv,\xv\>}/V_u$ and integrate over
$\xv$ over the unit cell. Then
\begin{equation} \label{fouriereva}
     \biggl(\2 \|\kv - \gv\|^2 - \e (\kv)\biggr) u_{\kv, \gv} +
     \la \sum_{\substack{\gv' \in \overline B\\ \gv' \ne \gv}}
         v_{\gv' - \gv} u_{\kv, \gv'} = 0 \epp
\end{equation}
This is an eigenvalue problem for the vector $\uv_\kv$ of the Fourier
components $u_{\kv, \gv}$, $\gv \in \overline B$. We will study it
perturbatively in $\la$.

The reference point are free electrons, $\la = 0$. Then
\begin{equation}
     \biggl(\2 \|\kv - \gv\|^2 - \e (\kv)\biggr) u_{\kv, \gv} = 0 \qd
        \forall\ \gv \in \overline B \epp
\end{equation}
Clearly, the solutions of this eigenvalue problem are
\begin{equation}
     \e_\hv (\kv) = \2 \|\kv - \hv\|^2 \epc \qd u_{\kv, \gv} = \de_{\gv, \hv}
\end{equation}
for all $\hv \in \overline B$. They are parameterized by reciprocal
lattice vectors.

Example: free electrons in 1d. If $a > 0$ is the lattice constant, then
\begin{equation}
     \overline B = \Bigl\{g \in {\mathbb R} \Big|
                          g = n \frac{2\p}{a}\epc\ n \in {\mathbb Z}\Bigr\} \epc \qd
     BZ = \Bigl[- \frac{\p}{a}, \frac{\p}{a}\Bigr] \epp
\end{equation}
It follows that
\begin{equation}
     \e_0 (k) = \frac{k^2}{2} \epc \qd
     \e_{\pm \frac{2\p}{a}} (k) = \2 \Bigl(k \mp \frac{2\p}{a}\Bigr)^2 \epc \dots \epc
\end{equation}
where $k \in BZ$, are the different branches of the dispersion
relation. The situation is sketched in Figure~\ref{fig:dipsfreeelec}.
\begin{figure}
\begin{center}
\includegraphics[width=.82\textwidth]{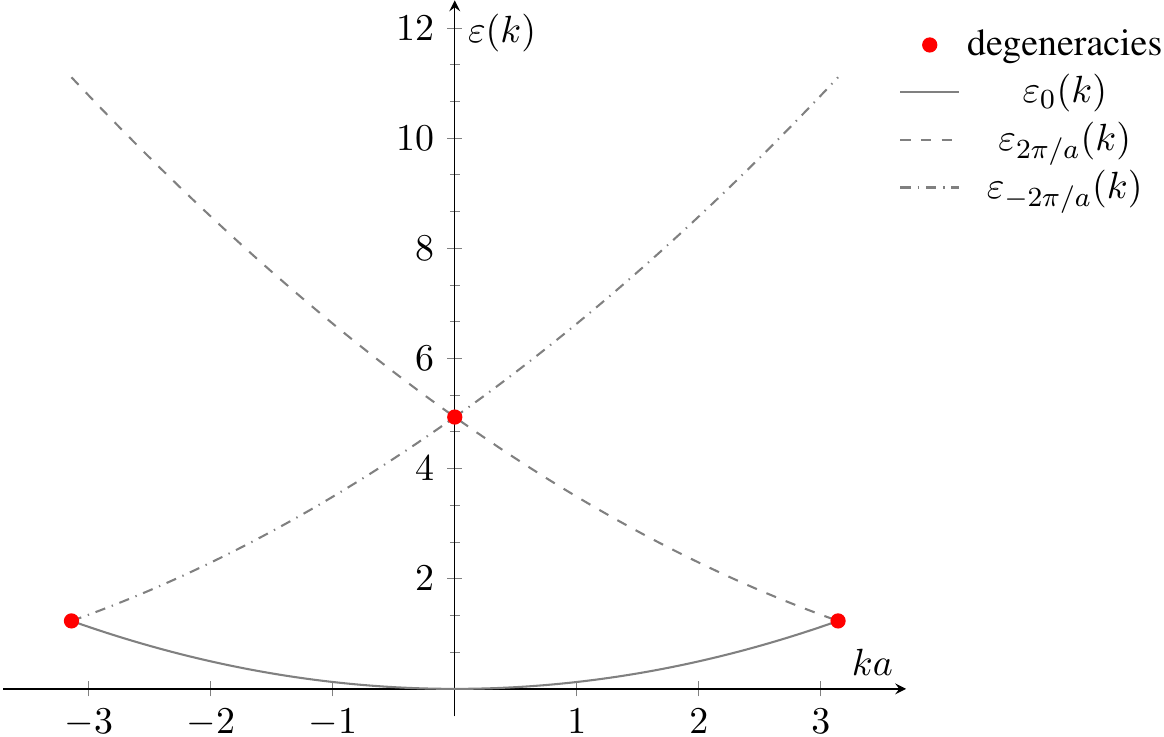}
\caption{\label{fig:dipsfreeelec} 
The dispersion relation of free particles in 1d, conceived as
particles in a periodic potential. The parabolic dispersion
relation splits up in branches, obtained by `back-folding'
into the Brillouin zone. The picture shows the three lowest
lying branches for the dispersion $\e(k) = k^2/2m$ with $m = 4$.
The red dots denotes degenerate points in the spectrum which
belonged to different momenta in the non-periodic picture,
but to the same lattice momentum in the periodic picture.
}
\end{center}
\end{figure}

For non-zero potential we assume that the eigenvalues and the
Fourier coefficients of the wave functions can be expanded in an
asymptotic series in $\la$,
\begin{subequations}
\label{fourierpert}
\begin{align} \label{fourierperten}
     \e_\hv (\kv) & = \2 \|\kv - \hv\|^2 + \la \e^{(1)}_\hv (\kv) 
                      + \la^2 \e^{(2)}_\hv (\kv)
		      + {\cal O} \bigl(\la^3\bigr) \epc \\[1ex]
		      \label{fourierpertwv}
     u_{\kv, \gv} & = \de_{\gv, \hv} 
                      + \la u^{(1)}_{\kv, \gv}
		      + {\cal O} \bigl(\la^2\bigr) \epp
\end{align}
\end{subequations}
We choose the normalization of $\uv_\kv$ such that $u^{(1)}_{\kv, \hv} = 0$.
Then (\ref{fourierpert}) in (\ref{fouriereva}) for $\gv = \hv$ implies
\begin{equation} \label{etwoone}
     - \la \e^{(1)}_\hv (\kv) - \la^2 \e^{(2)}_\hv (\kv)
        + \la \sum_{\substack{\gv' \in \overline B\\ \gv' \ne \hv}}
          v_{\gv' - \hv} \la u^{(1)}_{\kv, \gv'} = {\cal O} \bigl(\la^3\bigr) \epp
\end{equation}
Thus,
\begin{equation}
     \e^{(1)}_\hv (\kv) = 0 \epc \qd
     \e^{(2)}_\hv (\kv) = 
        \sum_{\substack{\gv \in \overline B\\ \gv \ne \hv}}
        v_{\gv - \hv} \la u^{(1)}_{\kv, \gv} \epc
\end{equation}
whence, for $\gv \ne \hv$,
\begin{multline}
     \biggl(\2 \|\kv - \gv\|^2 - \2 \|\kv - \hv\|^2
           - \la^2 \e_\hv^{(2)}(\kv) - \dots \biggr)
     \bigl(\la u^{(1)}_{\kv, \gv} + \dots \bigr) + \\[1ex]
        \la \sum_{\substack{\gv' \in \overline B\\ \gv' \ne \gv}}
        v_{\gv' - \gv} \bigl(\de_{\gv', \hv} + \la u^{(1)}_{\kv, \gv'} + \dots \bigr) = 0 \epc
\end{multline}
implying that
\begin{equation}
     u^{(1)}_{\kv, \gv} = \frac{2 v_{\hv - \gv}}{\|\kv - \hv\|^2 - \|\kv - \gv\|^2}
\end{equation}
if $\|\kv - \gv\| \ne \|\kv - \hv\|$. If we insert this back into
(\ref{etwoone}) and use that $v_{-\gv} = v_\gv^*$ we obtain
the second order corrections to the energies,
\begin{equation} \label{scdorderpertenel}
     \e_\hv (\kv) = \2 \|\kv - \hv\|^2
        + \la^2 \sum_{\substack{\gv \in \overline B\\ \gv \ne 0}}
                  \frac{2 |v_{\gv}|^2}{\|\kv - \hv\|^2 - \|\kv - \hv - \gv\|^2}
        + {\cal O} \bigl(\la^3\bigr) \epp
\end{equation}
As we see, for a weak periodic potential, the energy in second
order perturbation theory can be expressed in terms of the
Fourier coefficients of the potential.

If
\begin{equation} \label{singcon}
     \|\kv - \hv\| = \|\kv - \hv - \gv\|
\end{equation}
for some $\gv \in \overline B$, $\gv \ne 0$ and some $\kv \in BZ$,
then the expression (\ref{scdorderpertenel}) is singular at
that specific value of $\kv$. In order to interpret this problem
set $\hv = 0$. Then (\ref{singcon}) reduces to
\begin{equation} \label{singconbz}
     \|\kv\| = \|\kv - \gv\| \epp
\end{equation}
If $\gv$ runs through the nearest-neighbour sites of the origin
in $\overline B$, this equation describes the boundaries of the
Brillouin zone. Hence, we expect that the exact solution of the
problem exhibits a strong deviation of the dispersion relation
from the dispersion relation of free electrons at the 
boundary of the Brillouin zone. For $\hv \ne 0$ additional
singular manifolds appear due to back-folding into the
Brillouin zone.

For the interpretation of the perturbative result we further
recall that the perturbation theory it applied to each individual
energy level, in our case for fixed $\hv$ and $\kv$.
Equation (\ref{scdorderpertenel}) is valid for all $\hv$,
$\kv$ which do not satisfy (\ref{singcon}). If (\ref{singcon})
is satisfied for a pair $\hv$, $\kv$, we have to modify our
calculation, applying the scheme of degenerate perturbation
theory instead.

\mysection{Particles in a periodic potential}
\subsection{Degenerate levels}
Fix $\kv$ and assume that there are precisely two vectors
$\hv, \hv' \in \overline B$, $\hv \ne \hv'$, such that
\begin{equation}
     \|\kv - \hv\| = \|\kv - \hv'\| \epp
\end{equation}
Then (\ref{fouriereva}) has two degenerate zeroth order solutions
with energies
\begin{equation}
     \e_\hv (\kv) = \2  \|\kv - \hv\|^2 = \2 \|\kv - \hv'\|^2 = \e_{\hv'} (\kv) \epp
\end{equation}
and Fourier coefficients
\begin{equation}
     u_{\kv, \gv} = a \de_{\gv, \hv} + b \de_{\gv, \hv'}
\end{equation}
where $\av, \bv \in {\mathbb C}$ are to be determined. The
space of solutions is two-dimensional. We assume again an
asymptotic dependence of the dispersion relations on the
interaction parameter $\la$ as in (\ref{fourierperten}).
The corresponding ansatz for the Fourier coefficients
(\ref{fourierpertwv}) has to be modified due to the degeneracy,
\begin{equation}
     u_{\kv, \gv} = a \de_{\gv, \hv} + b \de_{\gv, \hv'} + {\cal O} (\la) \epp
\end{equation}
Substituting this into (\ref{fouriereva}) for $\gv = \hv, \hv'$
and comparing coefficients at order $\la$ we obtain
\begin{align}
     - \e^{(1)} (\kv) a + v_{\hv' - \hv} b = {\cal O} (\la) \epc \notag \\[1ex]
     - \e^{(1)} (\kv) b + v_{\hv - \hv'} a = {\cal O} (\la) \epp
\end{align}
The solvability condition for this homogeneous system implies
that $\e^{(1)} (\kv) = \pm |v_{\hv - \hv'}|$. Thus, the two
degenerate energy levels split into
\begin{equation}
     \e^\pm (\kv) = \2 \|\kv - \hv\|^2 \pm \la |v_{\hv - \hv'}|
                    + {\cal O} (\la^2) \epp
\end{equation}

By way of contrast to the non-degenerate case, the corrections
to the eigenstates are now of first order in $\la$. This
means that close to the Brillouin zone boundaries the
analytic structure of the corrections change. As mentioned
above, this makes only sense if we consider a finite system
under periodic boundary conditions. Still, we take it as
another indication that the strongest effect on the dispersion
relation of a free particle by a weak periodic perturbation is
at the boundaries of the Brillouin zone.

The possible values of the coefficients $a, b$ follow from
\begin{multline}
     \mp a_\pm |v_{\hv - \hv'}| + b_\pm v_{\hv' - \hv} = 0 \ \Leftrightarrow\
        \frac{a_\pm}{b_\pm} = \pm \frac{v_{\hv' - \hv}}{|v_{\hv - \hv'}|}
	= \pm \re^{2 \i \de} \\[1ex] \Leftarrow\
	\begin{pmatrix} a_+ \\ b_+ \end{pmatrix}
	= \frac{1}{\sqrt{2}} \begin{pmatrix} \re^{\i \de} \\ \re^{- \i \de} \end{pmatrix} \epc \qd
	\begin{pmatrix} a_- \\ b_- \end{pmatrix}
	= \frac{1}{\sqrt{2}} \begin{pmatrix} \re^{\i \de} \\ - \re^{- \i \de} \end{pmatrix} \epp
\end{multline}
\begin{remark}
We have considered a two-fold degeneracy. At special symmetric
points at the boundaries of the Brillouin zone in more than one
dimension higher degeneracies (like four or six) may occur.
\end{remark}

\subsection{\boldmath Example $d = 1$}
\begin{enumerate}
\item
Let $h = 0$ and $h' = 2\p/a$. Then a degeneracy occurs at $k = \p/a$,
and
\begin{equation}
     \e^\pm \Bigl(\frac{\p}{a}\Bigr)
        = \2 \Bigl(\frac{\p}{a}\Bigr)^2 \pm \la |v_{2\p/a}| \epp
\end{equation}
\item
Let $h = 2\p/a$ and $h' = - 2\p/a$. Then the spectrum is degenerate at
$k = 0$,
\begin{equation}
     \e^\pm (0) = \2 \Bigl(\frac{2\p}{a}\Bigr)^2 \pm \la |v_{4\p/a}| \epp
\end{equation}
\end{enumerate}
The two cases are sketched in Figure~\ref{fig:onedbandgap}. We observe
the `opening of band gaps' at the degenerate points in the dispersion
relation. Band gaps are most characteristic for the phenomenology of
crystalline solids.
\begin{figure}
\begin{center}
\includegraphics[width=.82\textwidth]{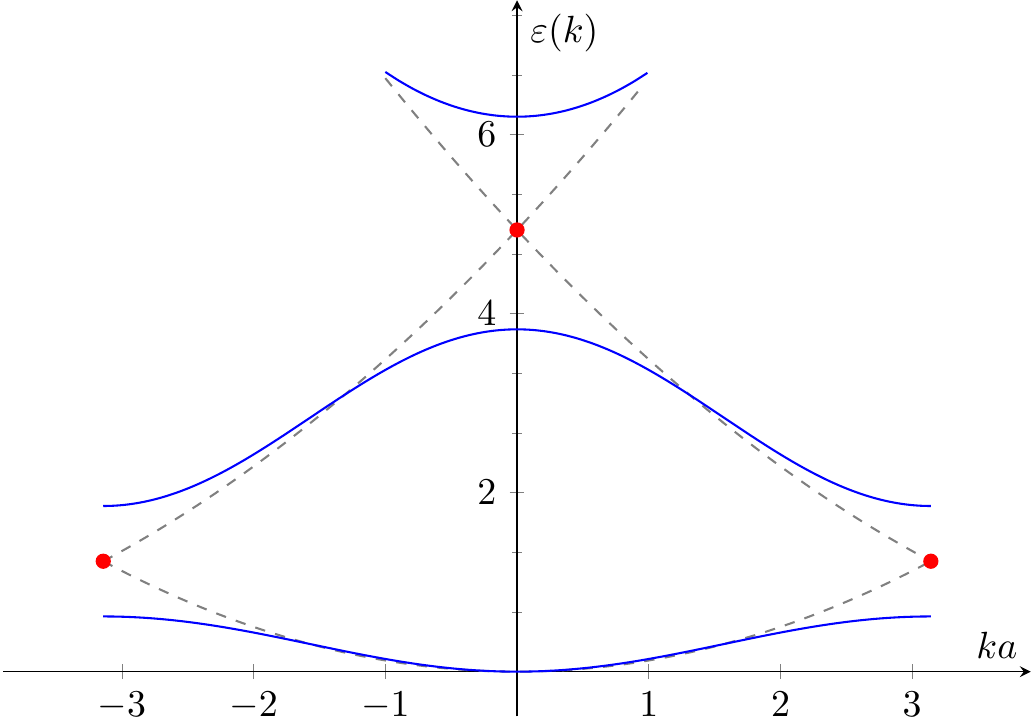}
\caption{\label{fig:onedbandgap} 
The opening of band gaps at degenerate points of the dispersion
relation schematically for almost free, periodically perturbed
electrons in 1d. The dashed lines represent the dispersion of
free particles parameterized by the lattice momentum and are
the same as in Figure~\ref{fig:dipsfreeelec}. The blue line
represent the deformation of the free-particle dispersion under
the influence of a periodic perturbation.
}
\end{center}
\end{figure}

In order to develop more intuition we discuss the wave functions
connected with the first case above. These are
\begin{multline}
     \Ps^\pm (x) = a_\pm \re^{\i \p x/a} + b_\pm \re^{- \i \p x/a} + {\cal O} (\la)
        = \frac{1}{\sqrt{2}}
	  \bigl(\re^{\i (\p x/a + \de)} \pm \re^{- \i (\p x/a + \de)}\bigr)
	  + {\cal O} (\la) \\[1ex]
        = \sqrt{2} \begin{cases} \cos(\p x/a + \de) \\
	                         \i \sin(\p x/a + \de) \end{cases} + {\cal O} (\la) \epp
\end{multline}
It follows that
\begin{equation} \label{pspmsquarealmost}
     |\Ps^\pm (x)|^2 = 1 \pm \cos \bigl(2(\p x/a + \de)\bigr) \epp
\end{equation}
The periodic potential $V$, on the other hand, has the Fourier series
representation
\begin{multline} \label{potentialalmost}
     V(x) = \la \bigl(v_{- 2\p/a} \re^{- i 2 \p x/a} + v_{2 \p/a} \re^{\i 2 \p x/a}\bigr)
            + \text{higher Fourier modes} \\[1ex]
          = 2 \la |v_{2 \p /a}| \cos \bigl(2(\p x/a + \de)\bigr)
            + \text{higher Fourier modes} \epp
\end{multline}
Drawing (\ref{pspmsquarealmost}) and (\ref{potentialalmost}) in
the same picture and choosing $\la < 0$ (attractive effective
potential near origin) we see that for $\Ps^+$ the particle is
in the average in the valleys of the potential, whereas it is
more on the hills for $\Ps^-$. Accordingly, $\e^+ (\p/a) <
\e^- (\p/a)$ in this case.

\subsection{The tight-binding method}
\label{sec:tbmethod}
So far we have studied the formation of energy bands starting from
free electrons (plane waves) subject to a periodic perturbation.
Now we would like to turn to the opposite extreme. We shall start
with atomic wave functions and ask what happens, when the atoms are
brought close to each other. For simplicity we assume a mono-atomic
lattice with $N_{\rm at}$ atoms.
\begin{enumerate}
\item
Let $\Ph_a (\xv)$ an atomic wave function of an electron. In order
to satisfy Bloch's theorem we consider the linear combination
\begin{equation} \label{tbwf}
     \ph_{a \kv} (\xv) = \frac{1}{\sqrt{N_{\rm at}}}
                     \sum_{\Rv \in B} \re^{\i \<\kv, \Rv\>} \Ph_a (\xv - \Rv) \epc
\end{equation}
$\kv \in BZ$.
\item
Further define
\begin{subequations}
\begin{align}
     j(\Rv - \Rv') & = \int_V \rd^3 x \: \Ph_a^* (\xv - \Rv') \Ph_a (\xv - \Rv) \epc \\[1ex]
     h(\Rv - \Rv') & = \int_V \rd^3 x \: \Ph_a^* (\xv - \Rv') (H \Ph_a) (\xv - \Rv) \epc
\end{align}
\end{subequations}
where $H$ is the one-particle Hamiltonian (\ref{hamper}). Since
the atomic wave functions decay exponentially fast with the
distance from the nucleus, we expect these functions to behave as
\begin{subequations}
\begin{align}
     j(\Rv - \Rv') & \sim \de_{\Rv, \Rv'}  \epc \\[1ex]
     h(\Rv - \Rv') & \sim \de_{\Rv, \Rv'} h(0) \epc
\end{align}
\end{subequations}
if the interatomic distance becomes large. It follows that, asymptotically
for large interatomic distances,
\begin{subequations}
\begin{align}
     \<\ph_{a \kv}, \ph_{a \kv'}\> & = \de_{\kv, \kv'} \epc \\[1ex]
     \<\ph_{a \kv}, H \ph_{a \kv'}\> & = h(0) \de_{\kv, \kv'} \epp
\end{align}
\end{subequations}
Thus, for large interatomic distances, the functions $\ph_{a \kv}$
are a set of approximate eigenfunctions of $H$ corresponding
to an $N_{\rm at}$-fold degenerate atomic level.
\item
This highly degenerate level splits under the influence of the
mutual perturbations of the atoms, if they come closer to each
other. In order to take into account the perturbation we determine
the norms and energy expectation value in these states,
\begin{multline}
     \|\ph_{a \kv}\|^2 = \int_V \rd^3 x \: |\ph_{a \kv} (\xv)|^2
        = \frac{1}{N_{\rm at}} \sum_{\Rv, \Rv' \in B} \re^{\i \<\kv,\Rv - \Rv'\>}
	  j(\Rv - \Rv') \\
        = 1 + \sum_{\substack{\Rv \in B\\ \Rv \ne 0}} \re^{\i \<\kv,\Rv\>} j(\Rv) \epp
\end{multline}
Note that the sum on the right hand side vanishes for large lattice
spacing. For the expectation value of the energy we obtain in a
similar way
\begin{multline}
     E(\kv) = \int_V \rd^3 x \: \frac{\ph_{a \kv}^* (\xv) H \ph_{a \kv} (\xv)}{\|\ph_{a \kv}\|^2}
        = \frac{1}{\|\ph_{a \kv}\|^2 N_{\rm at}} \sum_{\Rv, \Rv' \in B} \re^{\i \<\kv,\Rv - \Rv'\>}
	  h(\Rv - \Rv') \\[1ex]
        = \frac{1}{\|\ph_{a \kv}\|^2} \sum_{\Rv \in B} \re^{\i \<\kv,\Rv\>} h(\Rv)
        = \frac{1}{\|\ph_{a \kv}\|^2} \biggl(h(0) +
	      \sum_{\substack{\Rv \in B\\ \Rv \ne 0}} \re^{\i \<\kv,\Rv\>} h(\Rv) \biggr) \epc
\end{multline}
where again the sum in the brackets on the right hand side vanishes for
large lattice spacing.
\item
Let us now assume that the functions $j(\Rv)$ and $h(\Rv)$ decrease
rapidly with increasing distance from the origin in $B$. Then
\begin{subequations}
\begin{align}
     1 \gg j(\Rv_{\rm nn}) \gg j(\Rv_{nnn}) \gg \dots \epc \\[1ex]
     |h(0)| \gg h(\Rv_{\rm nn}) \gg h(\Rv_{nnn}) \gg \dots \epc
\end{align}
\end{subequations}
where `${\rm nn}$' refers to nearest neighbours to the origin,
`${\rm nnn}$' to next-to-nearest neighbours etc. It follows that
\begin{equation}
     E(\kv) = h(0) + \sum_{\Rv \in \{\Rv_{\rm nn}\} \subset B} \re^{\i \<\kv, \Rv\>}
                     \bigl(h(\Rv) - h(0) j(\Rv)\bigr) + \text{${\rm nnn}$ terms} \epp
\end{equation}
This formula describes the so-called `tight-binding bands'
which give a realistic description of bands with `pronounced
atomic character' for which the electrons are close to 
the atoms. We shall denote
\begin{equation}
     t(\Rv) = h(\Rv) - h(0) j(\Rv) \epp
\end{equation}
We conclude from the inversion symmetry of the Bravais lattice 
that $t(\Rv) = t(- \Rv)$. For this reason tight-binding bands
are sums over cosines.
\item
Example fcc lattice: The fcc lattice has 12 nearest neighbours
to the origin located at
\begin{equation}
     \Rv = \frac{a}{2}
           \begin{cases}
	      (\pm 1, \pm 1, 0) \\ (0, \pm 1, \pm 1) \\ (\pm 1, 0, \pm 1)
	   \end{cases} \epc
\end{equation}
and $t(\Rv) = t$ by symmetry. Hence,
\begin{multline}
     E(\kv) = h(0) + t \sum_{\s, \s' = \pm} \bigl(
                       \re^{\i a(\s k_x + \s' k_y)/2}
                       + \re^{\i a(\s k_y + \s' k_z)/2}
                       + \re^{\i a(\s k_x + \s' k_z)/2} \bigr) \\[1ex]
            = h(0) + 4 t \biggl\{\cos\Bigl(\frac{ak_x}{2}\Bigr)
	                         \cos\Bigl(\frac{ak_y}{2}\Bigr) +
				 \cos\Bigl(\frac{ak_y}{2}\Bigr)
				 \cos\Bigl(\frac{ak_z}{2}\Bigr) \\ +
				 \cos\Bigl(\frac{ak_z}{2}\Bigr)
				 \cos\Bigl(\frac{ak_x}{2}\Bigr) \biggr\} \epp
\end{multline}
Note that the band width is proportional to $t$. This means that
small overlaps of the wave function induces narrow bands. The
tight-binding model is thus expected to provide a good description
of narrow energy bands in solids.
\item
We close this section with two remarks. First, the tight binding
bands can be better justified by introducing so-called Wannier
orbitals instead of atomic orbitals (see below).

Second, perhaps the most important insight we can gain from the
above reasoning is the intuitive physical picture. Under the
influence of the mutual interaction an $N_{\rm at}$-fold degenerate
energy level splits into a band with $N_{\rm at}$ states. This means
that bands can be classified according to the character of the
underlying atomic orbitals as $s, p, d, f$ bands. Our calculation
above makes sense for an isolated $s$-orbital. For $p, d, f$ orbitals
we cannot start with a single atomic state $\Ph_a$, but have to
take into account a set $\{\Ph_a (\xv)\}_{a=1}^{a_{\rm max}}$ of
atomic states. This corresponds to a combination of the LCAO (linear
combination of atomic orbitals) method of molecular physics with
the tight-binding method. Like in molecular physics hybrid orbitals
(like $s$-$d$ orbitals) may appear in the general case.
\end{enumerate}

\mysection{Electrons in a periodic potential}
\subsection{Method of orthogonalized plane waves}
The tight-binding method works for low-energy narrow bands
formed by atomic states in which, in the average, the electrons
are close to the ions. A method for the calculation of more
realistic band structure is obtained by combining the tight
binding method with the method of almost free electrons. This
is called the OPW (orthogonalized plane waves) method.

We assume that the low-lying states are known. They may be, for
instance, sufficiently well described by the tight-binding wave
function (\ref{tbwf}),
\begin{equation}
     H \ph_{a \kv} \simeq E_{\rm tb} (\kv) \ph_{a \kv} \epc \qd
     E_{\rm tb} (\kv) = \sum_{\Rv_{\rm nn} \in B} \re^{\i \<\kv, \Rv\>} t(\Rv) \epp
\end{equation}
Let
\begin{equation}
     P = \sum_{\kv \in BZ} P_\kv \epc
\end{equation}
the projector onto the subspace of the full Hilbert space that
is spanned by the $\ph_{a \kv}$. In order to determine the remaining
part of the spectrum it suffices to consider
\begin{multline}
     H(1 - P) \Ps (\xv) = E(1 - P) \Ps (\xv)\\[1ex] \Leftrightarrow
     H \Ps (\xv) + (E - H) P \Ps (\xv)
     =  H \Ps (\xv) + \sum_{\kv \in BZ} (E - E_{\rm tb} (\kv)) P_\kv \Ps (\xv) \\
     =  \bigg[\frac{p^2}{2} + V(\xv)
              + \sum_{\kv \in BZ} (E - E_{\rm tb} (\kv)) P_\kv \biggr]
	      \Ps (\xv) = E \Ps (\xv) \epp
\end{multline}
Here
\begin{equation}
     W(E,\xv) = V(\xv) + \sum_{\kv \in BZ} (E - E_{\rm tb} (\kv)) P_\kv
\end{equation}
is called a `pseudo potential'. The pseudo potential is not a
potential, since in position representation it is represented
by an integral operator. Note that
\begin{equation}
     W(E,\xv) - V(\xv) > 0 \epc
\end{equation}
since $E > E_{\rm tb}$. We interpret this in such a way that
$W(E, \xv)$ includes the effects of screening as discussed above.
It is therefore a more appropriate starting point for a perturbation
theory for almost free electrons.
\subsection{Other methods}
Within the `augmented plane wave method' the Schrödinger equation
is solved in spheres around the ions and plane waves are fitted into
the space between the spheres (were the potential is assumed to
be negligible).

The KKR (Korringa, Kohn, Rostocker) method is a variant of the
augmented plane wave method, where, in a first step, the Green
function of the Laplace operator is used in order to transform the
Schr\"odinger equation into an integral equation.

\subsection{Summary}
\begin{enumerate}
\item
For the understanding of many of the electronic properties of
solids it suffices to take into account the interaction of the
electrons only in so far as they screen the attractive potential
of the core ions. The remaining problem is the problem of independent
electrons in a periodic potential.
\item
Electrons in a periodic potential are characterized by the branches
$\e_n (\kv)$, $n \in {\mathbb N}$, of their dispersion relation.
As opposed to the spectral problem of phonons the number of
branches for the electrons is infinite. The branches are called
energy bands, their entirety is called the `band structure' of
the solid.
\item
Like for the phonons the $\e_n (\kv)$ become differentiable functions
of $\kv$ in the thermodynamic limit which exhibit the full
translation symmetry of the reciprocal lattice and the full point
group symmetry of the solid.
\end{enumerate}

\subsection{Exercise 15. Electronic band structure in
one-dimensional solids by WKB}
The Schr\"odinger equation for a non-relativistic electron of mass $M$ in a
1d periodic potential of period $L$ reads
\begin{equation} \label{sch}
     \frac{d^2\psi}{d x^2} + [E - V(x)]\psi(x) = 0 \epp
\end{equation}
Here the energy and the length are measured in units of $\frac{\hbar^2}{2 M L^2}$
and $L$, respectively.

We would like to solve equation \eqref{sch} by means of the
WKB-approximation and find a condition which determines the band structure,
i.e.\ all allowed energy values $E$.
\begin{center}
\includegraphics[width=.95\textwidth]{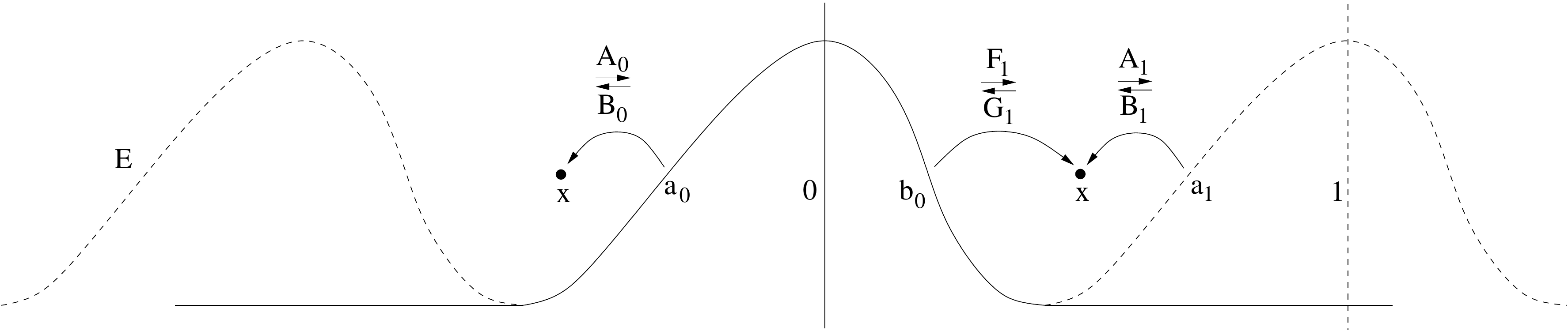}
\end{center}
Setting $p(x) = \sqrt{E-V(x)}$ we can express the general WKB-solution
of the Schr\"odinger equation for \emph{a single} potential barrier $V(x)$
(see figure) in the classically accessible regions left and right of the barrier as
\begin{align*}
   & E < V_{max}: \frac{1}{\sqrt{p(x)}} \exp [ \pm\i \int_{a_0}^x \!\!\!  \rd y p(y)]
   \epc\ x<a_0 \epc \qd
   \frac{1}{\sqrt{p(x)}} \exp [ \pm\i \int_{b_0}^x \!\!\!  \rd y p(y)]
   \epc\ x > b_0\\[1ex]
   & E > V_{max}: \frac{1}{\sqrt{p(x)}} \exp [ \pm\i \int_{0}^x \!\!\!  \rd y p(y)]
   \epc\ x<0 \epc \qd \frac{1}{\sqrt{p(x)}} \exp [ \pm\i \int_{0}^x \!\!\!  \rd y p(y)]
   \epc\ x > 0 \epp
\end{align*} 
\begin{enumerate}
\item
The general solution is in each case a linear combination with coefficients 
$A_0$, $B_0$ or $F_1$, $G_1$ for the left and right classically accessible
region, respectively. The coefficients are related by a $2\times2$ matrix $M$.
Conclude from the time-inversion invariance that $M_1^1 = {M_2^2}^*$,
$M_1^2 = {M_2^1}^*$. For this purpose consider 
\begin{equation}
     M \begin{pmatrix} A_0\\ B_0 \end{pmatrix} = 
        \begin{pmatrix} F_1\\ G_1 \end{pmatrix} \epc \qquad
     M \begin{pmatrix} \overline{B}_0\\ \overline{A}_0 \end{pmatrix} = 
        \begin{pmatrix} \overline{G}_1\\ \overline{F}_1 \end{pmatrix} \epp
\end{equation}
\item
The current conservation $|A_0|^2 - |B_0|^2 = |F_1|^2 - |G_1|^2$ implies $\det M = 1$.
Let $T$ be the transmission coefficient and $R=1-T$ the reflection coefficient with
the corresponding phase shifts $\re^{\i\mu}$ and $\re^{-\i\nu}$. Verify the representation
\begin{equation} \label{wall}
M =
\begin{pmatrix}
 \frac{1}{\sqrt{T}} \re^{\i\mu} & - \sqrt{\frac{R}{T}} \re^{\i(\mu+\nu)} \\
 -  \sqrt{\frac{R}{T}}\re^{-\i(\mu+\nu)}& \frac{1}{\sqrt{T}} \re^{-\i\mu}
\end{pmatrix} \epp
\end{equation}
\item
A periodic continuation of the potential barrier leads to another representation of
the solution in the classically accessible domain $b_0 < x < a_1$,
\begin{align}
E < V_{max}: \quad \frac{A_1}{\sqrt{p(x)}} \exp [ +\i \int_{a_1}^x \!\!\!  \rd y p(y)] &+
\frac{B_1}{\sqrt{p(x)}} \exp [ -\i \int_{a_1}^x \!\!\!  \rd y p(y)] \epc \\[1ex]
E > V_{max}: \quad \frac{A_1}{\sqrt{p(x)}} \exp [ +\i \int_{1}^x \!\!\!  \rd y p(y)] &+
\frac{B_1}{\sqrt{p(x)}} \exp [ -\i \int_{1}^x \!\!\!  \rd y p(y)] \epp
\end{align}
Which matrix connects $A_1$, $B_1$ with $F_1$, $G_1$, if we set $\phi(E) =
\int_{b_0}^{a_1} \!\!  \rd y p(y),\,E < V_{max}$ rsp.\ $\phi(E) =
\int_{0}^{1} \!\!  \rd y p(y),\,E \geq V_{max}$? Calculate the transfer matrix 
$P$ which links the coefficients $A_0$, $B_0$ and $A_1$, $B_1$ of two neighbouring cells.
\item
The full solution for the periodic potential stays bounded as long as $P$ has
eigenvalues of absolute value 1. Show that this fact implies a condition
that determines the band structure,
\begin{equation}
\left| \frac{\cos [\phi(E) + \mu]}{\sqrt{T(E)}} \right| \leq 1 \epp
\end{equation}
\item
For a potential of the form $V(x) = V_0 \cos(2\pi x)$, $V_0 > 0$, the Schr\"odinger
equation \eqref{sch} of the non-relativistic particle is equal to the Mathieu
equations. The transition coefficient $T(E)$ of {\it a single} potential barrier 
in WKB-approximation follows as
\begin{align*}
T(E) &= \frac{1}{1 + \re^{2W}} \epc & W(E) &= \int_{a_0}^{b_0} \!\!\!\! \rd y | p(y) |
\epc \qd E < V_{0} \epc \\[1ex]
T(E) &= \frac{1}{1 + \re^{-2W}} \epc & W(E) &= \bigg| \int_{y_1}^{y_2} \!\!\!\! \rd y p(y) \bigg|
\epc \qd E > V_{0} \epc
\end{align*}
where for $E>V_0$ the integral has to be calculated on the direct line between
both imaginary reversal points. Represent $\phi(E)$ and $W(E)$ by
complete elliptical integrals of the first and second kind $K(m)$ and $E(m)$,
with the dimensionless parameter $m=\frac{E+V_0}{2V_0}$, $E < V_0$ and
$m=\frac{2V_0}{E+V_0}$, $E > V_0$, respectively!
\item
Choosing different parameters $m$ (and thus fixing $E/V_0$) it is possible to
to sketch the regions in the $V_0$,$E$-diagram for $\mu \approx 0$, where
the solutions are bounded. Find such type of diagram in the literature.
\end{enumerate}

\subsection{The Fermi distribution}
Electrons are Fermions. According to the Pauli principle many-electron
wave functions must be totally anti-symmetric under the permutations
of any two electrons. Consequentially, in a system of many non-interacting
electrons (or holes) no two of them can be in the same single-particle
state. When we calculate the grand canonical partition function we
therefore have to count every single-particle state as either unoccupied
or occupied by just one electron (or hole),
\begin{equation} \label{elpartfun}
     Z_{\rm e} = \prod_{n=0}^\infty \prod_{\kv \in BZ}
                 \bigl(1 + \re^{- \frac{\e_n (\kv) - \m}{T}}\bigr) \epp
\end{equation}
Here $1 = \re^0$ stands for an unoccupied state, while
$\re^{- \frac{\e_n (\kv) - \m}{T}}$ stands for an occupied state.
Every factor on the right hand side of (\ref{elpartfun}) 
represents the partition function corresponding to a single-electron
state. Hence, the grand-canonical probability for having this state
occupied is
\begin{equation}
     f(\e_n (\kv) - \m) =
        \frac{\re^{- \frac{\e_n (\kv) - \m}{T}}}
	     {1 + \re^{- \frac{\e_n (\kv) - \m}{T}}} =
        \frac{1}{\re^{\frac{\e_n (\kv) - \m}{T}} + 1} \epp
\end{equation}
As we recall from our lecture on statistical mechanics, this is the
Fermi distribution function.

\subsection{Grand canonical potential of the electron gas}
As for every ideal gas of spin-$\2$ Fermions we can immediately
write down the total particle number $N$, internal energy $E$ and
entropy $S$ of the system as sums involving the Fermi function,
\begin{subequations}
\begin{align}
     N(T,\m) & = 2 \sum_{n=0}^\infty \sum_{\kv \in BZ} f(\e_n (\kv) - \m) \epc \\
     E(T,\m) & = 2 \sum_{n=0}^\infty \sum_{\kv \in BZ} f(\e_n (\kv) - \m) \e_n (\kv) \epc \\
     S(T,\m) & = - 2 \sum_{n=0}^\infty \sum_{\kv \in BZ} \bigl\{
                    f(\e_n (\kv) - \m) \ln \bigl(f(\e_n (\kv) - \m)\bigr)
		    \notag \\[-1ex] & \mspace{136.mu}
		    + (1 - f(\e_n (\kv) - \m)) \ln \bigl(1 - f(\e_n (\kv) - \m)
		    \bigr) \bigr\} \epp
\end{align}
\end{subequations}
Here the factor of 2 accounts for the spin degree of freedom.
The grand canonical potential is obtained as
\begin{multline} \label{grandcanfermi}
     \Om (T, \m) = E(T, \m) - T S(T, \m) - \m N (T, \m) \\
        = 2 T \sum_{n=0}^\infty \sum_{\kv \in BZ}
	    \bigl\{f \ln(1/f -1) + f \ln f + (1 - f)\ln(1 - f)\bigr\} \\[-1ex]
        = - 2 T \sum_{n=0}^\infty \sum_{\kv \in BZ}
	      \ln \bigl(1 + \re^{- \frac{\e_n (\kv) - \m}{T}}\bigr) \epp
\end{multline}

In a similar way as for the phonon gas we can rewrite the sums over
all lattice momenta asymptotically for large volume $V$ first as
an integral over the Brillouin zone and then as an integral over
all energies. For the second step we need to define an electronic
density of states. We proceed as for the phonon gas and first of all
introduce the density of states for a single branch of the dispersion
relation,
\begin{equation} \label{gnfermi}
     g_n (\e) = \frac{1}{(2 \p)^3} \int_{BZ} \rd^3 k \: \de(\e - \e_n (\kv))
              = \frac{1}{(2 \p)^3} \int_{S(\e)} \frac{\rd S}{\|\grad_\kv \e_n (\kv)\|} \epc
\end{equation}
where $S(\e)$ is the surface implicitly determined by the equation
$\e = \e_n (\kv)$. With this the (total) density of states is
defined as
\begin{equation} \label{gfermi}
     g(\e) = 2 \sum_{n=0}^\infty g_n (\e) \epp
\end{equation}
Again a factor of 2 is included to take the spin degrees of freedom
into account.
\begin{remark}
Our statements about van-Hove singularities in section~\ref{sec:vanh}
remain valid for electrons.
\end{remark}

\subsection{Fermi energy and Fermi surface}
Two important notions in solid state physics are the `Fermi
energy' and the `Fermi surface'. Using the density of states
introduced above we may write the particle number as a function
of temperature and chemical potential as
\begin{equation} \label{nofmu}
     N(T, \m) = V \int_{- \infty}^\infty \rd \e \: g(\e) f(\e - \m) \epp
\end{equation}
Since
\begin{equation} \label{dernofmu}
     \6_\m  N(T, \m) = \frac{V}{4T}
          \int_{- \infty}^\infty \rd \e \: \frac{g(\e)}{\ch^2((\e - \m)/2T)} > 0 \epc
\end{equation}
equation (\ref{nofmu}) can be inverted at any $T > 0$ to give
$\m = \m(T, N)$ (the latter fact follows, of course, also
from general arguments on the equivalence of thermodynamic
ensembles in the thermodynamic limit).

The Fermi energy $E_F$ is defined as the chemical potential at $T = 0$,
\begin{equation}
     E_F = \lim_{T \rightarrow 0+} \m (T, N) \epp
\end{equation}
It is clear from (\ref{nofmu}) and (\ref{dernofmu}) that the
Fermi energy depends only on the particle density $N/V$ and
is a monotonically increasing function of the particle density.
In a canonical ensemble description, i.e.\ if we fix the
particle number and consider $\m$ as a function of $N$ and $T$,
we obtain the pointwise pointwise limit,
\begin{equation}
       \lim_{T \rightarrow 0+} f\bigl(\e - \m(T,N)\bigr) = \Th(E_F - \e) \epc
\end{equation}
where $\Th$ is the Heaviside function. This reflects the fact
that in the ground state all single-particle states with
energies up to $E_F$ are occupied, while those with energies
larger than $E_F$ are unoccupied.

For the electronic ground state of solids (within the band model)
exist two alternatives,
\begin{equation}
     \text{(i)}\ g(E_F) > 0 \epc \qd \text{(ii)}\ g(E_F) = 0 \epc
\end{equation}
which come with drastically different phenomenologies. In case (i)
the Fermi energy lies within a band. In case (ii) it is situated
in a band gap (cf.\ Figure~\ref{fig:efermi}). In case (i)
excitations of arbitrarily small energy are possible. In case (ii),
due to the Pauli principle, the smallest possible excitation energy
is equal to the band gap. In case (i) an arbitrarily small electric
field causes a current, in case (ii) this does not happen. This
is our first explanation, within a single-particle picture, of 
the difference between conductors and insulators.
\begin{figure}
\begin{center}
\includegraphics[width=.75\textwidth]{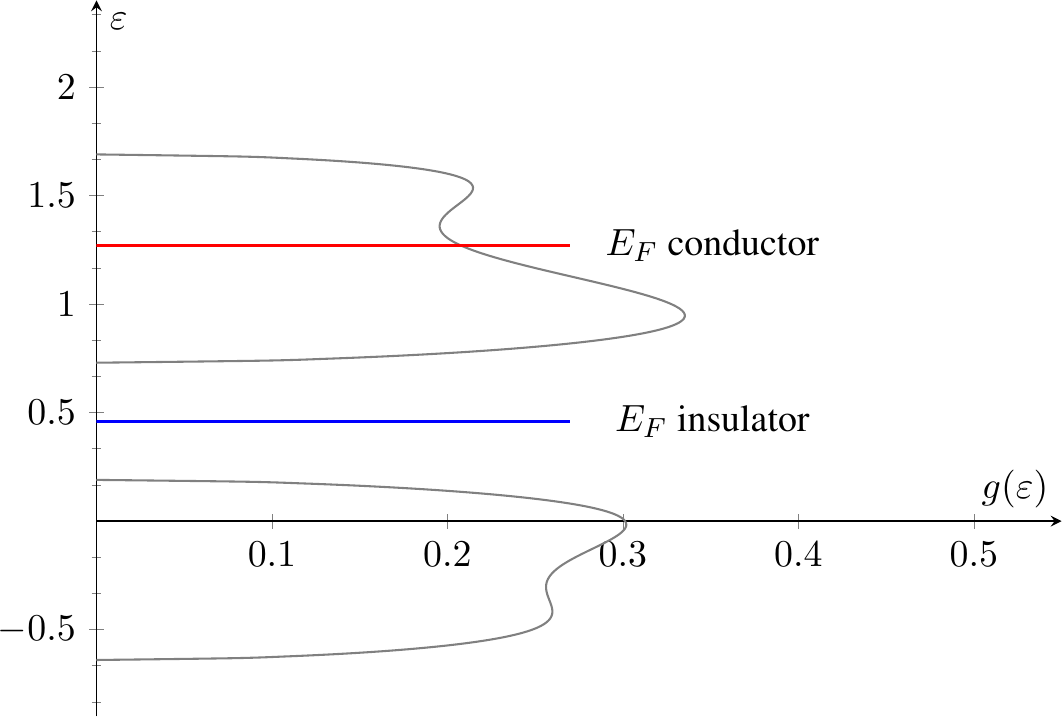}
\caption{\label{fig:efermi} 
Within the single-particle picture or band model a solid is a
conductor or an insulator depending on whether or not the
Fermi energy is situated within a band or in a band gap.
}
\end{center}
\end{figure}

In conductors the Fermi energy is in at least one band, and the
equations
\begin{equation}
     \e_n (\kv) = E_F \qd n \in {\mathbb N}
\end{equation}
determine a surface in reciprocal space which is called the
Fermi surface $S_F$. The volume enclosed by the Fermi surface
\begin{equation}
     \Int S_F = \bigl\{\kv \in {\mathbb R}^3 \Big|
                       \e_n (\kv) \le E_F,\ n \in {\mathbb N}\bigr\}
\end{equation}
is called the Fermi sphere. In general the Fermi surface has
a complicated shape and topology (not necessarily simply connected).
It is of crucial importance for the understanding of the transport
properties of solids, in particular in the presence of a magnetic
field.

\mysection{Low-temperature specific heat of the electron gas}
A first quantitative example, showing that it makes a big difference
whether or not the Fermi energy lies within a band, is the thermodynamics
of band electrons at low temperature. Expressing the grand canonical
potential (\ref{grandcanfermi}) by means of the electronic density
of states (\ref{gnfermi}), (\ref{gfermi}) we obtain
\begin{equation}
     \Om (T, \m) = - TV \int_{-\infty}^\infty \rd \e \:
                        g(\e) \ln\bigl(1 + \re^{- \frac{\e - \m}{T}}\bigr) \epp
\end{equation}
In order to prepare for the low-$T$ analysis we rewrite this as
\begin{equation} \label{omfermiforlowt}
     \Om (T, \m) = V \int_{-\infty}^\m \rd \e \: g(\e) (\e - \m)
                   - TV \int_{-\infty}^\infty \rd \e \:
                        g(\e) \ln\bigl(1 + \re^{- \frac{|\e - \m|}{T}}\bigr) \epp
\end{equation}
Since $g(x)$ is the density of states of an electronic band structure,
there are two interlaced sequences $(a_n)_{n \in {\mathbb N}}$,
$(b_n)_{n \in {\mathbb N}}$ with $a_n < b_n < a_{n+1} < b_{n+1}$ and
$g(x) \ne 0\ \forall\ x \in (a_n, b_n)$, $g(x) = 0\ \forall\ x
\in [b_n, a_{n+1}]$. If $\m$ is in one of the band gaps $(b_n, a_{n+1})$,
then the second integral on the right hand side of (\ref{omfermiforlowt})
vanishes exponentially fast for $T \rightarrow 0$, which is not the case,
if $\m$ is situated inside a band.

\subsection{Specific heat of metals}
Both cases require a separate asymptotic analysis. Let us start with the
metallic case $\m \in (a_n, b_n)$ for some $n \in {\mathbb N}$. We assume
that $g$ is analytic in $(a_n, b_n)$. Then it has a Taylor series expansion
around $\e = \m$ with some finite radius of convergence $\de > 0$,
\begin{equation}
     g(\e) = \sum_{k=0}^\infty \frac{g^{(k)}(\mu)}{k!} (\e - \m)^k \epc
\end{equation}
if $|\e - \m| < \de$. It follows that
\begin{multline} \label{omfermiasy1}
     \Om (T, \m) - V \int_{-\infty}^\m \rd \e \: g(\e) (\e - \m) \\[1ex]
        = - TV \sum_{k=0}^\infty \frac{g^{(k)}(\mu)}{k!}
	       \int_{\m - \de}^{\m + \de} \rd \e \: (\e - \m)^k
	          \ln \bigl(1 + \re^{- \frac{|\e - \m|}{T}}\bigr)
		  + {\cal O} \bigl(T^\infty\bigr) \\[1ex]
        = - V \sum_{k=0}^\infty \frac{g^{(k)}(\mu)}{k!} T^{k+2}
	       \int_{- \de/T}^{\de/T} \rd x \: x^k
	          \ln \bigl(1 + \re^{- |x|}\bigr)
		  + {\cal O} \bigl(T^\infty\bigr) \epp
\end{multline}
Here we have substituted $x = (\e - \m)/T$ in the second equation.
The remaining integral vanishes for symmetry reasons if $k$ is odd.
If $k = 2m$ we obtain
\begin{multline}
     \int_{- \de/T}^{\de/T} \rd x \: x^{2m} \ln \bigl(1 + \re^{- |x|}\bigr)
        = 2 \int_0^\infty \rd x \: x^{2m} \ln \bigl(1 + \re^{-x}\bigr)
	  + {\cal O} \bigl(T^\infty\bigr) \\[1ex]
        = 2 \sum_{n=1}^\infty \frac{(-1)^{n+1}}{n}
	    \int_0^\infty \rd x \: x^{2m} \re^{- nx}
	  + {\cal O} \bigl(T^\infty\bigr)
        = 2 (2m)! \sum_{n=1}^\infty \frac{(-1)^{n+1}}{n^{2m+2}}
	  + {\cal O} \bigl(T^\infty\bigr) \\[1ex]
        = (2m)! (2 - 2^{-2m}) \z (2m + 2) + {\cal O} \bigl(T^\infty\bigr) \epc
\end{multline}
where $\z$ is Riemann's zeta function. Inserting this back
into (\ref{omfermiasy1}) we obtain the low-$T$ asymptotic series
\begin{multline}
     \Om (T, \m) = V \int_{-\infty}^\m \rd \e \: g(\e) (\e - \m) \\[-1ex]
        - V \sum_{k=0}^\infty (2 - 2^{-2k}) \z(2 k + 2) g^{(2k)} (\mu) T^{2k+2}
        + {\cal O} \bigl(T^\infty\bigr)
\end{multline}
for the grand canonical potential.

The corresponding series for the particle number and entropy in
a grand canonical description are
\begin{subequations}
\label{fermilowmetal}
\begin{align}
     N(T, \m) & = - \frac{\6 \Om (T, \m)}{\6 \m}
        = V \int_{-\infty}^\m \rd \e \: g(\e) \notag \\[-.5ex] & \mspace{36.mu}
          + V \sum_{k=0}^\infty (2 - 2^{-2k}) \z(2 k + 2) g^{(2k+1)} (\mu) T^{2k+2}
          + {\cal O} \bigl(T^\infty\bigr) \epc \label{nfermilowmetal} \\[1ex]
     S(T, \m) & = - \frac{\6 \Om (T, \m)}{\6 T} \notag \\ &
        = V \sum_{k=0}^\infty (2 - 2^{-2k}) (2k + 2) \z(2 k + 2) g^{(2k)} (\mu) T^{2k+1}
          + {\cal O} \bigl(T^\infty\bigr) \epp \label{sfermilowmetal}
\end{align}
\end{subequations}

The asymptotic series allow us to calculate the low-temperature
expansion of the specific heat order by order, using
\begin{equation} \label{cvthermo}
     C_V = T \frac{\6 S}{\6 T}\biggr|_N
         = T \biggl(\frac{\6 S}{\6 T}\biggr|_\m + \frac{\6 S}{\6 \m}\biggr|_T
	            \frac{\6 \m}{\6 T}\biggr|_N\biggr) \epp
\end{equation}
Equation (\ref{nfermilowmetal}) can be used iteratively to obtain
$\m$ as a function of $N$ and $T$. The equation
\begin{equation} \label{efnv}
     \frac{N}{V} = \int_{- \infty}^{E_f} \rd \e \: g(\e) \epc
\end{equation}
following from (\ref{nfermilowmetal}) at $T = 0$, determines
$E_F = \m (0, N)$ as a function of the density of particles $N/V$.
Since only $T^2$ enters (\ref{nfermilowmetal}), $\m$ must be even
in $T$,
\begin{equation}
     \m = E_F + \a T^2 + {\cal O} \bigl(T^4\bigr) \epp
\end{equation}
Inserting this into (\ref{nfermilowmetal}) and using (\ref{efnv})
we can calculate $\a$,
\begin{multline} \label{muofnfermi}
     \int_{E_F}^\m \rd \e \: g(\e) + \z(2) g' (\m) T^2
        = (\m - E_F) g(\m) - g' (\m) \frac{(\m - E_F)^2}{2} + \z(2) g' (\m) T^2
	  + {\cal O} \bigl(T^6\bigr) \\[1ex]
        = \bigl(\a g(E_F) + \z(2) g' (E_F)\bigr) T^2
	  + {\cal O} \bigl(T^4\bigr) = {\cal O} \bigl(T^4\bigr) \\[1ex]
        \then\ \a = - \z(2) \frac{g' (E_F)}{g(E_F)}\
        \then\ \m = E_F - \z(2) \frac{g' (E_F)}{g(E_F)} T^2
	          + {\cal O} \bigl(T^4\bigr) \epp
\end{multline}
Furthermore, using (\ref{sfermilowmetal}) we get at once that
\begin{subequations}
\begin{align} \label{sofnfermider}
    T \frac{\6 S}{\6 T}\biggr|_\m & = 2 \z (2) T V g (E_F)
       + {\cal O} \bigl(T^3\bigr) \epc \\
    T \frac{\6 S}{\6 \m}\biggr|_T & = 2 \z (2) T^2 V g' (E_F)
       + {\cal O} \bigl(T^4\bigr) \epp
\end{align}
\end{subequations}
Using (\ref{muofnfermi}) and (\ref{sofnfermider}) in (\ref{cvthermo})
and recalling that $\z(2) = \p^2/6$ we finally arrive at the sought
for low-temperature asymptotics of the specific heat of a metal,
\begin{equation} \label{metallowtcv}
     C_V = \frac{\p^2}{3} TV g (E_F) + {\cal O} \bigl(T^3\bigr) \epp
\end{equation}

Let us add a few comments in conclusion.
\begin{enumerate}
\item
The above low-$T$ asymptotic expansion of the thermodynamic
quantities of a Fermi gas goes back to Sommerfeld \cite{Sommerfeld28}
and is called the Sommerfeld expansion.
\item
Equation (\ref{metallowtcv}) is an important result stating that
the electronic contribution to the specific heat of metals is linear
in $T$ and proportional to the density of states at the Fermi energy.
\item
Taking it the other way round, we see that we can experimentally
determine the density of states of a metal close to its Fermi energy 
by measuring its specific heat.
\item
Since typical electronic energies in solids, like band gaps or band
widths, are of the order of $1\, {\rm eV}$, low-temperature expansions
for the electrons in solids are usually valid even above room temperature.
\end{enumerate}

\subsection{Specific heat of insulators}
Starting once more from equation (\ref{omfermiforlowt}) we consider
the specific heat of an insulator. By definition a band insulator has
$g(E_F) = 0$. The Fermi energy is located inside a band gap, $b_n <
E_F < a_{n+1}$ for some $n \in {\mathbb N}$. In this situation the
closest band below the Fermi energy, the one with index $n$ here, is
called `valence band', the closest band above the Fermi energy,
the one with index $n+1$, `conduction band'. It follows from
(\ref{omfermiforlowt}) that, for $b_n < \m < a_{n+1}$,
\begin{multline} \label{omlowtins}
     \Om (T, \m) - V \int_{-\infty}^\m \rd \e \: g(\e) (\e - \m) \\[1ex]
        \sim - TV \int_{- \infty}^\infty \rd \e \: g(\e) \re^{- \frac{|\e - \m|}{T}}
        \sim - TV \biggl\{\int_{a_n}^{b_n} \rd \e \: g(\e) \re^\frac{\e - \m}{T}
             + \int_{a_{n+1}}^{b_{n+1}} \rd \e \: g(\e) \re^{- \frac{\e - \m}{T}}\biggr\} \\[1ex]
        = - TV \biggl\{\re^\frac{b_n - \m}{T}
	       \int_{a_n}^{b_n} \rd \e \: g(\e) \re^\frac{\e - b_n}{T}
             + \re^{- \frac{a_{n+1} - \m}{T}}
	       \int_{a_{n+1}}^{b_{n+1}} \rd \e \: g(\e) \re^{- \frac{\e - a_{n+1}}{T}}\biggr\} \\[1ex]
        = - T^2 V \biggl\{\re^\frac{b_n - \m}{T} \mspace{-36.mu}
	       \int\displaylimits_{- (b_n - a_n)/T}^{0} \mspace{-36.mu} \rd x \: g(b_n + Tx) \re^x
             + \re^{- \frac{a_{n+1} - \m}{T}} \mspace{-45.mu}
	       \int\displaylimits_0^{(b_{n+1} - a_{n+1})/T} \mspace{-45.mu} \rd x \: g(a_{n+1} + Tx)
	       \re^{- x} \biggr\} \epp
\end{multline}
Here we have neglected contributions that are exponentially smaller
than the displayed terms. In order to further simplify the remaining
integrals we have to recall that for 3d systems there are square-root
type van-Hove singularities at the band edges. This means that there
are $\a_v, \a_c > 0$ such that within the respective bands
\begin{equation}
     g(b_n + Tx) \sim \a_v \sqrt{- Tx} (1 + {\cal O} (Tx)) \epc \qd
     g(a_{n+1} + Tx) \sim \a_c \sqrt{Tx} (1 + {\cal O} (Tx)) \epp
\end{equation}
Substituting this into the integrals on the right hand side of
(\ref{omlowtins}) we obtain, for instance,
\begin{equation}
     \int\displaylimits_{- (b_n - a_n)/T}^{0} \mspace{-36.mu} \rd x \: g(b_n + Tx) \re^x
        \sim \sqrt{T} \a_v \int_0^\infty \rd x \: \sqrt{x} \re^{-x}
	= \sqrt{T} \a_v \G(3/2) = \2 \sqrt{\p T} \a_c
\end{equation}
and a similar expression for the other integral. Altogether we end
up with
\begin{equation}
     \Om (T, \m) \sim  V \int_{-\infty}^\m \rd \e \: g(\e) (\e - \m)
        - \frac{\sqrt{\p}}{2} T^\frac{5}{2} V
	  \biggl\{\a_v \re^{- \frac{\m - b_n}{T}}
	          + \a_c \re^{- \frac{a_{n+1} - \m}{T}}\biggr\} \epc
\end{equation}
which is the low-temperature asymptotics of the grand canonical
potential for an insulator, when $\m \in (b_n, a_{n+1})$.

Again the formulae for particle number and entropy in the grand
canonical ensemble are obtained by taking derivatives (cf.\
(\ref{fermilowmetal})),
\begin{subequations}
\begin{align} \label{nfermilowins}
     N(T,\m) & \sim V \int_{-\infty}^\m \rd \e \: g(\e)
        - \frac{\sqrt{\p}}{2} T^\frac{3}{2} V
	  \biggl\{\a_v \re^{- \frac{\m - b_n}{T}}
	          - \a_c \re^{- \frac{a_{n+1} - \m}{T}}\biggr\} \epc \\[1ex]
     S(T,\m) & \sim \frac{\sqrt{\p}}{2} T^\frac{1}{2} V
	  \biggl\{(\m - b_n) \a_v \re^{- \frac{\m - b_n}{T}}
	          + (a_{n+1} - \m) \a_c \re^{- \frac{a_{n+1} - \m}{T}}\biggr\} \epp
		  \label{sfermilowins}
\end{align}
\end{subequations}

For $T \rightarrow 0+$ in (\ref{nfermilowins}) we still get that
the integral on the right hand side with upper limit $E_F$ is equal
to the total particle number, but now this equality does not fix
$E_F$, since the integral as a function of $\m$ is constant for
$\m \in [b_n, a_{n+1}]$, thus non-invertible. Since $\m$ is continuous
in a vicinity of $T = 0$, we see that the integral exactly equals
$N$ even for small finite temperatures. Hence, the second term
on the right hand side must vanish asymptotically,
\begin{equation}
     \frac{\sqrt{\p}}{2} T^\frac{3}{2} V
        \biggl\{\a_v \re^{- \frac{\m - b_n}{T}}
	          - \a_c \re^{- \frac{a_{n+1} - \m}{T}}\biggr\} \sim 0 \epc
\end{equation}
which determines $\m$ at small $T$ to be
\begin{equation} \label{mufermilowtins}
     \m = \frac{b_n + a_{n+1}}{2} + \frac{T}{2} \ln \Bigl(\frac{\a_v}{\a_c}\Bigr)
          + {\cal O} \bigl(T^\infty\bigr) \epp
\end{equation}
In particular,
\begin{equation} \label{efins}
     E_F = \frac{b_n + a_{n+1}}{2} \epp
\end{equation}
The Fermi energy of an insulator is in the middle of the band gap.

By definition
\begin{equation} \label{defgap}
     \D = a_{n+1} - b_n
\end{equation}
is the `width of the band gap' or simply `the band gap'. Using
(\ref{cvthermo}), (\ref{sfermilowins}) and (\ref{mufermilowtins})-%
(\ref{defgap}) one straightforwardly obtains the expression
\begin{equation} \label{lowtcvins}
     C_V = \frac{\a_v + \a_c}{2} \sqrt{\frac{\p}{T}} V
           \Bigl(\frac{\D}{2}\Bigr)^2 \re^{- \frac{\D}{2T}}
	   \bigl(1 + {\cal O} (T)\bigr)
\end{equation}
for the low-temperature asymptotic behaviour of the specific heat
of a band insulator.

Here again a few concluding remarks are in order.
\begin{enumerate}
\item
The functional form of the low-$T$ specific heat in (\ref{lowtcvins})
is called `thermally activated behaviour'. One says that the
activation barrier equals half of the band gap.
\item
The size of the electronic contribution to the specific heat of an
insulator depends in an extremal way on the band gap. Example:
\begin{subequations}
\begin{align}
     \D & = 0,5\, {\rm eV}\ \then\ \re^{- \frac{\D}{2T}} \approx 10^{-4} \epc \\[1ex]
     \D & = 3,0\, {\rm eV}\ \then\ \re^{- \frac{\D}{2T}} \approx 10^{-24}
\end{align}
\end{subequations}
at $T \approx 300 {\rm K}$. This makes the difference, as far as
the electronic specific heat is concerned, between insulators and
`semi-conductors'.
\item
For an insulator in the low-temperature regime the electrons contribute
almost nothing to the specific heat. The specific heat of insulators
is mostly determined by the phonons. In conductors, on the other hand,
the electrons do contribute to the specific heat and even dominate it
at low enough temperatures. For this reason metals have, in general, a
larger heat capacity than insulators.
\end{enumerate}

\mysection{Electrons in solids -- the second quantized picture}
We recall that the Hamiltonian (\ref{hell}) of the electrons in solids
in adiabatic approximation is
\begin{equation}
     H_{\rm el} = \sum_{j=1}^N \Bigl\{ \2 \|\pv_j\|^2 + V_{\rm I} (\xv_j) \Bigr\}
                  + \sum_{1 \le j < k \le N} V_{\rm C} (\xv_j - \xv_k) \epc
\end{equation}
where $V_{\rm I} (\xv)$ is the periodic potential of the ions in the
crystal, $V_{\rm C} (\xv) = 1/\|\xv\|$ is the Coulomb potential, and
$N$ is the number of electrons we are taking into account.

\subsection{An auxiliary potential}
Our aim in the following lectures is to proceed beyond the single-particle
approximation. To begin with, we remark that we can introduce an
auxiliary potential $V_{\rm A} (\xv)$ without changing too much the
structure of the Hamiltonian. Let
\begin{equation}
     V(\xv) = V_{\rm I} (\xv) + V_{\rm A} (\xv) \epp
\end{equation}
Then
\begin{multline}
     H_{\rm el} - \sum_{j=1}^N \Bigl\{ \2 \|\pv_j\|^2 + V (\xv_j) \Bigr\}
        = \sum_{1 \le j < k \le N} V_{\rm C} (\xv_j - \xv_k)
	  - \sum_{j=1}^N V_{\rm A} (x_j) \\
	= \sum_{1 \le j < k \le N}
	  \biggl\{\underbrace{V_{\rm C} (\xv_j - \xv_k)
	  - \frac{V_{\rm A} (\xv_j) + V_{\rm A} (\xv_k)}{N-1}}_{= U(\xv_j, \xv_k)}\biggr\}
	  \\[1ex] \Leftrightarrow\
     H_{\rm el} = \sum_{j=1}^N \Bigl\{ \2 \|\pv_j\|^2 + V (\xv_j) \Bigr\}
	            + \sum_{1 \le j < k \le N} U(\xv_j, \xv_k) \epp
\end{multline}
Note that the choice of the auxiliary potential is entirely at
our disposal. We may, for instance, choose the mean field potential
defined in section~\ref{sec:reductionone}. If we neglect $U$ we
are applying a single-particle approximation. Good single-particle
approximations are obtained for appropriate choices of $V_{\rm A}$.
A single-particle approximation is good, if the two-particle matrix
elements of $U(\xv, \yv)$ calculated with eigenstates of the
single-particle Hamiltonian
\begin{equation}
     h(\xv, \pv) = \frac{\|\pv\|^2}{2} + V (\xv)
\end{equation}
are small for single-particle energies close to the Fermi surface.

\subsection{Bloch basis and Wannier basis}
In the following calculation we will not fix the auxiliary potential
$V_{\rm A}$. Our only explicit assumption is that it is periodic.
Implicitly we shall also assume that it allows us to take screening
into account. Due to the periodicity the eigenfunctions of $h$ are
`Bloch functions' $\ph_{\a \kv}$ labeled by a band index $\a \in
{\mathbb N} $ and a lattice momentum $\kv \in BZ$,
\begin{equation}
     h \ph_{\a \kv} (\xv) = \e_\a (\kv) \ph_{\a \kv} (\xv) \epp
\end{equation}
We shall call the $\e_\a (\kv)$ the single-particle energies. The
Bloch theorem implies
\begin{equation}
     \ph_{\a \kv} (\xv) = \re^{\i \<\kv, \xv\>} u_{\a \kv} (\xv)
\end{equation}
with a lattice periodic function $u_{\a \kv}$. The set
$\{\ph_{\a \kv}\}_{\a \in {\mathbb N}, \kv \in BZ}$ is a single-particle
orthonormal basis of the electronic Hilbert space called the
`Bloch basis'.

Define
\begin{equation}
     \Ph_\a (\xv) = \frac{1}{\sqrt{L}} \sum_{\kv \in BZ} \ph_{\a \kv} (\xv) \epc
\end{equation}
where $L$ is the number of unit cells.\footnote{Note the change of
notation!} Then $\{\Ph_\a (\xv - \Rv_j)|\a \in {\mathbb N}, \Rv_j \in B\}$
is another orthonormal basis (prove it!) called the Wannier basis.
$\Ph_\a (\xv)$ is called a Wannier function. The Wannier functions
generalize the atomic orbitals in section~\ref{sec:tbmethod}.

Bloch basis and Wannier basis are connected via Fourier transformation,
\begin{multline}
     \frac{1}{\sqrt{L}} \sum_{j=1}^L \re^{\i \<\kv, \Rv_j\>} \Ph_\a (\xv - \Rv_j)
        = \frac{1}{L} \sum_{j=1}^L \sum_{\pv \in BZ} \re^{\i \<\kv, \Rv_j\>}
	  \re^{\i \<\pv, \xv - \Rv_j\>} u_{\a \pv} (\xv - \Rv_j) \\
        = \sum_{\pv \in BZ} \re^{\i \<\pv, \xv\>} u_{\a \pv} (\xv)
          \frac{1}{L} \sum_{j=1}^L \re^{\i \<\kv - \pv, \Rv_j\>}
        = \ph_{\a \kv} (\xv) \epp
\end{multline}
Here we have used that the second sum on the right hand side of
the second equation equals $L \de_{\kv, \pv}$.

\subsection{Hamiltonian in second quantization}
Let $c_{\a \kv, a}^+$ the creation operator of a Bloch electron
of spin $a \in \{\auf, \ab\}$. The operators
\begin{equation}
     c_{\a j, a}^+ = \frac{1}{\sqrt{L}} \sum_{\kv \in BZ} \re^{- \i\<\kv,\Rv_j\>} c_{\a \kv, a}^+
\end{equation}
are then an alternative set of creation operators, creating electrons
in Wannier orbitals. This can be seen by writing down the corresponding
field operators,
\begin{multline}
     \Ps_a^+ (\xv) = \sum_{\a \kv} \ph_{\a \kv}^* (\xv) c_{\a \kv, a}^+
        = \sum_{\a \kv} \ph_{\a \kv}^* (\xv)
	     \frac{1}{\sqrt{L}} \sum_{j=1}^L \re^{\i \<\kv, \Rv_j\>} c_{\a j, a}^+ \\
        = \sum_{j=1}^L \biggl(\frac{1}{\sqrt{L}}
	   \sum_{\a \kv} \ph_{\a \kv}^* (\xv) \re^{\i \<\kv, \Rv_j\>}
	   \biggr)  c_{\a j, a}^+
        = \sum_\a \sum_{j=1}^L \Ph_\a^* (\xv - \Rv_j) c_{\a j, a}^+
\end{multline}
According to the general prescription (see lecture on QM) the Hamiltonian
in `occupation number representation' (`second quantization') can be written as
\begin{multline} \label{scdhell}
     H_{\rm el} = \int \rd^3 x \: \Ps_a^+ (\xv) h \Ps_a (\xv)
        + \2 \int \rd^3 x \int \rd^3 y \: \Ps_a^+ (\xv) \Ps_b^+ (\yv) U(\xv, \yv)
	                                  \Ps_b (\yv) \Ps_a (\xv) \\
        = \sum_{\a, \be; i, j}
	     \underbrace{\int \rd^3 x \: \Ph_\a^* (\xv - \Rv_i)
	                 h \Ph_\be (\xv - \Rv_j)}_{= \de_{\a \be} t_{ij}^\a}
			 c_{\a i, a}^+ c_{\be j, a}^{} \\
          + \2 \sum_{\substack{\a, \be, \g, \de\\i, j, k, \ell}}
	    \underbrace{\int \rd^3 x \int \rd^3 y \:
	                       \Ph_\a^* (\xv - \Rv_i)
			       \Ph_\be^* (\yv - \Rv_j) U(\xv, \yv)
	                       \Ph_\g (\yv - \Rv_k)
			       \Ph_\de (\xv - \Rv_\ell)}_{= U^{\a \be \g \de}_{i j k \ell}}
			       \times \\[-4ex] \mspace{360.mu} \times
            c_{\a i, a}^+ c_{\be j, b}^+ c_{\g k, b}^{} c_{\de \ell, a}^{} \\[3ex]
        = \sum_{\a; i, j}  t_{ij}^\a c_{\a i, a}^+ c_{\a j, a}^{} +
          \2 \sum_{\substack{\a, \be, \g, \de\\i, j, k, \ell}}
	        U^{\a \be \g \de}_{i j k \ell}
                c_{\a i, a}^+ c_{\be j, b}^+ c_{\g k, b}^{} c_{\de \ell, a}^{} \epp
		\mspace{54.mu}
\end{multline}
Here implicit summation over the spin indices $a, b$ is implied. The
$t_{ij}^\a$ are called `transition matrix elements' or `hopping matrix
elements', the $U^{\a \be \g \de}_{i j k \ell}$ are called `interaction
parameters'.

Note that:
\begin{enumerate}
\item
So far the Hamiltonian is only rewritten in `second quantization'.
No other approximation than the adiabatic approximation has been
applied.
\item
In this form it is the starting point of the `theory of strongly
correlated electron systems'.
\item
An optimal choice of the Wannier functions (through an optimal
choice of $V_{\rm A}$) minimizes the strength and the range of
the interaction parameters.
\item
Suppose the auxiliary potential $V_{\rm A}$ can be chosen in such
a way that the interaction parameters are always small. Then they
can be neglected and we are in the realm of band theory.
\item
If the Wannier functions can be constructed in such a way that
they resemble atomic wave function in the sense that they
are localized around the origin and decay sufficiently fast
away from it, then the interaction parameters become short-range,
and it may be justified to consider only on-site or near-neighbour
contributions.
\end{enumerate}
\subsection{Exercise 16. Wannier functions in one dimension}
The dispersion relation of a free non-relativistic particle of mass $m$ is 
$\varepsilon(k) = \hbar^2 k^2 / 2m$. Consider a lattice of $N$ sites and
physical length $L$ under periodic boundary conditions. Then only $N$ discrete
wave vectors of the form $k = 2\pi n/L$, $n\in\mathbb{Z}$, are possible inside
the Brillouin zone. All wave vectors outside will be folded back to the
Brillouin zone, i.e.\ new bands with dispersion $\varepsilon(k + 2\pi m / a)$
develop, where $a=L/N$ is the lattice constant and $m \in \mathbb{Z}$ is the
band index. The corresponding Bloch functions are labeled in the following way,
\begin{equation} \label{bloch}
     \varphi_{m,k}(x) = \frac{1}{\sqrt{L}} \exp\left[\i(k+m\frac{2\pi}{a})x\right] \epp 
\end{equation}
Construct the complementary Wannier basis $\{\phi_m(x-R_i)\}$, $i=1, \dots, N$, in
the thermodynamic limit $(N,L \to \infty, L/N = a = \mathrm{const}.)$. This requires
to calculate 
\begin{equation}
      \phi_m(x) = \frac{1}{\sqrt{N}} \sum_{k \in \mathrm{BZ}} \varphi_{m,k}(x) \epc
\end{equation}
where, in the thermodynamic limit, the summation over the first Brillouin zone can
be replaced by an integration between the zone boundaries $\pm \pi/a$.

\subsection{Exercise 17. Spinless Fermions on the lattice}
Consider creation and annihilation operators $c_m^+$, $c_n^{}$, $m, n = 1, \dots, L$,
of spinless Fermions. They satisfy canonical anti-commutation relations 
\begin{equation} \label{anti}
     \{c_m^{}, c_n^{}\} = \{c_m^+, c_n^+\} = 0 \epc \qd \{c_m^{}, c_n^+\} = \delta_{m, n} \epp
\end{equation}
Let $|0\rangle$ denote the Fock vacuum defined by $c_m |0\rangle = 0$,
$m = 1, \dots, L$. Non-interacting Fermions are described by the Hamiltonian 
\begin{equation} \label{ham}
     H = \sum_{m, n = 1}^L t^m_n c_m^+ c_n^{} \epp
\end{equation}
Here the matrix $t$ with matrix elements $t^m_n$ is called the transition matrix.
It is Hermitian, ${t^m_n}^* = t^n_m$, by definition.
\begin{enumerate}
\item 
Show that the canonical anti-commutation relations (\ref{anti}) are invariant
under transformations of the form
\begin{equation}
     \tilde c_k = \sum_{m=1}^L U^k_m c_m^{} \epc
\end{equation}
if the matrix $U$ with matrix elements $U^k_m$ is unitary.
\item
Show that a unitary transformation $U$ exists which transforms the Hamiltonian
(\ref{ham}) to the form
\begin{equation}
     H = \sum_{m=1}^L \varepsilon_m \tilde c_m^+ \tilde c_m^{} \epp
\end{equation}
Explain why this transformation solves the eigenvalue problem $H|\psi\rangle =
E |\psi \rangle $.
\item
We say that the transformation in (ii) diagonalizes $H$. Diagonalize the
Hamiltonian with transition matrix elements
\begin{equation}
     t^m_n = - t_0 (\delta_{m, n+1} + \delta_{m, n-1})
\end{equation}
explicitly. Here the indices of the right hand side should be understood
modulo $L$. 
\item
Also diagonalize the Hamiltonian with transition matrix elements
\begin{equation}
     t^m_n = \begin{cases} 0 & \text{if $m = n$} \\
                \frac{\pi \i}{L} \, \sin^{-1} \left( \frac{(m - n)\pi}{L}
		\right) & \text{else.}
             \end{cases}
\end{equation}
Hint: Use the canonical gauge transformation $c_n \mapsto e^{\i\pi n/2L}c_n$ as
well as the relation $1/(e^{-\i\pi(m-n)/L}-1) = \frac{1}{L}\sum_{k=0}^{L-1}
k e^{-\i\pi(m-n)k/L}$. The latter can be proved by differentiating
a geometric sum with respect to an appropriate parameter.
\end{enumerate}

\mysection{The Hubbard model}
\subsection{Motivation and definition}
The formalism of second quantization, in which the Hamiltonian of
the electrons takes the form (\ref{scdhell}), gives us a more
intuitive access to the problem. If the Fermi surface is located
within a single conduction band with band index $\a$, say, then
the interaction between different bands can be neglected for small
excitation energies, and we can suppress the band indices (Greek
indices) in (\ref{scdhell}). If, moreover, the `intra-atomic
Coulomb interaction' $U_{iiii}$ is dominant, then a single
effective interaction parameter $U$, say, remains, and $H_{\rm el}$
can be approximated by
\begin{equation} \label{hhubgen}
     H = \sum_{i, j} t_{ij} c_{i a}^+ c_{j a}^{}
         + \frac{U}{2} \sum_i c_{i a}^+ c_{i b}^+ c_{i b}^{} c_{i a}^{}
\end{equation}
which defines the so-called (one-band) Hubbard model \cite{Hubbard63,%
Gutzwiller63}.

\begin{enumerate}
\item
The Hubbard model is a `minimal extension' of the band model in
the sense that as few as possible of the interaction parameters
of the full Hamiltonian (\ref{scdhell}) are taken into account.
\item
In spite of its apparent simplicity the Hubbard model is a true
many-body model. In general it is very hard to deal with by
any of the means of modern theoretical physics as e.g.\ perturbation
theory, renormalization group analysis, quantum Monte-Carlo methods.
\item
As simple it is intuitively, the Hubbard model is the starting
point for extensions. Basically all models of `strongly
correlated electrons' (everything beyond the band model) are
extended Hubbard models which are either obtained by taking
interaction parameters of wider range into account (e.g.\ nearest
neighbours, next-to-nearest neighbours) or by taking into
account a larger number of bands (e.g.\ two-band Hubbard model,
three-band Hubbard model).
\item
As it stands the Hubbard model is believed to give a realistic
description of
\begin{itemize}
\item
the electronic properties of solids with tight bands
\item
band magnetism (iron, copper, nickel)
\item
the interaction-induced metal-insulator transition (Mott transition)
\end{itemize}
\item
If the Hubbard Hamiltonian is supplied with periodic boundary
conditions, the number of lattice sites (which is equal to the
number of Wannier orbitals) is finite. As we are are also dealing
with a finite number of states per site, the model has a
finite-dimensional space of states, and can be thought of as
a `fully regularized' quantum field theory. Its Hamiltonian can
be represented by a finite Hermitian matrix. This makes the
Hubbard model attractive for computer based approaches.
\item
The 1d Hubbard model has the amazing feature of being integrable
\cite{LiWu68,Shastry86b}. Rather much is known about its elementary
excitations and its thermodynamics \cite{thebook}.
\end{enumerate}
\subsection{Tight-binding approximation}
The assumption that the Wannier functions are strongly localized
in the vicinity of the `lattice sites' $\Rv_j$ is compatible with
the restriction of the hopping amplitudes $t_{ij}$ to nearest
neighbours $\<ij\>$ on the lattice, which is called the `tight-binding
approximation'. If we introduce the `density operators' (local
particle-number operators)
\begin{equation}
     n_{i \auf} = c_{i \auf}^+ c_{i \auf}^{} \epc \qd
     n_{i \ab} = c_{i \ab}^+ c_{i \ab}^{}
\end{equation}
and apply the tight-binding approximation to (\ref{hhubgen}) we obtain
\begin{equation} \label{hhub}
     H = - t \sum_{\<ij\>} c_{i a}^+ c_{j a}^{} + U \sum_i n_{i \auf} n_{i \ab} \epp
\end{equation}
Here we have assumed isotropic nearest-neighbour hopping of strength
$- t$ and the vanishing of the on-site energies $t_{ii}$, which for
a homogeneous model can always be assumed, since in this case it is
equivalent to a redefinition of the chemical potential. For the
interaction part we have calculated
\begin{equation}
     c_{i a}^+ c_{i b}^+ c_{i b}^{} c_{i a}^{} =
        c_{i \auf}^+ c_{i \ab}^+ c_{i \ab}^{} c_{i \auf}^{} +
        c_{i \ab}^+ c_{i \auf}^+ c_{i \auf}^{} c_{i \ab}^{}
	= 2 n_{i \auf} n_{i \ab} \epp
\end{equation}

With the Hamiltonians (\ref{hhubgen}) and (\ref{hhub}) the terminology
is not so sharp.  The Hamiltonian (\ref{hhub}) is also called `the Hubbard
Hamiltonian'. And, in fact, the term Hubbard model most commonly refers
to the model described by (\ref{hhub}). 

\subsection{Interpretation of the Hubbard Hamiltonian}
Let us have a closer look at the Hubbard Hamiltonian. For simplicity
we consider the one-dimensional model,
\begin{equation} \label{htu}
     H = - t \sum_{j=1}^L \bigl(c_{j a}^+ c_{j+1 a}^{} + c_{j+1 a}^+ c_{j a}^{}\bigr)
         + U \sum_{j=1}^L n_{j \auf} n_{j \ab} \epc
\end{equation}
subject to periodic boundary conditions, $c_{L+1 a} = c_{1 a}$.
Its space of states will be denoted ${\cal H}^{(L)}$. It is generated
by filling electrons into Wannier states. According to the Pauli
principle, every Wannier state may be unoccupied, occupied with
one electron of spin $\auf$ or $\ab$, or with two electrons with
opposite spins, giving altogether $\dim {\cal H}^{(L)} = 4^L$ states.

Let us construct the basis of Wannier states explicitly. For this
purpose define the row vectors $\xv = (x_1, \dots, x_N)$, $\av =
(a_1, \dots, a_N)$, where $x_j \in \{1, \dots, L\}$, $a_k \in
\{\auf, \ab\}$ for $j, k = 1, \dots N$; $N \in \{1, \dots, 2L\}$.
The state
\begin{equation}
     |\xv; \av\> = c_{x_N, a_N}^+ \dots c_{x_1, a_1}^+ |0\>
\end{equation}
is a Wannier state representing $N$ electrons at sites $x_j$ with
spins $a_j$. The set of all different states of this form,
\begin{equation}
     B_W = \biggl\{|\xv; \av\> \in {\cal H}^{(L)}\bigg|
		   \begin{array}{l}
                   N = 0, \dots, 2L \\
		   x_{j+1} \ge x_j, a_{j+1} > a_j \text{if}\ x_{j+1} = x_j
		   \end{array}
		   \biggr\} \epc
\end{equation}
is a basis of Wannier states, since all such states are linear
independent and since their number is
\begin{equation}
     \sum_{N = 0}^{2L} \binom{2L}{N} = 4^L \epp
\end{equation}

The operators $n_{j, \auf}$, $n_{j, \ab}$ are the local particle number
operators for electrons of spin $\auf$ and $\ab$ at site~$j$. Let us recall
why this name is justified. Using the canonical anti-commutation relations
of the Fermi operators and the fact that the $c_{j, a}$ annihilate the
Fock vacuum $|0\>$ we conclude that 
\begin{equation} \label{nccom}
     [n_{j, \auf}, c_{k, b}^\+] = \de_{jk} \de_{\auf b} c_{k, b}^\+ \epc \qd
	n_{j, \auf} |0\> = 0 \epc
\end{equation}
and therefore
\begin{equation} \label{nact}
   n_{j, \auf} |\xv, \av\> = \sum_{k=1}^N \de_{j, x_k} \de_{\auf, a_k}
			  |\xv, \av\>
\end{equation}
and similarly for $n_{j, \ab}$. Thus, $n_{j, a}|\xv, \av\> = |\xv, \av \>$,
if site $j$ is occupied by an electron of spin $a$, and zero elsewise.

A first interpretation of the Hubbard model can be obtained by
considering separately the two contributions that make up the
Hamiltonian (\ref{htu}). For $t = 0$ or $U = 0$ it can be diagonalized
and understood by elementary means. For $t = 0$ the Hamiltonian
reduces to $H = UD$, where
\begin{equation} \label{do}
     D = \sum_{j=1}^L n_{j \auf} n_{j \ab} \epp
\end{equation}
Using (\ref{nact}) we can calculate the action of $D$ on a state
$|\xv, \av\>$,
\begin{equation} \label{doev}
\begin{split}
     D |\xv, \av\> & = \sum_{k,l=1}^N \de_{x_k, x_l} \de_{\auf, a_k}
                                      \de_{\ab, a_l} |\xv, \av\> \\
                   & = \sum_{1 \le k < l \le N}
			  \de_{x_k, x_l}
			  (\de_{\auf, a_k} \de_{\ab, a_l}
			     + \de_{\ab, a_k} \de_{\auf, a_l})
			  |\xv, \av\> \\
                   & = \sum_{1 \le k < l \le N}
			  \de_{x_k, x_l}
			  (\de_{\auf, a_k} + \de_{\ab, a_k})
			  (\de_{\auf, a_l} + \de_{\ab, a_l})
			  |\xv, \av\> \\
                   & = \sum_{1 \le k < l \le N}
			  \de_{x_k, x_l} |\xv, \av\> \epp
\end{split}
\end{equation}
Here we used $\de_{\auf, a_k} \de_{\ab, a_k} = 0$ in the second equation
and the Pauli principle in the third equation. As we learn from
(\ref{doev}) every state $|\xv, \av\>$ is an eigenstate of the operator
$D$. Thus, $D$ is diagonal in the Wannier basis. The limit
$t \rightarrow 0$ of the Hubbard Hamiltonian (\ref{htu}) is called
the atomic limit, \index{atomic limit} because the eigenstate
$|\xv, \av\>$ describes electrons localized at the sites $x_1, \dots,
x_N$, which are identified with the loci of the atomic orbitals the
electrons may occupy.

The meaning of the operator $D$ is evident from equation (\ref{doev}).
$D$ counts the number of double-occupied sites in the state
$|\xv, \av\>$. The contribution of the term $UD$ to the energy is
non-negative for positive $U$ and increases with the number of double-%
occupied sites. This can be viewed as on-site repulsion among the
electrons. Negative $U$ on the other hand, means on-site attraction.
Hence, it is natural to refer to $D$ as to the operator of the on-site
interaction.

In the other extreme, when $U = 0$, the Hamiltonian (\ref{htu}) turns
into
\begin{equation} \label{htb}
     H_0 = - t \sum_{j=1}^L
	   (c_{j, a}^\+ c^{}_{j+1, a} + c_{j+1, a}^\+ c^{}_{j, a}) \epp
\end{equation}
This is called the tight-binding Hamiltonian. \index{tight-binding
Hamiltonian} Like every translation invariant one-body Hamiltonian
it can be diagonalized by discrete Fourier transformation. Let us
define
\begin{equation} \label{cfourier}
     \tilde c_{k, a}^\+ = \frac{1}{\sqrt{L}}
		         \sum_{j=1}^L \re^{\i \Ph kj} \, c_{j, a}^\+
		         \epc \qd k = 0, \dots, L - 1,
\end{equation}
where $\Ph = 2\p /L$. Then, by Fourier inversion
\begin{equation} \label{cfourierinv}
     c_{j, a}^\+ = \frac{1}{\sqrt{L}}
                   \sum_{k=0}^{L-1} \re^{- \i \Ph jk} \,
		   \tilde c_{k, a}^\+ \epc \qd j = 1, \dots, L \epp
\end{equation}
Equation (\ref{cfourierinv}) is readily verified by inserting
(\ref{cfourier}) into the right hand side and using the geometric sum
formula. Clearly, $\tilde c_{k + L, a}^\+ = \tilde c_{k, a}^\+$.
Insertion of (\ref{cfourierinv}) into (\ref{htb}) leads to
\begin{equation}
     H_0 = - 2t \sum_{k=0}^{L-1} \sum_{a = \auf, \ab}
             \cos (\Ph k) \tilde n_{k, a} \epc
\end{equation}
where $\tilde n_{k, a} = \tilde c_{k, a}^\+ \tilde c^{}_{k, a}$.

The Fourier transformation leaves the canonical anti-commutation
relations invariant,
\begin{subequations} \label{comabt}
\begin{align} \label{comat}
     & \{\tilde c^{}_{j, a}, \tilde c^{}_{k, b}\} =
       \{\tilde c_{j, a}^\+, \tilde c_{k, b}^\+\} = 0 \epc \\
     & \{\tilde c^{}_{j, a}, \tilde c_{k, b}^\+\} = \de_{jk} \de_{ab}
       \epp \label{combt}
\end{align}
\end{subequations}
A transformation with this property is called canonical. Applying
(\ref{cfourier}) to the empty lattice state $|0\>$ (the Fock vacuum),
we obtain
\begin{equation} \label{vacuumt}
     \tilde c_{k, a} |0\> = 0 \epc \qd k = 0, \dots, L - 1 \epc
                                   \qd a = \auf, \ab \epp
\end{equation}
Thus, acting with the creation operators $\tilde c_{k, a}^\+$ on the empty
lattice $|0\>$ we obtain an alternative basis $B_B$. Let us introduce
the row vectors $\qv = (q_1, \dots, q_N) = (k_1, \dots, k_N) \Ph$
and the states
\begin{equation} \label{basisstatet}
     |\qv, \av \> = \tilde c_{k_N, a_N}^\+ \dots
                    \tilde c_{k_1, a_1}^\+ |0\> \epp
\end{equation}
It can be shown that these states are eigenstates of a lattice
momentum operator with eigenvalue $\bigl(\sum_{j=1}^N q_j
\bigr) \mod 2\p$. The set
\begin{equation} \label{blochbasis}
     B_B = \left\{ |\qv, \av \> \in {\cal H}^{(L)} \left|
                       \begin{array}{l}
		       N = 0, \dots, 2L \\
		       q_{j+1} \ge q_j,\ a_{j+1} > a_j\ \text{if}\
		       q_{j+1} = q_j \end{array} \right. \right\}
\end{equation}
is a basis of ${\cal H}^{(L)}$. This basis is sometimes called the
Bloch basis. Electrons in Bloch states $|\qv, \av \>$ are delocalized,
\index{Bloch basis} but have definite momenta $q_1, \dots, q_N$.

By virtue of (\ref{comabt}), the analogues of (\ref{nccom}) and
(\ref{nact}) are satisfied by $\tilde n_{j, a}$ and
$\tilde c_{k, b}^\+$. It follows that
\begin{equation} \label{htbbloch}
     H_0 |\qv, \av\> = - 2t \sum_{j=1}^N \cos (q_j) |\qv, \av \>
                       \epp
\end{equation}
Thus, the tight-binding Hamiltonian $H_0$ is diagonal in the Bloch
basis. It describes non-interacting band electrons in a cosine-shaped
band of width~$4t$.

The tight-binding Hamiltonian $H_0$ and the operator $D$ which counts
the number of double-occupied sites do not commute. Therefore the
Hubbard Hamiltonian can neither be diagonal in the Bloch basis nor in
the Wannier basis. The physics of the Hubbard model may be understood
as arising from the competition between the two contributions, $H_0$
and $D$, to the Hamiltonian (\ref{htu}). The tight-binding contribution
$H_0$ prefers to delocalize the electrons, while the on-site interaction
$D$ favours localization. The ratio
\begin{equation} \label{smallu}
     u = \frac{U}{4t}
\end{equation}
is a measure for the relative contribution of both terms and is the
intrinsic, dimensionless coupling constant of the Hubbard model.

\subsection{Exercise 18. Peierls phases}
\label{sec:peierlsph}
If interacting electrons (charge $-e$, mass $m$) in a periodic
potential $V(\xv)$ are exposed to an external electro-magnetic
field, their one-particle Hamiltonian becomes
\begin{equation}
     h = \frac{1}{2m} \Bigl\| \pv + \frac{e}{c} \Av (\xv, t) \Bigr\|^2
                      + V(\xv) - e \Ph(\xv,t) \epc
\end{equation}
where $\Av (\xv,t)$ and $\Ph(\xv,t)$ are the vector and scalar
potentials of the external field. In the lecture we considered
the many-body Hamiltonian generated by $h$ relative to the Wannier
basis. It was parameterized by hopping matrix elements $t_{jk}$. Here
we would like to find out how the external field modifies the $t_{jk}$.

We start by fixing the gauge such that $\Ph(\xv,t) = 0$.
\begin{enumerate}
\item
Let $\la (\xv,t)$ an arbitrary differentiable function. Using that
$p^k = - \i \6^k$ verify the commutator relation
\begin{equation}
     \bigl[p^k, \re^{- \i e \la(\xv,t)/c}\bigl] =
        - \frac{e}{c} \frac{\6 \la(\xv,t)}{\6 x_k} \re^{- \i e \la(\xv,t)/c} \epp
\end{equation}
\item
Denote by $\Ph (\xv - \Rv_j)$ the Wannier orbital at site $\Rv_j$ and
recall that the hopping matrix elements for vanishing external fields
were
\begin{equation}
     t_{ij} = \int \rd^3 x \: \Ph^* (\xv - \Rv_i)
                           \biggl[\frac{\|\pv\|^2}{2m} + V(\xv)\biggr] \Ph(\xv - \Rv_j) \epp
\end{equation}
Show that in the presence of the external field this has to be
modified to become
\begin{multline}
     t_{ij} = \int \rd^3 x \: \Ph^* (\xv - \Rv_i) \re^{- \i e \la/c}
              \biggl[\frac{1}{2m}
	             \biggl\{p^k + \frac{e}{c} \biggl(A^k - \frac{\6 \la}{\6 x_k}\bigg)\biggr\}^2
		     + V(\xv)\biggr] \re^{\i e \la/c} \\
		     \times \Ph(\xv - \Rv_j) \epc
\end{multline}
where $\la$ is still arbitrary.
\item
Choosing now
\begin{equation}
     \la(\xv,t) = \int_{\xv_0}^\xv \rd y_k \Av^k (\yv,t)
\end{equation}
for an arbitrary fixed point $\xv_0$ and redefining
$\widetilde \Ph (\xv - \Rv_j) = \re^{\i e \la(\xv,t)/c} \Ph (\xv - \Rv_j)$
we obtain
\begin{equation}
     t_{ij} = \int \rd^3 x \: \widetilde \Ph^* (\xv - \Rv_i)
              \biggl[\frac{\|\pv\|^2}{2m} + V(\xv)\biggr] \widetilde \Ph(\xv - \Rv_j) \epp
\end{equation}
Discuss under which conditions the approximation 
$\widetilde \Ph(\xv - \Rv_j) \approx \re^{\i e \la(\Rv_j,t)/c}
\Ph(\xv - \Rv_j)$, the so-called Peierls substitution, should be valid.
\item
Show that the Peierls substitution leads to a modification of
the hopping part of the many body Hamiltonian according to the
rule $t_{ij} \rightarrow t_{ij} \re^{\i \la_{ij}}$, where
\begin{equation}
     \la_{ij} = \frac{e}{c} \int_{\Rv_i}^{\Rv_j} \rd y_k \Av^k (\yv,t) \epp
\end{equation}
\end{enumerate}

\mysection{The strong-coupling limit}
In this lecture we shall consider the limiting case, when the
intra-atomic Coulomb interaction $U$ of the Hubbard model is large
compared to the band width $t$. We define
\begin{equation} \label{defopt}
     T = \sum_{j, k = 1}^L t_{jk} c_{j, a}^+ c_{k, a}^{} \epp
\end{equation}
Here we only assume that $t_{jk} = t_{kj}^*$, guaranteeing Hermiticity
of $T$, and that $t_{jj} = 0$. For fixed particle number the latter
setting merely shifts the energy scale and, for this reason, does
not imply any restriction to generality. As we shall see, however,
this assumption will have several technical advantages in the
calculations below.

Using an appropriate definition of the $t_{jk}$ and an appropriate
enumeration of the lattice sites in (\ref{defopt}), we may write
the general Hubbard Hamiltonian (\ref{hhubgen}) on any finite lattice,
under any kind of boundary conditions and in any dimension in the form
\begin{equation} \label{hforstrong}
     H = T + U D \epp
\end{equation}

We shall assume that $U > 0$. This is natural, since positive
$U$ corresponds to the repulsion of electrons in the same Wannier
orbital which is expected as a consequence of their mutual 
Coulomb repulsion. If $|t_{jk}| \ll U$ we can consider $T$ as
a small perturbation of $UD$. As we have seen in (\ref{doev}),
$D$ counts the number of double-occupied sites. Thus, the
eigenvalues of $UD$ are $0, U, 2U, \dots, LU$. Their number grows
linearly with $L$, while the number of states in ${\cal H}^{(L)}$
grows exponentially like $4^L$. The eigenvalues of $UD$ are therefore
highly degenerate. Let us denote the projection operators onto
the corresponding eigenspaces ${\cal H}_n$ by $P_n$, $n = 0, 1,
\dots, L$. Then
\begin{equation}
     {\cal H}^{(L)} = {\cal H}_0 \oplus {\cal H}_1 \oplus \dots \oplus {\cal H}_L \epc
\end{equation}
and $D$ has the spectral decomposition
\begin{equation}
     D = \sum_{n=0}^L n P_n \epp
\end{equation}
If the particle number $N \le L$, then the ground states of
$UD$ are in ${\cal H}_0$ which has $\dim {\cal H}_0 = 3^L$.

\subsection{Degenerate perturbation theory -- a reminder}
Consider a Hamiltonian $H + \la V$ on a Hilbert space ${\cal H}$.
Assume that $H$ has the spectral decomposition
\begin{equation}
     H = \sum_n E_n P_n
\end{equation}
with mutually distinct eigenvalues $E_n$ and orthogonal
projectors $P_n$ (i.e.\ $P_n P_m = \de_{n m} P_n$), such
that $\dim P_n {\cal H}$ is not necessarily equal to one.
We consider $\la V$, $\la \in {\mathbb R}$, as a small
perturbation of `strength $\la$'.

Define
\begin{equation}
     R_n = \sum_{m, m \ne n} \frac{P_m}{E_n - E_m} \epp
\end{equation}
We assume that the eigenvalues and eigenvectors of the perturbed
Hamiltonian $H + \la V$ are characterized by the quantum numbers
$n$ of the unperturbed problem and sets of additional quantum numbers
$\n_n \in \{1, \dots, \dim P_n {\cal H}\}$. In other words
\begin{equation} \label{pertsplitdeg}
     (H + \la V) |\Ps_{n, \n_n}\> = \e_{n, \n_n} |\Ps_{n, \n_n}\>
\end{equation}
in such a way that
\begin{equation}
     \lim_{\la \rightarrow 0} |\Ps_{n, \n_n}\> =
        |\Ps_{n, \n_n}^{(0)}\> \in P_n {\cal H} \epc \qd
     \lim_{\la \rightarrow 0} \e_{n, \n_n} = E_n \epp
\end{equation}
We rewrite (\ref{pertsplitdeg}) in the form
\begin{equation} \label{pertsplitdegtwo}
     (E_n - H) |\Ps_{n, \n_n}\>
        = \bigl(\la V - (\underbrace{\e_{n, \n_n} - E_n}_{= \D_{n, \n_n}})\bigr)
	  |\Ps_{n, \n_n}\> \epp
\end{equation}
Then
\begin{align} \label{pertevdeg}
     R_n & (E_n - H) |\Ps_{n, \n_n}\>
        = (1 - P_n) |\Ps_{n, \n_n}\> = R_n (\la V - \D_{n, \n_n}) |\Ps_{n, \n_n}\>
	\epc \notag \\[1ex] & \Leftrightarrow\
          |\Ps_{n, \n_n}\> = \underbrace{P_n |\Ps_{n, \n_n}\>}_{= |\ph_{n, \n_n}\>}
	  + R_n (\la V - \D_{n, \n_n}) |\Ps_{n, \n_n}\> \epc
	\notag \\ & \Leftrightarrow\
          |\Ps_{n, \n_n}\> =
	    \sum_{k=0}^\infty
	      \bigl[R_n (\la V - \D_{n, \n_n})\bigr]^k |\ph_{n, \n_n}\> \epp
\end{align}
Here the series converges if ${\cal H}$ is finite dimensional and
$|\la|$ is small enough. Otherwise we may have to interpret it as
an asymptotic series. Applying, on the other hand, $P_n$ to
(\ref{pertsplitdegtwo}) we obtain
\begin{equation}
     P_n V |\Ps_{n, \n_n}\> = \frac{\D_{n, \n_n}}{\la} |\ph_{n, \n_n}\> \epp
\end{equation}
Thus,
\begin{equation} \label{nlspecprob}
     \sum_{k=0}^\infty P_n V \bigl[R_n (\la V - \D_{n, \n_n})\bigr]^k |\ph_{n, \n_n}\>
        = \frac{\D_{n, \n_n}}{\la} |\ph_{n, \n_n}\> \epp
\end{equation}

This is a non-linear spectral problem on $P_n {\cal H}$ describing
the splitting of the energy level $E_n$ under the influence of
the perturbation $\la V$. The operator on the left hand side
is an effective Hamiltonian on $P_n {\cal H}$. Up to second order
($\D_{n, \n_n} = \la \D_{n, \n_n}^{(1)} + \la^2 \D_{n, \n_n}^{(2)}
+ \dots$) it is given by
\begin{multline} \label{htwoabstract}
     H_2 = P_n V P_n + P_n V R_n (\la V - \D_{n, \n_n}) P_n
         = P_n V P_n + \la P_n V R_n V P_n \\[1ex]
	 = P_n V P_n + \la \sum_{m, m \ne n} \frac{P_n V P_m V P_n}{E_n - E_m} \epp
\end{multline}
Thus, up to second order in $\la$ equation (\ref{nlspecprob}) reduces
to a linear spectral problem with an effective Hamiltonian $H_2$.
\begin{remark}
The corresponding eigenstates of the perturbed problem are obtained
from (\ref{pertevdeg}), once the $|\ph_{n, \n_n}\>$ are known.
\end{remark}
\subsection{Application to the Hubbard model at strong coupling}
We now apply (\ref{htwoabstract}) to (\ref{hforstrong}). For this
purpose we divide by $U$. Then $H/U = D + T/U$. Recall that $D$ has
eigenvalues $0, 1, \dots, L$. Inserting this for $n = 0$ with $V = T$,
$\la = 1/U$ into (\ref{htwoabstract}) we obtain
\begin{equation} \label{htwoproject}
     H_2 = P_0 T P_0 - \frac{1}{U} \sum_{m=1}^L \frac{P_0 T P_m T P_0}{m} \epp
\end{equation}

\subsection{Explicit form of the projection operators}
In order to express the projection operators $P_m$ in terms of Fermi
operators, we introduce the function
\begin{equation}
     G(\a) = \prod_{j=1}^L (1 - \a n_{j \auf} n_{j \ab}) \epp
\end{equation}
Its action on the Wannier basis is (see (\ref{nact}))
\begin{equation}
     G(\a) |\xv; \av\> = (1 - \a)^n |\xv; \av\> \epc
\end{equation}
where $n$ is the number of double-occupied sites in $|\xv; \av\>$.
In particular,
\begin{equation}
     G(1) = \prod_{j=1}^L (1 - n_{j \auf} n_{j \ab}) = P_0 \epp
\end{equation}
Moreover,
\begin{align} \label{nodoublegenfun}
     \frac{(-1)^k}{k!} \6_\a^k G(\a)\Bigl|_{\a = 1} |\xv; \av\>
        & = |\xv; \av\>
	  \begin{cases}
	       \binom{n}{k} (1 - \a)^{n-k}\bigl|_{\a = 1} & k \le n \\
	       0 & k > n
	  \end{cases} \notag \\[1ex]
        & = \de_{n, k} |\xv; \av\> \epp
\end{align}
It follows that
\begin{equation}
     P_n = \frac{(-1)^n}{n!} \6_\a^n G(\a)\Bigl|_{\a = 1} \epc
\end{equation}
meaning that $G(\a)$ is a generating function for the projection
operators $P_n$, $n = 0, 1, \dots, L$.

\subsection{\boldmath Application to $H_2$}
We will now use the explicit construction of the projection operators
to express $H_2$ in terms of Fermions. First of all
\begin{equation} \label{pnullnull}
     P_0 n_{j \auf} n_{j \ab}
        = \biggl[\prod_{k=1}^L (1 - n_{k \auf} n_{k \ab}) \biggr]
	                        n_{j \auf} n_{j \ab}
	= 0 = n_{j \auf} n_{j \ab} P_0 \epc
\end{equation}
entailing that
\enlargethispage{3ex}
\begin{multline}
     G(\a) T P_0 = \sum_{j, k = 1}^L t_{jk} \biggl[
                      \prod_{\ell=1}^L (1 - \a n_{\ell \auf} n_{\ell \ab}) \biggr]
		      c_{j, a}^+ c_{k, a}^{} P_0 \\[-1ex]
        = \sum_{j, k = 1}^L t_{jk} (1 - \a n_{j \auf} n_{j \ab})
	                    c_{j, a}^+ c_{k, a}^{} P_0 \epc
\end{multline}
and further, using (\ref{nodoublegenfun}),
\begin{equation}
     \sum_{m=1}^L \frac{P_m T P_0}{m}
        = \sum_{j, k = 1}^L t_{jk} n_{j \auf} n_{j \ab} c_{j, a}^+ c_{k, a}^{} P_0 \epp
\end{equation}
Using the latter equation in (\ref{htwoproject}) we arrive at
\begin{multline} \label{htwotwo}
     H_2 = P_0 \biggl\{\sum_{j, k = 1}^L t_{jk} c_{j, a}^+ c_{k, a}^{}
           - \frac{1}{U} \sum_{j,k,k',\ell=1}^L t_{jk} t_{k'\ell} c_{j, a}^+ c_{k, a}^{}
	     n_{k' \auf} n_{k' \ab} c_{k', b}^+ c_{\ell, b}^{} \biggr\} P_0 \\[1ex]
         = P_0 \biggl\{\sum_{j, k = 1}^L t_{jk} c_{j, a}^+ c_{k, a}^{}
           - \frac{1}{U} \sum_{j,k,\ell=1}^L t_{jk} t_{k\ell} c_{j, a}^+ c_{k, a}^{}
	     n_{k \auf} n_{k \ab} c_{k, b}^+ c_{\ell, b}^{} \biggr\} P_0 \epp
\end{multline}
Here we have used (\ref{pnullnull}) and the fact that $t_{jj} = 0$
in the second equation.

$H_2$ is an effective Hamiltonian describing the splitting of the
lowest level ($n = 0$) of $D$ under the influence of the
perturbation by $T$. Sometimes $H_2$ is called `the $t$-$J$
Hamiltonian'. Using the canonical anti-commutation relations
of the Fermi operators equation (\ref{htwotwo}) can be further
simplified. After a straightforward but slightly cumbersome
calculation we obtain
\begin{multline} \label{htwofinal}
     H_2 = \sum_{\substack{j,k=1\\j \ne k}}^L
           \biggl\{t_{jk} c_{j, a}^+ c_{k, a}^{}(1 - n_j)
                   + \frac{2 |t_{jk}|^2}{U}
		     \Bigl(S_j^\a S_k^\a - \frac{n_j n_k}{4}\Bigr)\biggr\} \\[-2ex]
           - \frac{1}{U} \sum_{\substack{j,k,\ell=1\\j \ne k \ne \ell \ne j}}^L
	     t_{jk} t_{k\ell} \, n_k c_{j, a}^+ c_{k, a}^{}
	     c_{k, b}^+ c_{\ell, b}^{} (1 - n_j) \epp
\end{multline}
Note that the Hamiltonian leaves the space ${\cal H}_0$ invariant
by construction. In (\ref{htwofinal}) we have introduced the
notation
\begin{subequations}
\begin{align}
     n_j & = n_{j \auf} + n_{j \ab} \epc \\[1ex]
     S_j^\a & = \2 \sum_{a, b = \auf, \ab} \bigl(\s^\a)^a_b c_{j, a}^+ c_{j, b}^{} \epc \qd
     \a = x, y, z \epc
\end{align}
\end{subequations}
for the particle density and spin density operators. The matrices
$\s^\a$ are the well-known Pauli matrices.

\subsection{Exercise 19. Strong coupling limit of the Hubbard model}
Obtain (\ref{htwofinal}) from (\ref{htwotwo})!

\mysection{Heisenberg model and Mott transition}
\label{sec:heisenberg}
\subsection{Heisenberg Hamiltonian in the language of Fermi operators}
In the previous lecture we have derived the strong coupling
Hamiltonian $H_2$ by considering the `hopping' $T$ as small
perturbation to the atomic limit $UD$ of the Hubbard Hamiltonian
on the subspace ${\cal H}_0 \subset {\cal H}^{(L)}$ with no
double-occupied sites. The particle number operator on
${\cal H}^{(L)}$ is
\begin{equation}
     \hat N = \sum_{j=1}^L n_j \epp
\end{equation}
Since $[\hat N, P_n] = [\hat N, T] = 0$ we see from (\ref{htwoproject})
that $H_2$ preserves the particle number, $[\hat N, H_2] = 0$. This
implies that $H_2$ leaves the subspaces ${\cal H}_{0, N} \subset
{\cal H}_0$, $N = 0, \dots, L$, of fixed particle numbers $N$ invariant.
The case $N = L$ is called the case of `half-filling'. The corresponding
subspace ${\cal H}_{0, L}$ is spanned by all states of the form
\begin{equation}
     |\av\>_s = c_{L, a_L}^+ c_{L-1, a_{L-1}}^+ \dots c_{1, a_1}^+ |0\> \epp
\end{equation}
These are states for which every site is occupied by exactly one
electron. Thus,
\begin{equation} \label{onlyone}
     (1 - n_j) |\av\>_s = 0
\end{equation}
for $j = 1, \dots, L$ and for all $|\av\>_s \in {\cal H}_{0, L}$.

Hence, we can conclude with (\ref{htwofinal}) and (\ref{onlyone})
that the action of $H_2$ on the space ${\cal H}_{0, L}$ reduces
to the action of the Hamiltonian
\begin{equation}
     H_{\rm spin} =
        \sum_{\substack{j, k = 1\\j \ne k}} \frac{2 |t_{jk}|^2}{U}
	   \Bigl(S_j^\a S_k^\a - \4 \Bigr)
\end{equation}
which is called the (isotropic) Heisenberg Hamiltonian with exchange
integrals
\begin{equation}
     J_{jk} = \frac{2 |t_{jk}|^2}{U} \epp
\end{equation}
If we start from the Hubbard model with nearest-neighbour hopping
\begin{equation}
     t_{jk} = \begin{cases}
                 - t & \text{for nearest-neighbour sites} \\
                 0 & \text{else} 
              \end{cases} \epc
\end{equation}
the resulting `spin Hamiltonian' at half-filling is
\begin{equation} \label{heisenberg}
     H_{\rm spin} = J \sum_{\<jk\>} \Bigl(S_j^\a S_k^\a - \4\Bigr) \epc
                    \qd \text{where}\ J = \frac{2 t^2}{U} \epp
\end{equation}
This is called the Heisenberg-Hamiltonian with nearest-neighbour
exchange interaction. The Heisenberg model is \emph{the} model
for the antiferromagnetism of insulators (which is ubiquitous in
nature).

\subsection{Heisenberg model in the language of spin operators}
The space ${\cal H}_{0, L}$, spanned by the states $|\av\>_s$
with $a_j \in \{\auf, \ab\}$, $j = 1, \dots, L$, is $2^L$ dimensional,
hence isomorphic to $\bigl({\mathbb C}^2\bigr)^{\otimes L}$. This
makes it possible to describe the action of $H_{\rm spin}$ directly
by certain matrices acting on $\bigl({\mathbb C}^2\bigr)^{\otimes L}$.

The vectors space $\bigl({\mathbb C}^2\bigr)^{\otimes L}$ has the
canonical basis vectors
\begin{equation}
     |\av\> = \ev_{a_1} \otimes \dots \otimes \ev_{a_L} \epc
               \qd \ev_\auf = \binom{1}{0}\epc\ \ev_\ab = \binom{0}{1} \epp
\end{equation}
We would like to identify these vectors with the vectors $|\av\>_s$.
Then, on the one hand,
\begin{align}
     S_j^\a |\av\>_s & = \2 \bigl(\s^\a\bigr)^a_b c_{L, a_L}^+ \dots c_{j+1, a_{j+1}}^+
			  \mspace{-30.mu}
			  \underbrace{c_{j, a}^+ c_{j, b}^{} c_{j, a_j}^+}_{
			     = c_{j, a_j}^+ c_{j, a}^+ c_{j, b}^{} + \de^b_{a_j} c_{j, a}^+}
			  \mspace{-30.mu} \dots c_{1, a_1}^+ |0\> \notag \\
                      & = \2 \bigl(\s^\a\bigr)^a_{a_j}
		          |(a_1, \dots, a_{j-1}, a, a_{j+1}, \dots, a_L)\>_s \epc
\end{align}
while, on the other hand,
\begin{multline} \label{spinctwoell}
     \2 \bigl(\s^\a\bigr)^a_{a_j} |(a_1, \dots, a_{j-1}, a, a_{j+1}, \dots, a_L)\> \\
        = \2 \ev_{a_1} \otimes \dots \otimes \ev_{a_{j-1}} \otimes
	     \underbrace{\bigl(\s^\a\bigr)^a_{a_j} \ev_a}_{\s^\a \ev_{a_j}}
	     \otimes \ev_{a_{j+1}} \otimes \dots \ev_{a_L} \\
        = \underbrace{\2 \Bigl(I_2^{\otimes (j-1)} \otimes
	              \s^\a \otimes I_2^{\otimes (L-j)}\Bigr)}_{= S_j^\a} |\av\> \epp
\end{multline}
Thus, the isomorphism ${\cal H}_{0, L} \cong \bigl({\mathbb C}^2\bigr)^{\otimes L}$
induces the identification of operators
\begin{equation} \label{idspinfermi}
     \2 I_2^{\otimes (j-1)} \otimes \s^\a \otimes I_2^{\otimes (L-j)} \mapsto
        \2 \bigl(\s^\a\bigr)^a_b c_{j, a}^+ c_{j, b}^{} \epp
\end{equation}
Since the operators $\s^\a$, $\a = x, y, z$, and $I_2$ form a basis of
$\End \bigl({\mathbb C}^2\bigr)$, the operators $S_j^\a$, $j = 1, \dots, L$,
as defined in (\ref{spinctwoell}), together with the identity, generate
a basis of $\End \bigl({\mathbb C}^2\bigr)^{\otimes L}$. With the identification
(\ref{idspinfermi}) we can interpret the Heisenberg model on `spin space'
$\bigl({\mathbb C}^2\bigr)^{\otimes L}$ with Hamiltonian (\ref{heisenberg}),
where now the $S_j^\a$ are the spin matrices defined in (\ref{spinctwoell}).
\begin{remark}
Such an identification is not possible for the $t$-$J$ model. The operator
$H_2$ acts on ${\cal H}_0$ which contains ${\cal H}_{0, L}$ as a proper
subspace. We have e.g.\ $S_j^\a|0\> = 0$ (where $|0\>$ is the vacuum for
Fermions).
\end{remark}

\subsection{Interpretation of the exchange interaction}
Define
\begin{equation} \label{deftransposition}
     P = \2 \bigl(1 + \s^\a \otimes \s^\a\bigr)
       = \begin{pmatrix}
            1 &&& \\
	    && 1 & \\
	    & 1 && \\
	    &&& 1
         \end{pmatrix} \epp
\end{equation}
Then
\begin{equation}
     P \xv \otimes \yv =
        \begin{pmatrix} 1 &&& \\ && 1 & \\ & 1 && \\ &&& 1 \end{pmatrix}
        \begin{pmatrix} x_1 y_1 \\ x_1 y_2 \\ x_2 y_1 \\ x_2 y_2 \end{pmatrix}
	= \yv \otimes \xv \epp
\end{equation}
This means that $P$ is a permutation (transposition) matrix. In particular,
\begin{equation}
     P |\auf \ab\> = |\ab \auf\> \epp
\end{equation}
Define `symmetrizer' and `antisymmetrizer'
\begin{equation} \label{defsymantisym}
     P^\pm = \2 (P \pm 1) \epp
\end{equation}
They satisfy
\begin{subequations}
\begin{align}
     & \bigl(P^\pm\bigr)^2 = \4 (P \pm 1)^2 = \2 (1 \pm P) = \pm P^\pm \epc \\
     & P^+ P^- = P^- P^+ = 0 \epp
\end{align}
\end{subequations}
This means that $\pm P^\pm$ are orthogonal projectors onto the 
symmetric and antisymmetric subspaces of ${\mathbb C}^2 \otimes
{\mathbb C}^2$. The subspace $P^+ {\mathbb C}^2 \otimes {\mathbb C}^2$
has a basis
\begin{equation}
     B_t = \biggl\{|1,1\> = |\auf \auf\>,\ |1,0\> = \frac{1}{\sqrt{2}} \bigl(|\auf \ab\> +
                  |\ab \auf\>\bigr),\ |1,-1\> = |\ab \ab\>\biggr\} \epc
\end{equation}
while the 1d subspace $P^- {\mathbb C}^2 \otimes {\mathbb C}^2$
is generated by
\begin{equation}
     B_s = \biggl\{|0,0\> = \frac{1}{\sqrt{2}} \bigl(|\auf \ab\> -
                  |\ab \auf\>\bigr)\biggr\} \epp
\end{equation}
These bases are called `spin triplet' and `spin singlet'.

Comparing (\ref{deftransposition}), (\ref{defsymantisym}) and
(\ref{heisenberg}) we see that for $L = 2$
\begin{equation}
     H_{\rm spin} = J P^- \epp
\end{equation}
It follows that $H_{\rm spin}$ has eigenvalue $0, -J$. The corresponding
eigenvectors are $|1, s\>$, $s = 0, \pm 1$, and $|0,0\>$. If $J > 0$,
then the singlet $|0,0\>$ is the ground state. This is called
antiferromagnetism, since
\begin{equation}
     \s^z \otimes \s^z |0,0\> = - |0, 0\> \epc\ \then
     \<0,0|\s^z \otimes \s^z |0,0\> = - 1 \epp
\end{equation}
One says that the ground state has antiferromagnetic correlations.

In the general case
\begin{equation}
     H_{\rm spin} = J \sum_{\<jk\>} P_{jk}^-
\end{equation}
is a sum of projectors onto `local singlets'. An antiferromagnet
($J > 0$) may therefore be characterized as a system which prefers
local singlets. By contrast, a ferromagnet ($J < 0$) prefers
triplets. Local singlets are incompatible with the global $SU(2)$
invariance of the Hamiltonian. For this reason the ground state
of a macroscopic Heisenberg antiferromagnet is a complicated
`many-body state'. On the other hand, the state
\begin{equation}
     |\auf \dots \auf\> = \ev_\auf^{\otimes L}
\end{equation}
may be seen as a tensor product of states from local triplets.
It is annihilated by $H_{\rm spin}$ and is one of its ground
states in the ferromagnetic case. The full ground state subspace
can be constructed using the global $SU(2)$ symmetry.

\subsection{Mott transition}
In an external electromagnetic field the hopping term in the
Hubbard Hamiltonian is modified by so-called Peierls phases
\begin{equation}
     t_{jk} \rightarrow t_{jk} \re^{\i \la_{jk}} \epc \qd
        \la_{jk} \in {\mathbb R} \epp
\end{equation}
(cf.\ section~\ref{sec:peierlsph}). The Hubbard interaction
$U$, on the other hand remains, unchanged. It follows that
the corresponding Heisenberg model (strong-coupling limit
at half-filling) is \emph{not} modified, since it depends only
on $|t_{jk}|^2$. In other words, to leading order perturbation
theory the Hubbard model at half-filling does not couple
to an external field. The system is an insulator, although
it has an odd number of electrons per unit cell, and therefore
is a conductor at $U = 0$. This suggests that there might
be an interaction induced metal-insulator transition
(`Mott transition') somewhere in between, at a finite value
$U_c$ of the interaction. It is believed that such a transition
occurs indeed in any spatial dimension. A proof only
exists in $d = 1$, where $U_c = 0$.

\subsection{Exercise 20: A short XXX chain}
Consider the periodic Heisenberg Hamiltonian (XXX model) on a four-site 1d lattice
\begin{equation}
     H = P_{12} + P_{23} + P_{34} + P_{41} \epp
\end{equation}
Here $P_{jk}$ is the transposition, interchanging spin states on sites $j$ and $k$
of the chain. Recall that the space of states of the model is the tensor
product $\mathcal{H} = {\mathbb C}^2 \otimes {\mathbb C}^2 \otimes {\mathbb C}^2
\otimes {\mathbb C}^2$. The vectors $e_\uparrow = \left(\begin{smallmatrix} 1 \\0
\end{smallmatrix}\right)$ and $e_\downarrow = \left(\begin{smallmatrix} 0 \\1
\end{smallmatrix}\right)$ form a basis of ${\mathbb C}^2$. Hence,
\begin{equation}
     \mathcal{B} = \Bigl\{ |a_1 a_2 a_3 a_4\> = e_{a_1} \otimes e_{a_2}
                           \otimes e_{a_3} \otimes e_{a_4} \in \mathcal{H}\
			   \big|\ a_1, a_2, a_3, a_4 = \uparrow, \downarrow \Bigr\}
\end{equation}
is a basis of $\mathcal{H}$. 

Denote the embeddings of the Pauli matrices into $\End (\mathcal{H})$ by
$\sigma_j^\alpha$, $j=1,2,3,4$, $\alpha=x,y,z$. Then the transpositions can
be expressed as
\begin{equation}
     P_{jk} = \2 \bigl(\id + \sigma_j^\alpha\sigma_k^\alpha\bigr) \epp
\end{equation}
The operator of the total spin has components $S^\alpha = \frac{1}{2}
(\sigma_1^\alpha + \sigma_2^\alpha + \sigma_3^\alpha + \sigma_4^\alpha)$ which
can be used to define the ladder operators $S^\pm = S^x \pm \i S^y$. The shift
operator for the chain of length four is defined as
\begin{equation}
     \hat U = P_{12}P_{23}P_{34} \epp
\end{equation}

\begin{enumerate}
\item \label{i}
Show that the operators $H$, $S^2$, $S^z$ and $\hat U$ can be simultaneously
diagonalized. For this purpose verify the commutation relations
\begin{equation}
     [H, S^\alpha] = [\hat U, S^\alpha] = [H, \hat U] =0 \epc \quad
     [S^\alpha, S^\beta] = \i \varepsilon^{\alpha\beta\gamma} S^\gamma
\end{equation}
for $\a, \be = x, y, z$.

\item
Construct a basis of common eigenvectors and obtain the corresponding
eigenvalues of the above operators. Sketch the spectrum of $H$.
Hint: Use, for instance, that $\hat U^4=\mathrm{id}$, and use the
angular momentum algebra.
\end{enumerate}

\addtocontents{toc}{\protect\pagebreak}
\mysection{Linear response theory}
In this lecture we study the response of a quantum many-body system,
in thermal equilibrium with a heat bath of temperature $T$, to 
a small external perturbation. A solid is responding, for instance,
with a current to an external voltage, with a thermal current to
a temperature gradient, or with a deformation to external mechanic
stress. An example which will be considered in some detail is the
absorption of microwaves by a spin system in a homogeneous magnetic
field (ESR experiment). This will allow us to take a glimpse at our
own work \cite{BGKKW11,Zeisner_etal17}.

\subsection{Time evolution of the statistical operator}
\renewcommand{\CH}{H}
Consider a quantum system with Hamiltonian $\CH$ possessing a
discrete spectrum $(E_n)_{n=0}^\infty$ with corresponding eigenstates
$\{|n\>\}_{n=0}^\infty$. At time $t_0$ a time-dependent perturbation
$V(t)$ is adiabatically switched on. We are interested in the time
evolution of the system, assuming it initially, at times $t < t_0$,
in an equilibrium state described by the statistical operator
\begin{equation}
     \r_0 = \frac1Z \sum_{n=0}^\infty \re^{- \frac{E_n}T} |n\>\<n|
\end{equation}
of the canonical ensemble. Here $T$ denotes the temperature and $Z$
the canonical partition function.

Let $U(t)$ the time evolution operator of the perturbed system. It
satisfies the Schr\"odinger equation
\begin{equation} \label{ut}
     \i \6_t U(t) = \bigl( \CH + V(t) \bigr) U(t)
\end{equation}
with initial condition $U(t_0) = \id$ (note that we have set
$\hbar = 1$ as before). Under the influence of the perturbation
the eigenstate state $|n\>$ evolves into $|n, t\> = U(t) |n\>$,
and the statistical operator at time $t$ becomes
\begin{equation} \label{rhot}
     \r (t)  = \frac1Z \sum_{n=0}^\infty \re^{- \frac{E_n}T} |n, t\>\<n, t|
             = U (t) \r_0 U^{-1} (t) \epp
\end{equation}

We are interested in an approximation to $\r (t)$ linear in the
strength of the perturbation $V(t)$. In order to derive this
approximation we define
\begin{subequations}
\begin{align}
     R(t) & = \re^{\i \CH t} \bigl(\r(t) - \r_0\bigr) \re^{- \i \CH t} \epc \\[1ex]
     W(t) & = \re^{\i \CH t} V(t) \re^{- \i \CH t} \epp
\end{align}
\end{subequations}
Then
\begin{multline}
	\i \6_t R(t) = \i \6_t \re^{\i \CH t} U (t) \r_0
	                 \bigl( \re^{\i \CH t} U(t) \bigr)^{-1} \\[1ex]
	             = [ W(t), \re^{\i \CH t} \r(t) \re^{- \i \CH t}]
		     = [ W(t), R(t) + \r_0] \epp
\end{multline}
Since $R(t_0) = 0$ by construction, we obtain
\begin{equation}
     R(t) = - \i \int_{t_0}^t \rd t' \: [W(t'), R(t') + \r_0] \epp
\end{equation}
This Volterra equation is an appropriate starting point for a
perturbation theory. Assuming $|W(t)|$ to be small we conclude that
\begin{equation}
     R(t) = - \i \int_{t_0}^t \rd t' \: [W(t'), \r_0] + \CO \bigl(|W|^2\bigr) \epc
\end{equation}
and therefore, to lowest order in $W$,
\begin{equation} \label{rhoborn}
     \r (t) = \r_0 - \i \re^{- \i \CH t} \int_{t_0}^t \rd t' \: [W(t'), \r_0]
              \re^{\i \CH t} \epp
\end{equation}
This is the statistical operator in so-called Born approximation.
In the following $t_0$ will be sent to $ - \infty$.

\subsection{Time evolution of expectation values}
Using \eqref{rhoborn} we can calculate the time evolution of the
expectation value of an operator $A$ under the influence of the
perturbation. We shall denote the canonical ensemble average by
$\<\ \cdot\ \>_T = \tr \{ \r_0\ \cdot\ \}$ and write $A(t) = \re^{\i \CH t}
A \re^{- \i \CH t}$ for the Heisenberg time evolution of $A$.
Then
\begin{multline}
  \de \< A\>_T = \tr \{ (\r (t) - \r_0) A\} 
               = -\i\int_{- \infty}^t \rd t' \: \tr \bigl\{ [W(t'), \r_0] \re^{\i \CH t}
                   A \re^{- \i \CH t} \bigr\} \\[1ex]
               = - \i \int_{- \infty}^t \rd t' \: \bigl\< [A(t), W(t')] \bigr\>_T
               = - \i \int_{- \infty}^t \rd t' \: \bigl\< [A(t - t'), V(t')]
	           \bigr\>_T \epp
\end{multline}
Here we have used the cyclic invariance of the trace in the third
equation and the fact that $H$ commutes with $\r_0$ in the fourth
equation.

A typical example of a perturbation, which will be relevant for our
discussion below, is a classical time-dependent field $h^\a (t)$ coupling
linearly to operators $X^\a$,
\begin{equation} \label{classpert}
     V(t) = h^\a (t) X^\a \epp
\end{equation}
In this case
\begin{equation} \label{aborn}
     \de \< A\>_T = - \i \int_{- \infty}^t \rd t' \: h^\a (t')
                      \bigl\< [A(t - t'), X^\a ] \bigr\>_T \epp
\end{equation}

\subsection{Absorption of energy}
The absorbed energy per unit time is
\begin{multline}
     \frac{d E}{dt} = \frac{d}{dt} \tr \{ \r (t) (\CH + V(t)) \} \\[1ex]
                    = - \i \tr \{ [\CH + V(t), \r (t)] (\CH + V(t)) \}
		      + \tr \{ \r (t) \dot V(t) \} \\[1ex]
		    = \<\dot V(t)\>_T + \de \< \dot V(t) \>_T \epp
\end{multline}
Here we used \eqref{ut},~\eqref{rhot} in the second equation and the cyclic
invariance of the trace in the third equation. Assuming that $V(t)$ is of the
form \eqref{classpert} and using \eqref{aborn} we obtain
\begin{equation} \label{esrabsorbenergy}
     \frac{d E}{dt} = \dot h^\a (t) \<X^\a\>_T
                      - \i \int_{- \infty}^t \rd t' \: \dot h^\a (t) h^\be (t')
			\bigl\< [X^\a (t - t'), X^\be ] \bigr\>_T \epp
\end{equation}

\subsection{Application to quantum spin chains}
Let us now apply the above formalism to the Heisenberg-Ising
(alias XXZ) spin chain in a longitudinal static magnetic field
of strength $h$. The Hamiltonian of this model is defined as
\begin{equation} \label{hamapp}
     \CH_0 = J \sum_{j=1}^L \bigl(s_{j-1}^x s_j^x + s_{j-1}^y s_j^y
                                  + \D s_{j-1}^z s_j^z \bigr) \epp
\end{equation}
Here we have switched from Pauli matrices $\s^\a$ to spin
operators $s^\a = \s^\a/2$. We shall assume periodic boundary
conditions $s_0^\a = s_L^\a$, $\a = x, y, z$. The real parameter
$\D$ is called the anisotropy parameter. For $\D = 1$ the
Hamiltonian $H_0$ turns into the Heisenberg Hamiltonian considered
in lecture~\ref{sec:heisenberg}. Values of $\D$ different from
one may be needed for a more accurate modeling of real magnetic
materials and are due to a combination of crystal symmetry,
spin-orbit interactions and dipole-dipole interactions.

We define the operator of the total spin as
\begin{equation}
     \Sv = \begin{pmatrix} S^x \\ S^y \\ S^z \end{pmatrix} \epc \qd
     S^\a = \2 \sum_{j=1}^L \s^\a \qd \text{for $\a = x, y, z$.}
\end{equation}
If a spin system is exposed to a homogeneous magnetic field $\hv$
a so-called `Zeeman term' $- \<\hv, \Sv\>$ must be added to the
Hamiltonian. Assuming a field in $z$-direction our Hamiltonian
takes the form
\begin{equation}
     H = H_0 - h S^z \epc
\end{equation}
Note that the $z$-direction is special in that $[H_0, S^z] = 0$.

We perturb the spin chain by a circular polarized electro-magnetic
wave propagating in $z$-direction. We assume that the wave length is 
large compared to the length of the spin chain\footnote{The wavelength
of microwaves is of the order of $10\, {\rm cm}$.} and idealize this
assumption by setting the wave number $k=0$. Then the magnetic field
component of the wave is
\begin{equation} \label{extfield}
     \hv (t) = A \begin{pmatrix}
	        \cos (\om t) \\ - \sin(\om t) \\ 0
             \end{pmatrix} \epc \qd
	     A > 0 \epp
\end{equation}
It couples to the system by another Zeeman term 
\begin{equation}
     V(t) = h^\a (t) S^\a \epp
\end{equation}

Using (\ref{esrabsorbenergy}) we obtain
\begin{multline} \label{dedt}
     \frac{d E}{dt}
        = \<S^\a\>_T \, \dot h^\a (t) - \i \int_{- \infty}^t \rd t' \:
	         \bigl\< [S^\a (t - t'), S^\be ] \bigr\>_T \,
		          \dot h^\a (t) h^\be (t') \\
        = \frac{A^2 \om}{4} \int_0^\infty \rd t'
	         \bigl\{ \re^{\i \om (2t - t')} \bigl\< [S^+ (t'), S^+] \bigr\>_T
		 - \re^{- \i \om (2t - t')} \bigl\< [S^- (t'), S^-] \bigr\>_T \\
		 + \re^{\i \om t'} \bigl\< [S^+ (t'), S^-] \bigr\>_T
		 - \re^{- \i \om t'} \bigl\< [S^- (t'), S^+] \bigr\>_T \bigr\} \epp
\end{multline}
Here we have used that
\begin{equation}
     \<S^\a\>_T \, \dot h^\a (t) = \<S^z\>_T \, \dot h^z (t) = 0 \epc
\end{equation}
since $\<S^\pm\>_T = \<S^x \pm \i S^y\>_T = \pm \<[S^z, S^\pm]\>_T = 0$
which holds, in turn, because $H$ commutes with $S^z$. Under the integral
we have used (\ref{extfield}),
\begin{align}
     & \bigl\< [S^\a (t - t'), S^\be ] \bigr\>_T \, \dot h^\a (t) h^\be (t')
        \notag \\[1ex] & \mspace{36.mu}
      = - A^2 \om \bigl\< [\sin(\om t) S^x (t - t') + \cos(\om t) S^y (t - t'),
                           \cos(\om t') S^x - \sin(\om t') S^y]\bigr\>_T
        \notag \\[1ex] & \mspace{36.mu}
      = \frac{\i A^2 \om}{4}
        \bigl\< [\re^{\i \om t} S^+ (t - t') - \re^{- \i \om t} S^- (t - t'),
	         \re^{\i \om t'} S^+ + \re^{- \i \om t'} S^-]\bigr\>_T
        \notag \\[1ex] & \mspace{36.mu}
      = \frac{\i A^2 \om}{4}
        \bigl\< \re^{\i \om (t + t')} [S^+ (t - t'), S^+]
	        - \re^{- \i \om (t + t')} [S^- (t - t'), S^-]
        \notag \\ & \mspace{126.mu}
                + \re^{\i \om (t - t')} [S^+ (t - t'), S^-]
	        - \re^{- \i \om (t - t')} [S^- (t - t'), S^+]\bigr\>_T \epp
\end{align}

The ability to absorb radiation is a material property. Hence, we generally
expect the absorbed energy per unit time to be proportional to the number of
constituents of a physical system and to diverge in the thermodynamic limit.
In order to define a quantity that truly characterizes the material and is
finite in the thermodynamic limit we should therefore normalize by the average
intensity $A^2$ of the incident wave and by the number of lattice sites $L$.
Further averaging the normalized absorption rate over a half-period $\pi/\om$ of the applied
field, we obtain the normalized absorbed intensity
\begin{align}
     I(\om, h) & = \frac{\om}{L A^2 \p} \int_0^{\frac\p\om} \rd t \: \frac{d E}{dt}
            \notag \\[1ex]
            & = \frac{\om}{4 L} \int_0^\infty \rd t \: \bigl\{
	      \re^{\i \om t} \bigl\< [S^+ (t), S^-] \bigr\>_T
            + \re^{- \i \om t} \bigl\< [S^+,S^- (t)] \bigr\>_T \bigr\}
	      \notag \\[1ex]
            & = \frac{\om}{4 L} \int_{- \infty}^\infty \rd t \:
	        \re^{\i \om t} \bigl\< [S^+ (t), S^-] \bigr\>_T \epp
\end{align}
Introducing the function
\begin{equation} \label{defchi}
     \chi_{+-}'' (\om, h) = \frac{1}{2L} \int_{- \infty}^\infty \rd t \:
        \re^{\i \om t} \bigl\< [S^+ (t), S^-] \bigr\>_T \epc
\end{equation}
the normalized absorbed intensity can be written as
\begin{equation}
     I (\om, h) = \frac\om2 \chi_{+-}'' (\om, h) \epp
\end{equation}
The function $\chi_{+-}'' (\om, h)$ is called the (imaginary part
of) the dynamic susceptibility. This function is a typical
`response function'. Note that it is calculated as Fourier transform
of the dynamical correlation function $\<[S^+ (t), S^-]\>_T$.

\mysection{Microwave absorption by the Heisenberg-Ising chain}
\subsection{The isotropic chain}
In the general case $\chi_{+-}''$ cannot be calculated. For $\D = 1$,
however, the situation simplifies drastically.
\begin{equation}
     [H, \Sv] = - h [S^z, \Sv] \epc
\end{equation}
and the Heisenberg equation of motion for $\Sv$ can be solved,
\begin{equation} \label{spmxxxt}
     \dot S^\pm = \i [H, S^\pm] = - \i h [S^z, S^\pm] = \mp \i h S^\pm \epc
        \qd \then S^\pm (t) = \re^{\mp \i ht} S^\pm \epp
\end{equation}
Also $S^z (t) = S^z$. Hence, the total spin behaves as
\begin{equation} \label{rots}
     \Sv (t) = \begin{pmatrix}
                  \cos(ht) & \sin(ht) & \\
		  - \sin(ht) & \cos(ht) & \\
		  && 1
               \end{pmatrix} \Sv \epp
\end{equation}
It is precessing clockwise about the $z$ axis.

On the other hand, inserting (\ref{spmxxxt}) into (\ref{defchi}) we obtain
\begin{equation} \label{chixxx}
     \chi_{+-}'' (\om) = \frac{1}{2L} \int_{- \infty}^\infty \rd t \:
                         \re^{\i (\om - h) t} \bigl\< [S^+, S^-] \bigr\>_T
                       = 2\p \de(\om - h) m(T, h) \epp
\end{equation}
where $m(T, h) = \<S^z\>_T/L$ is the magnetization per lattice site. The
corresponding normalized absorbed intensity is
\begin{equation} \label{intxxx}
     I (\om) = \p \de(\om - h) h\, m(T, h)
\end{equation}
and is proportional to the magnetic energy $h\, m(T,h)$ per lattice site.
This case includes the familiar paramagnetic resonance (Zeeman effect)
for which the magnetization is known more explicitly, namely, $m(T,h)
= \2 \tgh \bigl( \frac{h}{2T} \bigr)$ for $J = 0$. In general, an
exact calculation of the magnetization of the isotropic Heisenberg
chain at any finite temperature is not elementary and requires
the machinery of Bethe Ansatz and quantum transfer matrix \cite{Kluemper93}.

Comparing (\ref{extfield}), (\ref{rots}) and (\ref{intxxx}) we interpret
the absorption of energy as a resonance between the rotating field of the
incident wave and the precessing total spin of the chain. Both are
rotating clockwise with angular velocity $\om = h$. If we deviate
from the isotropic point $\D = 1$ of the Hamiltonian (\ref{ham}) we
expect that energy is transferred form the `coherent motion of the
total spin' to `other modes', causing a damping of the spin precession
and hence a shift and a broadening of the $\de$-function shaped spectral
line (\ref{intxxx}).

\subsection{Resonance shift and line width in the anisotropic case}
At the present state of the art dynamical correlation functions, such as
$\chi_{+-}''$, cannot be calculated exactly, not even for models as
simple as the XXZ chain. The interaction induced shift of the resonance
frequency and the width of the spectral line, on the other hand, are
more simple quantities which need less information to be calculated.

For every finite $L$ and for every $n \in {\mathbb N}$ the integrals
\begin{equation}
     I_n = \int_{- \infty}^\infty \rd \om \: \om^n I (\om)
\end{equation}
exist. Since $I (\om)$ is everywhere non-negative and since $I_0 > 0$
(see below), we may interpret $I (\om)/I_0$ as a probability
distribution with moments $I_n/I_0$. $I_1/I_0$ is the mean value
of the distribution and $\bigl(I_2/I_0 - (I_1/I_0)^2\bigr)^\frac{1}{2}$
its variance. The quantities may be used to define the `resonance
frequency' and the line width.

Instead of the moments of the normalized absorption intensity
we shall use the `shifted moments' of the dynamical susceptibility,
\begin{equation}
	m_n = \int_{- \infty}^\infty
	      \frac{\rd \om}{2 \p J^n} (\om - h)^n \chi_{+-}'' (\om) \epp
\end{equation}
As we shall see these quantities are more natural from a theoretical
point of view as they can be more easily calculated. Using the binomial
formula we can relate them to the moments of the normalized intensity,
\begin{equation}
     \frac{I_n}{I_0} =
        \frac{\int_{- \infty}^\infty \rd \om \: \om^{n+1} \chi_{+-}'' (\om)}
             {\int_{- \infty}^\infty \rd \om \: \om \chi_{+-}'' (\om)} =
        \frac{\sum_{k=0}^{n+1} h^k J^{n+1-k} m_{n+1-k}} {J m_1 + h m_0} \epp
\end{equation}

Let us now define
\begin{equation} \label{rshiftmom}
     \de \om = \<\om\> - h = \frac{I_1}{I_0} - h
        = J \frac{J m_2 + h m_1}{J m_1 + h m_0} \epp
\end{equation}
$\de \om$ is the `resonance shift', i.e., the deviation of the resonance
frequency from the resonance frequency in the isotropic case. Similarly,
\begin{equation} \label{widthmom}
     \D \om^2 = \<\om^2\> - \<\om\>^2 = \frac{I_2}{I_0} - \frac{I_1^2}{I_0^2} =
        J^2 \frac{J m_3 + h m_2}{J m_1 + h m_0} - \de \om^2
\end{equation}
is a measure of the line width. Hence, in order to calculate the
resonance shift and the line width, we need to know the first
four shifted moments $m_0$, $m_1$, $m_2$, $m_3$ of the dynamic
susceptibility $\chi_{+-}''$.

\subsection{The shifted moments of the dynamic susceptibility}
In the following we shall employ the notation $\ad_X \cdot =
[X,\ \cdot\ ]$ for the adjoint action of an operator $X$.
Then, using that $\ad_H = \ad_{H_0} - h \ad_{S^z}$, that
$[H_0, S^z] = 0$ and that $[S^z, S^+] = S^+$, we see that
\begin{equation}
     S^+ (t) = \re^{\i t \ad_H} S^+
             = \re^{\i t \ad_{H_0}} \re^{- \i h t \ad_{S^z}} S^+
	     = \re^{- \i ht} \re^{\i t \ad_{H_0}} S^+ \epp
\end{equation}
Thus,
\begin{multline} \label{chiad}
     \chi_{+-}'' (\om) = \frac{1}{2L} \int_{- \infty}^\infty \rd t \:
        \re^{\i (\om - h) t} \bigl\< [\re^{\i t \ad_{H_0}} S^+, S^-] \bigr\>_T \\[1ex]
        = \frac{1}{2L} \int_{- \infty}^\infty \rd t \:
        \re^{\i (\om - h) t} \bigl\< [S^+, \re^{- \i t \ad_{H_0}} S^-] \bigr\>_T
	\epp
\end{multline}
It follows that
\begin{multline} \label{calcallmoments}
     (\om - h)^n \chi_{+-}'' (\om)
        = \frac{(- \i)^n}{2L} \int_{- \infty}^\infty \rd t \:
          \bigl( \6_t^n \re^{\i (\om - h) t} \bigr)
	  \bigl\< [S^+, \re^{- \i t \ad_{H_0}} S^-] \bigr\>_T \\
        = \frac{1}{2L} \int_{- \infty}^\infty \rd t \: \re^{\i (\om - h) t}
	  \bigl\< [S^+, \ad_{H_0}^n \re^{- \i t \ad_{H_0}} S^-] \bigr\>_T \epc
\end{multline}
entailing that
\begin{equation} \label{allmoments}
     m_n = \frac1{2L} \bigl\< [S^+, \ad_{H_0/J}^n S^-] \bigr\>_T \epp
\end{equation}

The latter formula shows that the moments $m_n$ are static correlation
functions whose range and complexity grows with growing $n$. The first
few of them can be easily calculated by hand. The most basic one is
\begin{equation} \label{m0}
     m_0 = \frac1{2L} \bigl\< [S^+, S^-] \bigr\>_T
         = \frac1{L} \bigl\< S^z \bigr\>_T = m(T, h) \epc
\end{equation}
which is the magnetization per lattice site. The subsequent moments
vanish in the isotropic point $\D = 1$. It turns our that they
are polynomials in
\begin{equation}
     \de = \D - 1 \epp
\end{equation}
Unlike the magnetization they do not have an immediate interpretation.
Using (\ref{allmoments}) they can be calculated one by one. We obtain, for
instance,
\begin{subequations}
\label{m13}
\begin{align} \label{m1}
     m_1 & = \de \< s_1^+ s_2^- - 2 s_1^z s_2^z \>_T \epc \\[1ex]
     m_2 & = \2 \de^2 \< s_1^z + 4 s_1^z s_2^z s_3^z - 4 s_1^z s_2^+ s_3^- \>_T
               \epc \\[1ex]
     m_3 & = \4 \de^2 \bigl\< 2 s_1^+ s_2^+ s_3^- s_4^-
                          + 4 s_1^+ s_2^- s_3^+ s_4^-
                          - 2 s_1^+ s_2^- s_3^- s_4^+
                          - 8 s_1^z s_2^z s_3^+ s_4^-
                          - 4 s_1^z s_2^+ s_3^z s_4^- \notag \\ & \mspace{90.mu}
                          + 8 s_1^z s_2^+ s_3^- s_4^z
			  - 4 s_1^+ s_2^- - s_1^+ s_3^-
			  + 8 s_1^z s_2^z s_3^z s_4^z
			  + 2 s_1^z s_3^z - 4 s_1^z s_2^z
			    \notag \\[1ex] & \mspace{90.mu}
			  + \de ( 8 s_1^z s_2^+ s_3^- s_4^z + 2 s_1^+ s_2^-
			  - 8 s_1^z s_2^z ) \bigr\>_T \epp
\end{align}
\end{subequations}
The moments are certain combinations of static short-range correlation
functions. This implies, in particular, that they all exist in the
thermodynamic limit. Substituting (\ref{m13}) in to (\ref{rshiftmom})
and (\ref{widthmom}) we have expressed our measures for the resonance
shift and line width in terms of short-range static (rather than
dynamical) correlation functions.

We close this subject with a number of comments.
\begin{enumerate}
\item
Short range static correlation functions of the XXZ chain can be
calculated exactly at all values of temperatures and magnetic
fields.
\item
Our formulae show that special combinations of short-range static
correlation functions can be measured, at least in principle,
by macroscopic experiments.
\item
In practice such measurements are difficult due to limitations
of the experimental accuracies and limitations to experimental
techniques, or due to the fact, that $\de$ in typical spin
chain materials is too small.
\item
It follows from the existence of the moments that the shape of
the ESR absorption lines cannot be Lorentzian, as is assumed by
many experimentalists and in most of the more conventional
theoretical approaches.
\item
Based on the expressions (\ref{m0}), (\ref{m13}) for the moments in terms of
correlation functions we can prove and specify our claim that $I (\om)$ is
generally positive, even for vanishing magnetic field $h$. Using (\ref{m0}) and
(\ref{m1}) in $I_0 = \p (J m_1 + h m_0)$ and noting that $\< s_1^+ s_2^- \>_T =
2 \< s_1^x s_2^x \>_T$ we obtain
\begin{equation}
     I_0 = \p \< h s_1^z + 2 \de J ( s_1^x s_2^x - s_1^z s_2^z)\>_T \epp
\end{equation}
Here the first term in the brackets is positive if $h \ne 0$ as the magnetization
is positive for positive $h$ and negative for negative $h$. Note, however, that
it becomes extremely small for small temperatures in the massive phase $\de > 0$.
As for the second term, for small enough $h$ and $\de > -1$ the neighbour correlators
are negative, and the $zz$-correlations are weaker than the $xx$-correlations for
negative $\de$ and stronger than the $xx$-correlations for positive $\de$. Hence,
the second term is positive even for vanishing $h$ as long as $\de$ is non-zero.
\item
The resonance shift and line width as determined by (\ref{rshiftmom}) and
(\ref{widthmom}) show a simple scaling behaviour as functions of the exchange
interaction $J$. Namely $\de \om/J$ and $\D \om/J$ depend on $J$ only through the
ratios $T/J$ and $h/J$. This is true for the statistical operator $\r$, hence also
for all spin correlation functions, and then by our formulae (\ref{m0}), (\ref{m13})
for $\de \om/J$ and $\D \om/J$ as well. Note also that $\de \om/J$ and $\D \om/J$
both vanish proportional to $\de$ as we approach the isotropic point $\de = 0$.
\end{enumerate}

\subsection{Exercise 21: Shifted moments}
Obtain $m_1$ and $m_2$ from (\ref{allmoments}).

\clearpage

\bibliographystyle{amsplain}
\bibliography{hub}

\end{document}